\documentclass[a4paper,11pt]{amsart} 
\usepackage{amsmath,amssymb}
\usepackage[mathscr]{eucal}
\usepackage{enumerate}
\usepackage{mathtools}
\usepackage{pict2e}
\usepackage[a4paper]{geometry}
\usepackage{hyperref}
\theoremstyle{plain}
\newtheorem{thm}{Theorem}[section]
\newtheorem{coro}[thm]{Corollary}
\newtheorem{lem}[thm]{Lemma}
\newtheorem*{stoll*}{Stollmann's Lemma}
\newtheorem*{main*}{Main Result}
\theoremstyle{definition}

\newtheorem{defa}[thm]{Definition}
\newtheorem{rem}[thm]{Remark}
\newtheorem*{rem*}{Remark}
\makeatletter
\@addtoreset{equation}{section}

\newcommand{\dd}{\mathrm{d}}

\newcommand{\LL}{\mathbb{L}}
\newcommand{\KK}{\mathbb{K}}

\renewcommand{\Re}{\operatorname{Re}}
\newcommand{\dec}{\textup{\,dec}}
\newcommand{\out}{\textup{out}}
\newcommand{\N}{\mathbb{N}}
\newcommand{\Z}{\mathbb{Z}}
\newcommand{\Q}{\mathbb{Q}}
\newcommand{\R}{\mathbb{R}}
\newcommand{\C}{\mathbb{C}}

\newcommand{\D}{\mathbb{D}}
\DeclareMathOperator{\dist}{dist}

\DeclareMathOperator{\expect}{\mathbb{E}}
\DeclareMathOperator{\inter}{int}

\DeclareMathOperator{\one}{\mathbf{1}}
\DeclareMathOperator{\prob}{\mathbb{P}}

\DeclareMathOperator{\supp}{supp}
\DeclareMathOperator{\vol}{vol}
\DeclareMathOperator{\tr}{tr}
\title[Localization for a Multi-Particle Quantum Graph]{Anderson Localization\\ for a Multi-Particle Quantum Graph}
\author{Mostafa Sabri}
\address{Institut de Math\'{e}matiques de Jussieu, Universit\'{e} Paris 7, B\^atiment Sophie Germain, 75013, Paris, France}
\email{sabri@math.jussieu.fr}
\date{July 18, 2013}
\begin{document}
\begin{abstract}
We study a multi-particle quantum graph with random potential. Taking the approach of multiscale analysis we prove exponential and strong dynamical localization of any order in the Hilbert-Schmidt norm near the spectral edge. Apart from the results on multi-particle systems, we also prove Lifshitz-type asymptotics for single-particle systems. This shows in particular that localization for single-particle quantum graphs holds under a weaker assumption on the random potential than previously known.
\end{abstract}
\maketitle
\setcounter{tocdepth}{1}
\tableofcontents
\section{Introduction}    \label{sec:intro}
Quantum graphs are one-dimensional geometric structures composed of edges and vertices and equipped with a Schr\"odinger operator. They arise naturally when one tries to understand the propagation of waves through wires, and their mathematical study goes back at least to the early 1980s, see \cite{PB}, \cite{Ku} for a review. The phenomenon which interests us is known as Anderson localization, and predicts that impurities may suppress the diffusion of a wave, depending on its energy. To verify this for wires and other quasi-one-dimensional materials, one may interpret the impurities and defects in the medium as sources of randomness in the quantum graph. For models with a $\Z^d$ structure, localization has been established for a random potential model in \cite{EHS}, for a random vertex coupling model in \cite{KP1}, and for a random edge length model in \cite{KP2}. Related questions were considered in \cite{ASW} and \cite{HP} for quantum tree graphs. For random potential models with general geometries and vertex couplings, we record the recent result of \cite{Schu}.

In this article we study localization for a multi-particle Hamiltonian on a quantum graph. This is in contrast to the above results, which were concerned with single-particle systems. To study the interaction between two particles, one lying on an edge $e_1$ and the other lying on an edge $e_2$, we have to consider the edge pair $(e_1,e_2)$. We are thus led to consider a certain cartesian product of two quantum graphs, which may be regarded as a two-dimensional network of rectangles. More generally, the configuration space of an $N$-particle system is given by an $N$-dimensional cubical network.

To state our main result, let us briefly describe our model; the elaborate constructions are given in Section \ref{sec:model}. Consider the graph $(\mathcal{E},\mathcal{V})$ with vertex set $\mathcal{V} = \Z^d$ and edge set $\mathcal{E}$ consisting of all line segments of length 1 between two neighbouring vertices. This graph is naturally embedded in $\R^d$, and we denote by $\Gamma^{(1)} \subset \R^d$ the image of the embedding. Define $\Gamma^{(N)}:= \Gamma^{(1)} \times \ldots \times \Gamma^{(1)} \subset (\R^d)^N = \R^{Nd}$ and regard $\Gamma^{(N)}$ as a couple $(\mathcal{K},\mathcal{S})$, where $\mathcal{K}$ is a collection of $N$-dimensional unit cubes $\kappa$, and $\mathcal{S}$ is the collection of the boundaries $\sigma$ of $\kappa$. Each $\sigma$ is a closed union of $2N$ ``open faces'' $\sigma^i$, i.e. $\sigma = \cup_i \bar{\sigma}^i$. Now fix $q_-, q_+ \in \R$, $q_- < q_+$, let $\mu$ be a probability measure on $\R$ with support $[q_-,q_+]$, and consider the probability space $(\Omega, \prob)$ given by $\Omega := [q_-, q_+]^{\mathcal{E}}$ and $\prob:= \mathop \otimes _{e \in \mathcal{E}} \mu$, and the Hilbert space $\mathcal{H} := \mathop \oplus _{\kappa \in \mathcal{K}} L^2\big([0,1]^N\big)$. Then given $\omega = (\omega_e) \in \Omega$, we define the form
\[
\mathfrak{h}_{\omega}^{(N)}[f,g] = \sum_{\kappa \in \mathcal{K}} \big[\langle \nabla f_{\kappa}, \nabla g_{\kappa} \rangle + \langle V_{\kappa}^{\omega} f_{\kappa},g_{\kappa} \rangle \big],
\]
\[
D(\mathfrak{h}_{\omega}^{(N)}) = \left\{f = (f_{\kappa}) \in \mathop \oplus \limits_{\kappa \in \mathcal{K}} {W^{1,2}}\big((0,1)^N\big) \left| \begin{array} {l} f \text{ is continuous on each } \sigma^i, \\ \quad \sum_{\kappa \in \mathcal{K}} \|f_{\kappa}\|_{W^{1,2}}^2 < \infty 
 \end{array} \right. \right\}.
\]
Here $V_{\kappa}^{\omega} := U_{\kappa} + W_{\kappa}^{\omega}$, where $U_{\kappa}$ is a non-random interaction potential. We assume $U_{\kappa}$ is non-negative, bounded, and has a finite range $r_0$ (see Section~\ref{sec:model}). If $\kappa \equiv (e_1,\ldots, e_N)$,  then $W_{\kappa}^{\omega} := \omega_{e_1} + \ldots + \omega_{e_N}$ is an $N$-particle random potential. By continuity on $\sigma^i$, we mean that if $\sigma^i$ is a common face to $\kappa_1$ and $\kappa_2$, then $\left. {f_{\kappa_1}}\right|_{\sigma^i} = \left. {f_{\kappa_2}}\right|_{\sigma^i}$ in the trace sense.

We prove in Theorem~\ref{thm:dom} that there exists a unique self-adjoint operator $H^{(N)}(\omega)$ corresponding to $\mathfrak{h}_{\omega}^{(N)}$. We also prove in Theorem~\ref{thm:as} that
\begin{equation}
[Nq_-,Nq_+] \subset \sigma(H^{(N)}(\omega)) \subseteq [Nq_-,+\infty) \qquad \prob\text{-a.s.}      \label{eq:spec}
\end{equation}

Given $\mathbf{x} \in \Z^{Nd}$ put $\mathbf{C}(\mathbf{x}):= \{ \mathbf{y} \in \R^{Nd} : | \mathbf{y} - \mathbf{x} | <1 \}$, where $|\mathbf{z}| := \| \mathbf{z} \|_{\infty}$ and let $\chi_{\mathbf{x}} := \chi_{\Gamma^{(N)} \cap \, \mathbf{C}(\mathbf{x}) }$. We say that $\psi \in \mathcal{H}$ \emph{decays exponentially with mass $m>0$} if
\[ 
\limsup_{| \mathbf{x} | \to \infty} \frac{ \log \| \chi_{\mathbf{x}} \psi \|}{| \mathbf{x} |} \le -m \, , \qquad \text{where } \log 0 := -\infty \, . 
\]
We may now state our main results. In both theorems we assume $\mu$ is H\"older continuous.

\begin{thm}      \label{thm:exp}
There exist $\varepsilon_0 = \varepsilon_0(N,d,q_-,r_0)> 0$ and $m>0$ such that for a.e. $\omega$ the spectrum of $H^{(N)}(\omega)$ in $I = [Nq_- , Nq_- + \varepsilon_0 ]$ is pure point and the eigenfunctions corresponding to eigenvalues in $I$ decay exponentially with mass $m$.
\end{thm}

\begin{thm}  \label{thm:main}
There exists $\varepsilon_0 = \varepsilon_0(N,d,q_-,r_0)> 0$ such that for $I = [Nq_-,Nq_-+ \varepsilon_0]$, we have for any bounded $K \subset \Gamma^{(N)}$ and all $s>0$,
\[ 
\expect \Big\{ \sup_{\|f\| \le 1} \| X^{s/2} f(H^{(N)}(\omega)) E_{\omega}(I) \chi_K \|_2^2 \Big\} < \infty,
\]
where $(X \psi)(\mathbf{x}) := |\mathbf{x}| \cdot \psi(\mathbf{x})$ for $\psi \in \mathcal{H}$, $E_{\omega}$ is the spectral projection of $H^{(N)}(\omega)$ and the supremum is taken over bounded Borel functions, $\|f\| := \|f\|_{\infty}$.
\end{thm}

In view of (\ref{eq:spec}), our results simply state that $H^{(N)}(\omega)$ exhibits exponential and strong dynamical localization of any order near its spectral bottom, in the Hilbert-Schmidt norm. Theorems~\ref{thm:exp} and \ref{thm:main} are proved in Sections~\ref{sec:exp} and \ref{sec:dyn}, respectively, using the multi-particle multiscale analysis introduced by Chulaevsky and Suhov in \cite{CS2}, and adapted to the continuum by A. Boutet de Monvel \emph{et al.} in \cite{BCS}. The now traditional single-particle multiscale analysis was introduced by Fr\"ohlich and Spencer in \cite{FS}.

Let us note that for $N=1$ our theorems improve the main result of \cite{EHS}, first by removing the technical assumption ($\exists \tau > \frac{d}{2} : \mu([q_-,q_-+h]) \le h^{\tau}$ for small $h$), next by establishing strong dynamical localization in the HS norm. We are able to remove the assumption on $\mu$ by proving Lifshitz-type asymptotics for this model in Section~\ref{sec:ILS}, which to the best of our knowledge, were not proven in earlier papers.

\section{Multi-particle Quantum Graphs}    \label{sec:model}

\subsection{1-Graphs}      \label{sec:1-graphs}
Our building block is the quantum graph $(\mathcal{E},\mathcal{V})$ of \cite{EHS}, given by the vertex set $\mathcal{V} = \Z^d$ and the edge set $\mathcal{E}$ of all line segments of length $1$ between two neighbouring vertices. This graph is naturally embedded in $\R^d$ and we denote by $\Gamma^{(1)} \subset \R^d$ the image of the embedding. To describe $\Gamma^{(1)}$ explicitly, let $(h_j)_{j=1}^d$ be the standard basis of $\Z^d$. Then 
\[
\Gamma^{(1)} := \{ x \in \R^d: x=m + t h_j \text{ for some } m \in \Z^d, j\in \{1,\ldots,d\} \text{ and } t \in [0,1]\} \, .
\]

We denote the edge between $m$ and $m+h_j$ by $e=(m,j)$. Such an edge is identified with the interval $[0,1]$ by sending $x=m+t h_j \in e$ to the point $t$. The Lebesgue measure on $[0,1]$ then induces a natural measure on $\Gamma^{(1)}$ which we denote by $m^{(1)}$.

A function $f$ on $\Gamma^{(1)}$ induces a sequence $(f_e)$, $f_e:(0,1) \to \C$ by setting $f(x) =: f_{(m,j)}(t)$ when $x=m+t h_j$, for some $m \in \Z^d$ and $t \in (0,1)$. As equality in $L^2$ is a.e., this in turn identifies $L^2\big(\Gamma^{(1)},\dd m^{(1)}\big)$ with $\mathop \oplus _{e\in \mathcal{E}} L^2(0,1)$.

Now fix $q_-,q_+ \in \R$, $q_-<q_+$ and let $\mu$ be a probability measure on $\R$ with support $[q_-,q_+]$. Consider the Hilbert space $\mathcal{H} := \mathop \oplus _{e\in \mathcal{E}} L^2(0,1)$, the probability space $(\Omega,\prob)$, where $\Omega :=[q_-,q_+]^{\mathcal{E}}$ and $\prob = \mathop \otimes _{e\in \mathcal{E}} \mu$, and given $\omega=(\omega_e) \in \Omega$, define the form
\[ 
\mathfrak{h}^{(1)}_{\omega}[f,g] = \sum_{e\in \mathcal{E}} \big[\langle f_e',g_e'\rangle + \langle \omega_ef_e,g_e \rangle \big], \qquad D(\mathfrak{h}^{(1)}_{\omega}) = W^{1,2}(\Gamma^{(1)}),
\]
where
\[
W^{1,2}(\Gamma^{(1)}) := \left\{f\in \mathop \oplus \limits_{e \in \mathcal{E}} W^{1,2}(0,1)\left| \begin{array} {l} f \text{ is continuous at each } v\in \mathcal{V}, \\ 
 \qquad \sum_{e\in \mathcal{E}} \|f_e\|_{W^{1,2}}^2 < \infty  \\ 
 \end{array} \right. \right\}.
\]

This form corresponds to the self-adjoint operator $H^{(1)}(\omega):(f_e)\mapsto(-f_e''+\omega_ef_e)$ with Kirchhoff boundary conditions (i.e. if $f \in D(H^{(1)}(\omega))$ and $v \in \mathcal{V}$, then $f$ is continuous at $v$ and satisfies $\sum_{j=1}^d f_{(v,j)}'(0) - \sum_{j=1}^d f_{(v-h_j,j)}'(1)=0$). It is shown in \cite{EHS} that $H^{(1)}(\omega)$ has an almost sure spectrum $\Sigma = [q_-,+\infty)$ and that localization holds near $q_-$.

\subsection{\emph{n}\:-Graphs}           \label{sec:n-gra}
Let us emphasize that throughout this article, the number of particles
\[
N \text{ is fixed. }
\]
We will need to consider Hamiltonians $H^{(n)}(\omega)$ for $1 \le n \le N$ because we will later deduce some spectral properties of $H^{(N)}(\omega)$ from those of $H^{(n)}(\omega)$.

So let us fix $1 \le n \le N$ and consider $n$-particle systems. Formally, quantum mechanics tells us that the Hilbert space corresponding to $n$ distinguishable particles, each living in $\Gamma^{(1)}$, is the tensor product $L^2(\Gamma^{(1)},\dd m^{(1)}) \otimes \ldots \otimes L^2(\Gamma^{(1)},\dd m^{(1)})$. Taking
\[
\Gamma^{(n)} := \Gamma^{(1)} \times \ldots \times \Gamma^{(1)} \quad \text{and} \quad m^{(n)} := m^{(1)} \otimes \ldots \otimes m^{(1)} \, ,
\]
this space may be identified with $L^2\big(\Gamma^{(n)},\dd m^{(n)}\big)$.

If $(h_j)_{j=1}^d$ is the canonical basis of $\Z^d$, then each point $\mathbf{x}=(x_1,\ldots,x_n) \in \Gamma^{(n)}$ takes the form $x_k = m_k + t^k h_{j_k}$ for some $m_k \in \Z^d$, $t^k \in [0,1]$ and $j_k \in \{1,\ldots,d\}$. By varying $t^k$ from $0$ to $1$, we thus obtain a cube $\kappa$ which may be identified with $[0,1]^n$ by sending such an $\mathbf{x}$ to $(t^1,\ldots,t^n)$. Thus, we may regard $\Gamma^{(n)}$ as a couple $(\mathcal{K},\mathcal{S})$,
where $\mathcal{K}$ is a collection of $n$-dimensional cubes $\kappa$
and $\mathcal{S}$ is the collection of the boundaries $\sigma$ of $\kappa$.

For $d=1$, $\Gamma^{(2)} = \R^2$. If we regard it as a couple $(\mathcal{K},\mathcal{S})$, then it consists of unit squares covering $\R^2$ and cornered in $\Z^2$. For $d=2$, let $x,y,z,t$ be the coordinate axes of $\R^4$. Then $\Gamma^{(2)}$ lives in the planes $xz$, $xt$, $yz$ and $yt$, and all their $\Z^4$-translates, and consists of unit squares cornered in $\Z^4$. Squares in the planes $xy$ and $zt$ (and their $\Z^4$-translates) are not allowed. More generally, $\Gamma^{(2)}$ lives in the translates of $d^2$ planes in $\R^{2d}$ and each affine plane is an infinite collection of $\kappa$.

For $n=3$, the only case that can be visualized is that of $d=1$, in which case $\Gamma^{(3)} = \R^3$, and is regarded as the set of all cubes of unit volume cornered in the lattice $\Z^3$.

If the points of $\kappa$ take the form $(x_1,\ldots,x_n)$ with $x_k = m_k + t^k h_{j_k}$ for some $m_k \in \Z^d$, $t_k \in [0,1]$ and $j_k \in \{1,\ldots,d\}$, we will denote $\kappa = \big((m_1,j_1),\ldots,(m_n,j_n)\big)$. Hence, any $\kappa \in \mathcal{K}$ may be written as $\kappa = (e_1,\ldots,e_n)$ for some $e_j \in \mathcal{E}$.

A function $f$ on $\Gamma^{(n)}$ induces a sequence $(f_{\kappa})$, $f_{\kappa} : (0,1)^n \to \C$ by setting $f(\mathbf{x}) =: f_{((m_1,j_1),\ldots,(m_n,j_n))}(t^1,\ldots,t^n)$ when $x_k=m_k+t^k h_{j_k}$, for some $m_k \in \Z^d$ and $t^k \in (0,1)$. As equality in $L^p$ is a.e., this in turn identifies $L^p\big(\Gamma^{(n)},\dd m^{(n)}\big)$ with $\mathop \oplus _{\kappa \in \mathcal{K}} L^p(0,1)^n$ for $1 \le p < \infty$, where $\|(f_{\kappa}) \|^p_{L^p} := \sum_{\kappa \in \mathcal{K}} \|f_{\kappa}\|_{L^p(0,1)^n}^p$.

Each $\sigma$ is the closed union of $2n$ ``open faces'' $\sigma^i$
which may be identified with $(0,1)^{n-1}$.

Given $\mathbf{x}=(x_1,\ldots,x_n) \in \Gamma^{(n)} \subset \R^{nd}$ and a partition $\{1, \ldots, n\} = \mathcal{J} \cup \mathcal{J}^c$, we put $x_{\mathcal{J}} := (x_j)_{j \in \mathcal{J}}$,
$x_{\mathcal{J}^c} := (x_j)_{j \in \mathcal{J}^c}$ and define
\[ 
\dist(x_{\mathcal{J}},x_{\mathcal{J}^c}) := \min \{ |x_i - x_j| : i \in \mathcal{J}, j \in \mathcal{J}^c \}, \quad \text{where } |y| := \|y\|_{\infty} \text{ for } y \in \R^d \, . 
\]

Now fix $q_-,q_+ \in \R$, $q_-<q_+$, and let $\mu$ be a probability measure on $\R$ with support $[q_-,q_+]$. Consider the probability space $(\Omega,\prob)$ with $\Omega:=[q_-,q_+]^{\mathcal{E}}$, $\prob:= \mathop \otimes _{e \in \mathcal{E}} \mu$, the Hilbert space $\mathcal{H} := \mathop \oplus _{\kappa \in \mathcal{K}} L^2(0,1)^n$, and given $\omega = (\omega_e) \in \Omega$, define the form
\[
\mathfrak{h}^{(n)}_{\omega}[f,g] = \sum_{\kappa \in \mathcal{K}} \big[\langle \nabla f_{\kappa}, \nabla g_{\kappa} \rangle + \langle V_{\kappa}^{\omega} f_{\kappa},g_{\kappa} \rangle \big], \qquad D(\mathfrak{h}^{(n)}_{\omega}) = {W^{1,2}}(\Gamma^{(n)}), 
\]
where 
\[
{W^{1,2}}(\Gamma^{(n)}):= \left\{f\in \mathop \oplus \limits_{\kappa \in \mathcal{K}} {W^{1,2}}((0,1)^n)\left| \begin{array} {l} f \text{ is continuous on each } \sigma^i, \\ 
 \quad \sum_{\kappa \in \mathcal{K}} \|f_{\kappa}\|_{W^{1,2}}^2 < \infty  \\ 
 \end{array} \right. \right\}.
\]

By continuity on $\sigma^i$ we mean that whenever $\sigma^i$ is a common face to $\kappa_1$ and $\kappa_2$, then $\left. {f_{\kappa_1}}\right|_{\sigma^i} = \left. {f_{\kappa_2}}\right|_{\sigma^i}$ in the trace sense. The potential is given by $V_{\kappa}^{\omega} := U_{\kappa}^{(n)} + W_{\kappa}^{\omega}$, where $W_{\kappa}^{\omega}$ is an $n$-particle random potential, $W_{\kappa}^{\omega}:= \omega_{e_1} + \ldots + \omega_{e_n}$ if $\kappa = (e_1, \ldots, e_n)$. The sequence $(U_{\kappa}^{(n)})$ is induced from a non-random interaction potential $U^{(n)}:\Gamma^{(n)} \to \R$ with the following properties:

\begin{enumerate}[(1)]
\item $U^{(n)}$ is bounded and non-negative: there exists $u_0>0$ such that
\[
0 \le U^{(n)}(\mathbf{x}) \le u_0 \qquad \text{for } \mathbf{x} \in \Gamma^{(n)} \, .
\]
\item $U^{(n)}$ has finite range\footnote{This includes the $2$-body interaction potentials $U^{(n)}(\mathbf{x}) = \sum_{1 \le i<j \le n} F(x_i-x_j)$, where $F:\Gamma^{(1)} \to \R$ satisfies $F(y) = 0$ if $|y|\ge r_0$. Indeed, if $\dist(x_{\mathcal{J}}, x_{\mathcal{J}^c}) \ge r_0$, then we will have $F(x_i-x_j) = 0$ whenever $i \in \mathcal{J}$ and $j \in \mathcal{J}^c$, so that $U^{(n)}(\mathbf{x})$ indeed decouples into $U^{(n')}(x_{\mathcal{J}}) + U^{(n'')}(x_{\mathcal{J}^c})$.}: there exists $r_0>0$ such that
\[ 
\dist(x_{\mathcal{J}}, x_{\mathcal{J}^c}) \ge r_0 \implies U^{(n)}(\mathbf{x}) = U^{(n')}(x_{\mathcal{J}}) + U^{(n'')}(x_{\mathcal{J}^c})
\]
for any partition $\{1, \ldots, n\} = \mathcal{J} \cup \mathcal{J}^c$ with $|\mathcal{J}| = n'$ and $|\mathcal{J}^c| = n''$.
\item There is no one-particle potential:
\[
U^{(1)} \equiv 0.
\]
\end{enumerate}

For $n=2$, $U^{(2)}$ is thus function satisfying for $(x_1,x_2) \in \Gamma^{(2)} \subset (\R^d)^2$
\[
0 \le U^{(2)}(x_1,x_2) \le u_0 \quad \text{and} \quad |x_1 - x_2| \ge r_0 \implies U^{(2)}(x_1,x_2) = 0.
\]
Notice that if $|x_i - x_j| \ge r_0$ for all $i \neq j$, then $U^{(n)}(\mathbf{x}) = U^{(1)}(x_1)+\ldots+U^{(1)}(x_n) = 0$. Condition (2) says that more generally, if $x_{\mathcal{J}}$ and $x_{\mathcal{J}^c}$ are far apart, then $U^{(n)}$ decouples as prescribed.

We may assume that $r_0 \in \N$; if this is not the case, we just consider $\lfloor r_0 \rfloor+1$, where $\lfloor x \rfloor$ denotes the integer part of $x \in \R$.

\begin{thm}  \label{thm:dom}
Given $\omega \in \Omega$, $\mathfrak{h}_{\omega}^{(n)}$ is closed, densely defined and bounded from below. The unique self-adjoint operator $H^{(n)}(\omega)$ associated with $\mathfrak{h}_{\omega}^{(n)}$ is given by
\[ 
H^{(n)}(\omega) : (f_{\kappa}) \mapsto (- \Delta f_{\kappa} + V_{\kappa}^{\omega} f_{\kappa}), \qquad \text{for } (f_{\kappa}) \in D(H^{(n)}(\omega)).
\]
\end{thm}
\begin{proof}
See the Appendix (Section~\ref{sec:app}).
\end{proof}

We did not provide the explicit domain of $H^{(n)}(\omega)$ as it is not needed in the sequel. It is a subtle question to know exactly how regular the functions $(f_{\kappa}) \in D(H^{(n)})$ are; in particular, it is not clear if the normal derivatives of $f_{\kappa}$ have a trace on $\sigma^i$. For $n=1$, it is easy to see that if $(f_e) \in D(H^{(1)})$, then $f_e \in W^{2,2}(0,1)$ for each $e$. This gives a meaning in particular to the Kirchhoff conditions. Once $n \ge 2$ however, corner singularities appear which, in general, destroy the regularity of the $f_{\kappa}$, see e.g. \cite{Gri}. If we had asked each $f_{\kappa}$ to satisfy Dirichlet or Neumann conditions, we would have $f_{\kappa} \in W^{2,2}((0,1)^n)$ (see \cite[Section 3.2]{Gri}). However, as we ask $f_{\kappa}$ to be continuous on $\sigma^i$, this regularity result is no longer clear. See \cite[Section 2.3.2]{Nic2} for some results when $n=2$ and \cite{BK} for some boundary conditions ensuring regularity also when $n=2$. For general $n$-dimensional polyhedral interface problems, we record the result of \cite{BMNL}.

The following theorem identifies the lower part of $\sigma(H^{(n)}(\omega))$.

\begin{thm}      \label{thm:as}
There exists $\Omega_0 \subseteq \Omega$ with $\prob(\Omega_0) = 1$ such that for all $\omega \in \Omega_0:$
\[ 
[nq_-,nq_+] \subset \sigma(H^{(n)}(\omega)) \subseteq [nq_-,+\infty). 
\]
In particular, $\inf \sigma(H^{(n)}(\omega)) = nq_-$ almost surely.
\end{thm}
\begin{proof}
See the Appendix (Section~\ref{sec:app}).
\end{proof}

\section{Finite-volume operators and geometry of cubes}    \label{sec:MSA}
\subsection{Finite-volumes operators}
Fix $1 \le n \le N$. Throughout the paper we use the sup norm of $\R^{nd}$ :
\[ 
|x| := \|x\|_{\infty}, \qquad | \mathbf{x}| := \|\mathbf{x}\|_{\infty}
\]
for $x \in \R^d$ and $\mathbf{x} \in \R^{nd}$. Given $L \in \N^{\ast}$, we define $1$-cubes with center $u \in \Z^d$ by
\[ 
\Lambda_L^{(1)}(u) = \{x \in \R^d : |x-u| < L \}, \qquad |\Lambda_L^{(1)}(u)| = (2L)^d \, .
\]
Given $\mathbf{u} = (u_1,\ldots, u_n) \in \Z^{nd}$ and $\LL=(L_1,\ldots,L_n) \in \N^n$ with $L_j \ge 1$, we define $n$-rectangles and $n$-cubes by
\[ 
\Lambda_{\LL}^{(n)}(\mathbf{u}) = \prod_{j=1}^n \Lambda_{L_j}^{(1)}(u_j), \qquad \Lambda_L^{(n)}(\mathbf{u}) = \Lambda_{(L,\ldots,L)}^{(n)}(\mathbf{u}) = \prod_{j=1}^n \Lambda_L^{(1)}(u_j) \, .
\]

Note that a cube is always open. We take $\mathbf{u} \in \Z^{nd}$ and $\LL \in \N^n$ above to ensure that the closure of $\Gamma^{(n)} \cap \Lambda^{(n)}_{\LL}(\mathbf{u})$ is a subgraph of $\Gamma^{(n)}$. Abusing notation, we also denote this closure by $\Gamma^{(n)} \cap \Lambda^{(n)}_{\LL}(\mathbf{u})$. For $1$-graphs, taking the closure means that we add the vertices lying on $\partial \Lambda_L^{(1)}(u)$ that belong to inner edges. This should not be confused with the larger subgraph $\Gamma^{(1)} \cap \overline{\Lambda}_L^{(1)}(u)$.

\begin{lem}      \label{lem:NB}
The following estimates hold:
\begin{align*}
\#\{\mathcal{E}(\Gamma^{(1)} \cap \Lambda_L^{(1)}) \} & = d(2L)(2L-1)^{d-1} \le d \cdot |\Lambda_L^{(1)}| \, , \tag{\textsf{NB.$1$}} \\
\#\{\mathcal{K}(\Gamma^{(n)} \cap \Lambda_{\LL}^{(n)}) \} & = \prod_{j=1}^n\Big(d(2L_j)(2L_j-1)^{d-1}\Big) \le d^n \cdot |\Lambda_{\LL}^{(n)}| \, . \tag{\textsf{NB.$n$}}
\end{align*}
\end{lem}

\begin{proof}
See the Appendix (Section~\ref{sec:app}).
\end{proof}

We define the \emph{discrete cubes} $\mathbf{B}^{(n)}_L(\mathbf{u})$ and the \emph{cells} $\mathbf{C}(\mathbf{u})$ by
\[ 
\mathbf{B}^{(n)}_L(\mathbf{u}) = \Lambda_L^{(n)}(\mathbf{u}) \cap \Z^{nd}, \qquad \mathbf{C}(\mathbf{u}) = \Lambda_1^{(n)}(\mathbf{u}) \subset \R^{nd} \, .
\]

A finite union of cells will be called a \emph{cellular set}. For $L \ge 7$, we denote
\[ 
\Lambda_L^{\out}(\mathbf{u}) = \Lambda_L^{(n)}(\mathbf{u}) \setminus \Lambda_{L-6}^{(n)}(\mathbf{u}), \qquad \mathbf{B}_L^{\out}(\mathbf{u}) = \Lambda_L^{\out}(\mathbf{u}) \cap \Z^{nd} \, .
\]

We define the restriction of $H^{(n)}(\omega)$ to a rectangle $\Lambda = \Lambda_{\LL}^{(n)}$ to be the operator $H_{\Lambda}^{(n)}(\omega)$
associated with the form
\[ 
\mathfrak{h}_{\omega,\Lambda}^{(n)}[f,g] = \sum_{\kappa \in \mathcal{K}(\Gamma \cap \Lambda)} \big[ \langle \nabla f_{\kappa} , \nabla g_{\kappa} \rangle + \langle V_{\kappa}^{\omega} f_{\kappa} , g_{\kappa} \rangle \big], \qquad D(\mathfrak{h}_{\omega,\Lambda}^{(n)}) = W^{1,2}(\Gamma^{(n)} \cap \Lambda) \, ,
\]
where $W^{1,2}(\Gamma^{(n)} \cap \Lambda)$ is the set of $f \in \mathop \oplus_{\kappa \in \mathcal{K}(\Gamma \cap \Lambda)} W^{1,2}((0,1)^n)$ which are continuous on inner $\sigma^i$. For $n=1$, the functions $f \in D(H_{\Lambda}^{(1)}(\omega))$ satisfy Kirchhoff conditions at each vertex in $\Lambda$. Note that for boundary vertices, Kirchhoff conditions are just Neumann conditions.

\begin{lem}      \label{lem:WEYL}
$H_{\Lambda_{\LL}}^{(n)}(\omega)$ has a compact resolvent. Its discrete set of eigenvalues denoted by $E_j(H_{\Lambda_{\LL}}^{(n)}(\omega))$ counting multiplicity satisfies the following Weyl law:
\[ 
\forall S \in \R \ \exists C = C(n,d,S-nq_-):\ j > C |\Lambda_{\LL}^{(n)}| \implies E_j(H_{\Lambda_{\LL}}^{(n)}(\omega)) > S \, .  \tag{\textsf{WEYL.$n$}} \]
Moreover, $C$ is independent of $\omega$, and if $S>S^{\ast}(n,q_-)$, then $C \le \lfloor \frac{d^n(S-nq_-)^{n/2}}{(4\pi)^{n/2}\Gamma(n/2)} \rfloor +1$.
\end{lem}
\begin{proof}
See the Appendix (Section~\ref{sec:app}).
\end{proof}

In the rest of this paper, for a bounded volume $K \subset \R^{nd}$ we put
\[ 
\chi_K := \chi_{\Gamma^{(n)} \, \cap \, K}, \qquad \chi_{\mathbf{x}} := \chi_{\Gamma^{(n)} \, \cap \, \mathbf{C}(\mathbf{x})} \, .
\]
Given $\omega \in \Omega$, $E \notin \sigma(H_{\Lambda}^{(n)}(\omega))$ and $\mathbf{i},\mathbf{j} \in \Lambda^{(n)} \cap \Z^{nd}$, we define
\[ 
G_{\Lambda^{(n)}}(E):= (H_{\Lambda}^{(n)}(\omega) -E)^{-1}, \qquad G_{\Lambda^{(n)}}(\mathbf{i},\mathbf{j};E) := \chi_{\mathbf{i}} G_{\Lambda^{(n)}}(E) \chi_{\mathbf{j}} \, .
\]

\subsection{Geometry of cubes}
\begin{defa}
Given $n \ge 2$ and a partition $\{1,\ldots,n\} = \mathcal{J} \cup \mathcal{J}^c$, we say that $\Lambda_L^{(n)}(\mathbf{u})$ is $\mathcal{J}$-\emph{decomposable} if
\[ 
\dist(u_{\mathcal{J}},u_{\mathcal{J}^c}) \ge 2L+r_0 \, .
\]
We say that $\Lambda_L^{(n)}(\mathbf{u})$ is \emph{decomposable} if there exists a partition $\{1,\ldots,n\} = \mathcal{J} \cup \mathcal{J}^c$ such that $\Lambda_L^{(n)}(\mathbf{u})$ is $\mathcal{J}$-decomposable.
\end{defa}

A $\mathcal{J}$-decomposable cube $\Lambda_L^{(n)}(\mathbf{u})$ will henceforth be denoted by
\[
\Lambda_L^{(n)}(\mathbf{u}) = \Lambda_L^{(n')}(u_{\mathcal{J}}) \times \Lambda_L^{(n'')}(u_{\mathcal{J}^c}), \qquad \text{where } n' = |\mathcal{J}| \text{ and } n'' = |\mathcal{J}^c| \, .
\]

\begin{rem}               \label{rem:dec}
Suppose $\Lambda_L^{(n)}(\mathbf{u})$ is $\mathcal{J}$-decomposable and identify $L^2(\Gamma^{(n)} \cap \Lambda_L^{(n)}(\mathbf{u})) \equiv L^2(\Gamma^{(n')} \cap \Lambda_L^{(n')}(u_{\mathcal{J}})) \otimes L^2(\Gamma^{(n'')} \cap \Lambda_L^{(n'')}(u_{\mathcal{J}^c}))$. Any $\mathbf{x} \in \Lambda_L^{(n)}(\mathbf{u})$ satisfies $\dist(x_{\mathcal{J}},x_{\mathcal{J}^c}) > r_0$, hence $U^{(n)}(\mathbf{x}) = U^{(n')}(x_{\mathcal{J}}) + U^{(n'')}(x_{\mathcal{J}^c})$. Consequently, $H_{\Lambda_L(\mathbf{u})}^{(n)} = H_{\Lambda_L(u_{\mathcal{J}})}^{(n')} \otimes I + I \otimes H_{\Lambda_L(u_{\mathcal{J}^c})}^{(n'')}$. If now $\{(\varphi_a,\lambda_a)\}_a$ and $\{(\psi_b,\mu_b)\}_b$ are orthonormal bases of eigenfunctions of $H^{(n')}_{\Lambda_L(u_{\mathcal{J}})}$ and $H^{(n'')}_{\Lambda_L(u_{\mathcal{J}^c})}$, respectively, then $\Psi_{a,b}(\mathbf{x}) := \varphi_a(x_{\mathcal{J}}) \otimes \psi_b(x_{\mathcal{J}^c})$ form an orthonormal basis of eigenfunctions for $H_{\Lambda_L(\mathbf{u})}^{(n)}$ with corresponding eigenvalues $E_{a,b} = \lambda_a + \mu_b$. Since 
\[
P_{a,b} := \langle \cdot , \Psi_{a,b} \rangle \Psi_{a,b} = \langle \cdot , \varphi_a \otimes \psi_b \rangle \varphi_a \otimes \psi_b = \big(\langle \cdot , \varphi_a \rangle \varphi_a\big)\otimes\big(\langle \cdot , \psi_b \rangle \psi_b\big) =: P_a \otimes P_b \, ,
\]
by the functional calculus, we get for any Borel function $\eta:\sigma(H^{(n)}_{\Lambda_L(\textbf{u})}) \to \C$
\begin{equation}
\eta (H_{\Lambda_L(\mathbf{u})}^{(n)}) = \sum_{a,b} \eta(E_{a,b}) P_{a,b} = \sum_a P_a \otimes \Big( \sum_b \eta(E_{a,b}) P_b \Big) \, .  \label{eq:dec}
\end{equation}
\end{rem}

\begin{defa}
Let $\D := \{ \mathbf{x} = (x,\ldots,x) : x \in \Z^d \} \subset \Z^{nd}$. A cube $\Lambda_L^{(n)}(\mathbf{u})$ is \emph{partially interactive} (PI) if $\dist(\mathbf{u},\D) \ge (n-1)(2L+r_0)$, and \emph{fully interactive} (FI) otherwise.
\end{defa}

\begin{lem}      \label{lem:PI}
A partially interactive cube is decomposable.
\end{lem}
\begin{proof} 
See the Appendix (Section~\ref{sec:app}).
\end{proof}

For $n \ge 1$, $j=1,\ldots, n$, we define projections of $n$-rectangles on $\R^d$ by
\[ 
\Pi_j \Lambda_{\LL}^{(n)}(\mathbf{u}) = \Lambda_{L_j}^{(1)}(u_j), \qquad \Pi \Lambda_{\LL}^{(n)}(\mathbf{u}) = \bigcup_{j=1}^n \Lambda_{L_j}^{(1)}(u_j) \, .
\]
We define $\Pi_{\emptyset} \Lambda_{\LL}^{(n)}(\mathbf{u}) := \emptyset$ and put for $\emptyset \neq \mathcal{J} \subseteq \{1,\ldots, n\}$,
\[ 
\Pi_{\mathcal{J}} \Lambda_{\LL}^{(n)}(\mathbf{u}) = \bigcup_{j \in \mathcal{J}} \Pi_j \Lambda_{\LL}^{(n)}(\mathbf{u}) = \bigcup_{j \in \mathcal{J}} \Lambda_{L_j}^{(1)}(u_j) \, .
\]
\begin{defa}              \label{def:sep}
We say $\Lambda_{\LL}^{(n)}(\mathbf{u})$ is $\mathcal{J}$-pre-separable from $\Lambda_{\KK}^{(n)}(\mathbf{v})$ if
\[ 
\Pi_{\mathcal{J}} \Lambda_{\LL}^{(n)}(\mathbf{u}) \cap \big( \Pi_{\mathcal{J}^c} \Lambda_{\LL}^{(n)}(\mathbf{u}) \cup \Pi \Lambda_{\KK}^{(n)}(\mathbf{v}) \big) = \emptyset \, .
\]
$\Lambda_{\LL}^{(n)}(\mathbf{u})$ and $\Lambda_{\KK}^{(n)}(\mathbf{v})$ are said to be \emph{pre-separable} if there exists $\emptyset \neq \mathcal{J} \subseteq \{1,\ldots, n\}$ such that $\Lambda_{\LL}^{(n)}(\mathbf{u})$ is $\mathcal{J}$-pre-separable from $\Lambda_{\KK}^{(n)}(\mathbf{v})$ or $\Lambda_{\KK}^{(n)}(\mathbf{v})$ is $\mathcal{J}$-pre-separable from $\Lambda_{\LL}^{(n)}(\mathbf{u})$. 

Two cubes $\Lambda_L^{(n)}(\mathbf{u})$ and $\Lambda_L^{(n)}(\mathbf{v})$ are said to be \emph{separable} if they are pre-separable and if $|\mathbf{u} - \mathbf{v} | \ge r_{n,L}$, where
\[ 
r_{n,L}:=4(n-1)(2L+r_0)+2L \, .
\]
Finally, they are said to be \emph{completely separated} if they are separable with $\mathcal{J} = \{1,\ldots, n\}$, i.e. if $\Pi \Lambda_L^{(n)}(\mathbf{u}) \cap \Pi \Lambda_L^{(n)}(\mathbf{v}) = \emptyset$ and $|\mathbf{u} - \mathbf{v} | \ge r_{n,L}$.
\end{defa}
Notice that if two cubes are completely separated, the corresponding Hamiltonians $H_{\Lambda_L(\mathbf{u})}^{(n)}$ and $H_{\Lambda_L(\mathbf{v})}^{(n)}$ have independent spectra (because $\Pi \Lambda_L^{(n)}(\mathbf{u}) \cap \Pi \Lambda_L^{(n)}(\mathbf{v}) = \emptyset$).

Let us give some criteria for separability. Given $\mathbf{x} \in \Z^{nd}$, put $k_{\mathbf{x}} := \# \{ x_1, \ldots , x_n \}$. Then each $\mathbf{x} \in \Z^{nd}$ gives rise to $k_{\mathbf{x}}^n$ related points denoted by $\mathbf{x}^{(j)}=(x_1^{(j)},\ldots,x_n^{(j)})$, with $x_k^{(j)} \in \{ x_1, \ldots, x_n \}$ for all $k$. For example, for $d=1$, the point $(1,5)\in \Z^2$ gives rise to $(1,1)$, $(1,5)$, $(5,1)$ and $(5,5)$. Taking
\[ 
K(n) := n^n, 
\]
we have $k_{\mathbf{x}}^n \le K(n)$ and the following lemmas hold.

\begin{lem}      \label{lem:separable}
Let $\mathbf{x}, \mathbf{y} \in \Z^{nd}$, $L \in \N^{\ast}$ and take $r_{n,L}$ as in Definition~\ref{def:sep}. Then
\begin{enumerate}[\rm 1)]
\item If $\mathbf{y} \notin \bigcup_{j=1}^{K(n)} \Lambda_{2nL}^{(n)}(\mathbf{x}^{(j)})$, then $\Lambda_L^{(n)}(\mathbf{y})$ and $\Lambda_L^{(n)}(\mathbf{x})$ are pre-separable.
\item If $\mathbf{y} \notin \bigcup_{j=1}^{K(n)} \Lambda_{r_{n,L}}^{(n)}(\mathbf{x}^{(j)})$, then $\Lambda_L^{(n)}(\mathbf{y})$ and $\Lambda_L^{(n)}(\mathbf{x})$ are separable.
\item If $\mathbf{y} \notin \Lambda_{2r_{n,L}}^{(n)}(\mathbf{0})$, then $\Lambda_L^{(n)}(\mathbf{y})$ is separable from any $\Lambda_L^{(n)}(\mathbf{x})$ satisfying $\mathbf{x} \in \Lambda_{r_{n,L}}^{(n)}(\mathbf{0})$.
\end{enumerate}
\end{lem}
\begin{proof}
See the Appendix (Section~\ref{sec:app}). \qedhere
\end{proof}

\begin{lem}      \label{lem:FI}
Separable FI cubes are completely separated.
\end{lem}
\begin{proof}
See the Appendix (Section~\ref{sec:app}). \qedhere
\end{proof}

\subsection{MSA strategy}      \label{sub:MSA}
We summarize here the multiscale analysis strategy which we follow to prove localization in an interval $I$. Let us start with $1$-particle systems:
\begin{enumerate}[1.]
\item Find $L_0>0$ and $m_0>0$ such that the probability of having one ``good'' cube among any disjoint pair $\Lambda_{L_0}(u), \Lambda_{L_0}(v)$ is high. Here $\Lambda_{L_0}$ is good if for any $E \in I$ and $i,j$ far apart, $\|G_{\Lambda_{L_0}}(i,j;E)\| \le e^{-m_0L_0}$. This is the \emph{initial length scale estimate} (ILS). 
\item Find a sequence of length scales $L_k \nearrow ^{\infty}$ for which a similar decay property holds, with an increasingly good  probability (typically $1-L_k^{-2p}$ for some $p>0$). This is done by induction on $k$ and is the heart of multiscale analysis. 
\item Use this sequence to show that the generalized eigenfunctions of $H^{(1)}$ corresponding to generalized eigenvalues in $I$ exhibit an exponential decay.
\item Deduce exponential localization by proving that generalized eigenfunctions of $H^{(1)}$ exist spectrally almost everywhere.
\item Establish dynamical localization.
\end{enumerate}

For step 1, one shows that if a cube $\Lambda$ is ``bad'', then $\dist(\sigma(H_{\Lambda}^{(1)}),\inf \sigma(H^{(1)}))$ must be very small. This is done ad absurdum using a \emph{Combes-Thomas estimate}. Then one proves this distance cannot be too small using \emph{Lifshitz tails} (one can also prove step 1 without Lifshitz tails in some cases). For step 2, one first relates $G_{\Lambda'}(x,y;E)$ to $G_{\Lambda}(z,y;E)$ for $\Lambda' \supset \Lambda$ to deduce the decay of $G_{\Lambda_{L_k}}(x,y;E)$ from the decay of $G_{\Lambda_{L_{k-1}}}(z,y;E)$. This is done using the \emph{Geometric resolvent inequality}. However, in this inequality the decay term from $G_{\Lambda_{L_{k-1}}}(z,y;E)$ gets multiplied by $\| G_{\Lambda_{L_k}}(x,w;E) \|$. So to make sure the product remains very small, it is necessary to show that $\| G_{\Lambda_{L_k}}(x,w;E) \|$ is not too big. This is done using \emph{Wegner estimates}. The remaining steps will be explained in more detail later.

The main difficulty in adapting the previous scheme to multi-particle systems lies in the fact that Hamiltonians restricted to disjoint cubes are no longer independent. One may think of replacing disjoint cubes by completely separated ones, since the corresponding Hamiltonians will then be independent. Unfortunately this cannot work, as there is no analog of Lemma~\ref{lem:separable} for such cubes (e.g. $[0,1] \times [1,2]$ and $[0,1] \times [r,r+1]$ are not completely separated, no matter how big $r$ is) and consequently no analog of Lemma~\ref{lem:DKn} either. This is why one is forced to work with the larger class of separable cubes. As Hamiltonians restricted to such cubes are not independent, a new strategy must be conceived especially in the induction step; see Section~\ref{sec:stratn}.

\section{Combes-Thomas estimate}    \label{sec:CT}

We prove our Combes-Thomas estimate by deriving good bounds on the Schr\"odinger semigroup. This was done before in \cite{FLM} using the Feynmann-Kac formula and the explicit form of the heat kernel. We shall instead prove our bound via a Davies-Gaffney estimate. This method has several advantages: it does not presuppose a heat kernel estimate, it proves the Combes-Thomas estimate for any energy below the spectral bottom, not just below the infimum of the potential, and the resulting upper bound is easier to control.

Let us mention that the idea of proving Combes-Thomas estimates via semigroups appeared much earlier in \cite[Lemma B.7.11]{Si82}. Compared to our proof and the proof of \cite{FLM}, the method of \cite{Si82} requires much more input, but it has the advantage of being valid for arbitrary energies outside the spectrum.

We start with a technical lemma.

\begin{lem}      \label{lem:sobolev}
Let $\Lambda^{(n)}$ be a cube or $\Lambda^{(n)} = \R^{nd}$. If $u \in W^{1,2}(\Gamma^{(n)} \cap \Lambda^{(n)})$ and $\varphi$ is a bounded Lipschitz continuous function on $\Gamma^{(n)} \cap \Lambda^{(n)}$, then $\varphi u \in W^{1,2}(\Gamma^{(n)} \cap \Lambda^{(n)})$ and $\nabla(\varphi u) = u \nabla \varphi + \varphi \nabla u$. 
\end{lem}
Here $\varphi u := (\varphi_{\kappa} u_{\kappa})$, where $(\varphi_{\kappa})$ is obtained from $\varphi$ as in Section~\ref{sec:n-gra}.
\begin{proof}
By \cite[Proposition 4.1.27]{Sto}, we have $\varphi_{\kappa} u_{\kappa} \in W^{1,2}((0,1)^n)$ and $\nabla(\varphi_{\kappa} u_{\kappa}) = u_{\kappa} \nabla \varphi_{\kappa} + \varphi_{\kappa} \nabla u_{\kappa}$ for all $\kappa$. So it remains to show $\varphi u$ is continuous on inner $\sigma^i$. By the density of $C^{\infty}([0,1]^n)$ in $W^{1,2}((0,1)^n)$ (see \cite[Section 1.1.6]{Ma}) and the continuity of the trace operator $\gamma:W^{1,2}((0,1)^n) \to L^2((0,1)^{n-1})$, we may assume all $u_{\kappa} \in C([0,1]^n)$. Since each $\varphi_{\kappa}$ is bounded and uniformly continuous on $(0,1)^n$, it has a unique bounded continuous extension $\tilde{\varphi}_{\kappa}$ on $[0,1]^n$. Thus, $\tilde{\varphi}_{\kappa} u_{\kappa} \in C([0,1]^n)$ and $\gamma(\varphi_{\kappa} u_{\kappa})$ is just the restriction of $\tilde{\varphi}_{\kappa} u_{\kappa}$ to $\partial \kappa$. Now if $\sigma^i$ is a common face to $\kappa_1$ and $\kappa_2$, the extensions $\tilde{\varphi}_{\kappa_1}$ and $\tilde{\varphi}_{\kappa_2}$ must coincide on $\sigma^i$ since $\varphi$ is Lipschitz continuous. Hence,
\[
\gamma ( \varphi_{\kappa_1} u_{\kappa_1}) = (\tilde{\varphi}_{\kappa_1} u_{\kappa_1})|_{\sigma ^i} = (\tilde{\varphi}_{\kappa_2} u_{\kappa_2})|_{\sigma^i} = \gamma ( \varphi_{\kappa_2} u_{\kappa_2}) \, ,
\]
since $u$ is continuous on $\sigma^i$. Hence $\varphi u$ is continuous on $\sigma^i$.
\end{proof}

In the following $\dist(\cdot,\cdot)$ refers to the distance induced by the sup norm of $\R^{nd}$. 

\begin{lem}[Improved Davies-Gaffney estimate]      \label{lem:DG}
Let $\Lambda^{(n)}$ be a cube or $\Lambda^{(n)} = \R^{nd}$. Let $A_1,A_2 \subset \Lambda^{(n)}$ be cellular sets such that $\dist(A_1,A_2) =: \delta \ge 1$ and suppose $f,g \in L^2(\Gamma^{(n)} \cap \Lambda^{(n)})$, $\supp f \subset A_1$ and $\supp g \subset A_2$. Then if $s_{\omega} := \inf \sigma(H_{\Lambda}^{(n)}(\omega))$, we have
\[ 
\forall t>0 : | \langle e^{-tH_{\Lambda}^{(n)}(\omega)} f , g \rangle | \le e^{-t s_{\omega}} e^{-\frac{\delta^2}{4t}} \|f \| \|g\| \, .
\]
\end{lem}
\begin{proof}
We first assume $\Lambda^{(n)}$ is a cube. Put $H:=H_{\Lambda}^{(n)}(\omega) - s_{\omega}$. Given $\mathbf{x} \in \Lambda^{(n)}$, let $\tilde{w}(\mathbf{x}) := \dist(\mathbf{x},A_1)$. Then $|\tilde{w}(\mathbf{x})-\tilde{w}(\mathbf{y})| \le |\mathbf{x} - \mathbf{y} |$, hence $\|\nabla \tilde{w}\|_{\infty} \le 1$ and $e^{\zeta \tilde{w}(\cdot)}$ is bounded, Lipschitz continuous on $\Lambda^{(n)}$ for $\zeta>0$. Let $w$ be the restriction of $\tilde{w}$ to $\Gamma^{(n)} \cap \Lambda^{(n)}$. Then by Lemma~\ref{lem:sobolev}, if $\mathfrak{h}$ is the form associated to $H$, then $e^{\zeta w}u \in D(\mathfrak{h})$ whenever $u \in D(\mathfrak{h})$. Now given $f \in D(H)$, $t>0$ put $f_t := e^{-t H} f$ and note that $f_t \in D(H)$. Fix $\beta >0$ and as in \cite[Theorem 3.3]{CoSi} consider
\[
E(t) = \langle f_t, f_t e^{\beta w} \rangle = \| f_t e^{\beta w/2} \|^2 .
\]
Then
\[
E'(t) = - 2 \Re \langle H f_t , f_t e^{ \beta w} \rangle = - 2\Re \mathfrak{h}[f_t , f_t e^{\beta w} ]
\]
and thus
\begin{align*}
\frac{E'(t)}{2} & = - \Re \big( \langle \nabla f_t , \nabla (f_t e^{ \beta w}) \rangle + \langle (V^{\omega}-s_{\omega}) f_t , f_t e^{\beta w} \rangle \big) \\
& = - \Re \langle \nabla f_t , \nabla (f_t e^{ \beta w}) \rangle - \langle V^{\omega} f_t, f_t e^{\beta w} \rangle + s_{\omega} \| f_t e^{\beta w/2} \|^2 \, .
\end{align*}
Now by min-max for forms we have
\[
s_{\omega} = \inf_{f \in D(\mathfrak{h}^{(n)}_{\omega,\Lambda}),\|f\|=1} \mathfrak{h}_{\omega,\Lambda}^{(n)}[f,f] \le \| f_t e^{\beta w/2} \|^{-2} \cdot \mathfrak{h}_{\omega,\Lambda}^{(n)}[ f_t e^{\beta w/2}, f_t e^{\beta w/2}] \, ,
\]
where $\mathfrak{h}_{\omega,\Lambda}^{(n)}$ is the form associated to $H_{\Lambda}^{(n)}(\omega)$. Thus,
\begin{align*}
s_{\omega} \| f_t e^{\beta w/2} \|^2 & \le \langle \nabla(f_t e^{\beta w/2}), \nabla(f_t e^{\beta w/2}) \rangle + \langle V^{\omega} f_t e^{\beta w/2}, f_t e^{\beta w/2} \rangle \\
& = \langle (\nabla f_t) e^{\beta w/2}, (\nabla f_t) e^{\beta w/2} \rangle + 2 \Re \langle (\nabla f_t) e^{\beta w/2},f_t(\textstyle{\frac{\beta}{2}} \nabla w) e^{\beta w/2} \rangle \\
& \quad + \langle f_t (\textstyle{\frac{\beta}{2}} \nabla w) e^{\beta w/2}, f_t (\textstyle{\frac{\beta}{2}} \nabla w) e^{\beta w/2} \rangle + \langle V^{\omega} f_t, f_t e^{\beta w} \rangle \\
& = \langle \nabla f_t, (\nabla f_t) e^{\beta w} \rangle + \Re \langle \nabla f_t , f_t (\beta \nabla w) e^{\beta w} \rangle \\
& \quad + \textstyle{\frac{\beta^2}{4}} \| f_t (\nabla w) e^{\beta w/2} \|^2 + \langle V^{\omega} f_t, f_t e^{\beta w} \rangle \\
& = \Re \langle \nabla f_t, \nabla(f_t e^{\beta w}) \rangle + \langle V^{\omega} f_t, f_t e^{\beta w} \rangle + \textstyle{\frac{\beta^2}{4}} \| f_t (\nabla w) e^{\beta w/2} \|^2 \, ,
\end{align*}
where we used Lemma~\ref{lem:sobolev}. We thus have
\[
\frac{E'(t)}{2} \le \frac{\beta^2}{4} \| f_t(\nabla w) e^{\beta w/2} \|^2 \le \frac{\beta^2}{4} \| f_t e^{\beta w/2} \|^2 = \frac{ \beta^2 E(t)}{4} \, .
\]
Hence, $E(t) \le e^{\beta^2 t/2} E(0)$. Moreover,
\[
\| \chi_{A_2} f_t\|^2 \le \| \chi_{A_2} e^{-\beta w/2} \|_{\infty}^2 \|e^{\beta w/2} f_t \|^2 \le e^{-\beta \delta} E(t) \, .
\]
Since $\supp f \subset A_1$ and $w=0$ on $A_1$, we have $E(0) = \|e^{\beta w/2} f\|^2 = \|f\|^2$. Hence,
\[
\| \chi_{A_2} f_t\|^2 \le e^{- \beta \delta} E(t) \le \exp \big( \frac{ \beta^2 t}{2} - \beta \delta \big) E(0) = \exp \big( \frac{ \beta^2 t}{2} - \beta \delta \big) \| f\|^2 \, .
\]
Choose $\beta = \delta / t$. Since $\supp g \subset A_2$ we finally get
\[
| \langle e^{-tH} f , g \rangle |^2 = | \langle \chi_{A_2} f_t , g \rangle |^2 \le \| \chi_{A_2} f_t\|^2 \cdot \|g\|^2 \le e^{-\delta^2/2t} \|f\|^2 \|g\|^2 \, .
\]
The assertion follows (if $\Lambda^{(n)}$ is a cube) by noting that $H$ is densely defined and that
\[
e^{-tH} = \exp( -t(H_{\Lambda}^{(n)}(\omega) - s_{\omega}) ) = e^{ts_{\omega}} e^{-tH_{\Lambda}^{(n)}(\omega)} \, .
\]

Finally, all the arguments remain valid if $\Lambda^{(n)}=\R^{nd}$, except that $e^{\zeta w}$ is no longer bounded. We thus consider a large cube $\Xi$ containing $A_1$ and $A_2$ and replace $\tilde{w}$ by a Lipschitz function $\rho$ of compact support such that $\rho(\mathbf{x}) = \dist(\mathbf{x},A_1)$ if $\mathbf{x} \in \Xi$ and $\| \nabla \rho \|_{\infty} \le 1$, then take $w$ to be the restriction of $\rho$ to $\Gamma^{(n)}$.
\end{proof}

\begin{thm}[Combes-Thomas estimate]     \label{thm:CT2}
Let $\Lambda^{(n)}$ be a cube or $\Lambda^{(n)} = \R^{nd}$ and let $A,B \subset \Lambda^{(n)}$ be cellular sets such that $\dist(A,B) =: \delta \ge 1$. Then for $E< s_{\omega} := \inf \sigma(H_{\Lambda}^{(n)}(\omega))$ and $\eta := s_{\omega} -E$ we have
\[
\| \chi_A (H_{\Lambda}^{(n)}(\omega) - E)^{-1} \chi_B \| \le \sqrt{ \frac{\pi}{2}} \left( \frac{\sqrt{\delta}}{\eta^{3/4}} + \frac{3}{8 \sqrt{\delta} \eta^{5/4}} \right) e^{ - \delta \sqrt{\eta}} \, .
\]
\end{thm}
\begin{proof}
Put $H = H_{\Lambda}^{(n)}(\omega)$. Given $f,g \in L^2(\Gamma^{(n)} \cap \Lambda^{(n)})$ with $\|f\| = \|g\| = 1$ we have
\[ 
|\langle \chi_A e^{-tH} \chi_B f,g \rangle| = |\langle e^{-tH} \chi_B f, \chi_A g \rangle| \le e^{-t s_{\omega}} e^{-\frac{\delta^2}{4t}} \| \chi_B f \| \| \chi_A g\| \le e^{-t s_{\omega}} e^{-\frac{\delta^2}{4t}}
\]
by Lemma~\ref{lem:DG}. Thus
\[
\| \chi_A e^{-tH} \chi_B \| \le e^{-t s_{\omega}} e^{-\frac{\delta^2}{4t}} \, .
\]
Now for $E<s_{\omega}$ we have $(H-E)^{-1} = \int_0^{\infty} e^{tE} e^{-tH} \, \dd t$. Hence
\[
\| \chi_A (H - E)^{-1} \chi_B \| \le \int_0^{\infty} e^{-t \eta} e^{-\frac{\delta^2}{4t}} \, \dd t = \frac{\delta}{\sqrt{\eta}} K_1(\delta \sqrt{\eta})
\]
where $K_1$ is the modified Bessel function and we used \cite[Formula 3.324]{GR} to evaluate the integral. Now by \cite[Formula 9.7.2]{AS} and the remark after it, we have for real $z>0$ the estimate $K_1(z) \le \sqrt{\frac{\pi}{2z}} e^{-z} \Big(1 + \frac{3}{8z}\Big)$. This proves the assertion.
\end{proof}

\section{Geometric Resolvent Inequalities}    \label{sec:GRI}

In this section we follow \cite{Sto} to prove Theorems~\ref{thm:GRE} and \ref{thm:GRI} and use arguments from \cite{BCSS2} to prove Theorem~\ref{thm:GRI.2}.

Throughout this section,  $\Gamma := \Gamma^{(n)}$. If $Q \subset \R^{nd}$ is a cellular set and $1 \le k \le \infty$, put
\[ 
\tilde{C}_c^k(\Gamma \cap Q):= \{f|_{\Gamma} : f \in C_c^k(Q)\}, \qquad W_0^{1,2}(\Gamma \cap Q) := \{f \in W^{1,2}(\Gamma \cap Q) : f|_{\partial Q} = 0 \},
\]
where $f|_{\partial Q}$ is understood in the trace sense. We start with a lemma which has to be justified in the context of multi-particle quantum graphs.

\begin{lem}      \label{lem:sym}
Let $\Lambda \subset \R^{nd}$ be a cube. Then for all $h \in (W_0^{1,2}(\Gamma \cap \Lambda))^n$ and $w \in W^{1,2}(\Gamma \cap \Lambda):$
\[ 
\langle \nabla \cdot h,w \rangle = - \langle h, \nabla w \rangle \, .
\]
\end{lem}
\begin{proof} 
Let $h = ((h_{\kappa}^{(1)}),\ldots,(h_{\kappa}^{(n)}))$ and $w = (w_{\kappa})$. Fix $\kappa \in \mathcal{K}(\Gamma \cap \Lambda)$ and let $\sigma = \partial \kappa$. Using the notation $\frac{\partial }{\partial x^i} \equiv \partial_i$, we have by Green's formula (see e.g. \cite[Theorem 1.5.3.1]{Gri})
\begin{equation}
\langle \partial_i h_{\kappa}^{(i)} , w_{\kappa} \rangle = - \langle h_{\kappa}^{(i)} , \partial_i w_{\kappa} \rangle + \int_{\sigma(\kappa)} h_{\kappa}^{(i)} \bar{w}_{\kappa} \nu^{(i)} \, \dd\sigma \, ,       \label{eq:gri1}
\end{equation}
where the values of $h_{\kappa}^{(i)} \bar{w}_{\kappa}$ on $\sigma:=\sigma(\kappa)$ are understood in the trace sense and $\nu := (\nu^{(1)},\ldots,\nu^{(n)})$ is the outward unit vector normal to $\sigma$, well defined on each $\sigma^j$. Identify $\kappa \equiv [0,1]^n$ as in Section~\ref{sec:n-gra} and denote points in $\kappa$ by $(x^1, \ldots, x^n)$, with $x^i \in [0,1]$. If $\sigma^j$ is the face with points $(x^1,\ldots,x^{j-1},0,x^{j+1},\ldots,x^n)=:\hat{x}^j_0$ and if $\sigma^{o(j)}$ is the face opposite to it with points $(x^1,\ldots,x^{j-1},1,x^{j+1},\ldots,x^n)=:\hat{x}^j_1$, then $\left. {\nu} \right|_{\sigma^j} = (0,\ldots,0,-1,0,\ldots,0)$ and $\left. {\nu} \right|_{\sigma^{o(j)}} = (0,\ldots,0,1,0,\ldots,0)$. Hence
\[ 
\int_{\sigma(\kappa)} h_{\kappa}^{(i)} \bar{w}_{\kappa} \nu^{(i)} \dd\sigma = \int_{\sigma^{o(i)}(\kappa)} h_{\kappa}^{(i)}(\hat{x}_1^i) \bar{w}_{\kappa}(\hat{x}_1^i) \, \dd \hat{x}^i - \int_{\sigma^i(\kappa)} h_{\kappa}^{(i)}(\hat{x}_0^i) \bar{w}_{\kappa}(\hat{x}_0^i) \, \dd \hat{x}^i \, , 
\]
where $\dd \hat{x}^i := \dd x^1 \ldots \dd x^{i-1} \dd x^{i+1} \ldots \dd x^n$. Now consider
\[
\sum_{\kappa \in \mathcal{K}(\Gamma \cap \Lambda)} \Big( \int_{\sigma^{o(i)}(\kappa)} h_{\kappa}^{(i)}(\hat{x}_1^i) \bar{w}_{\kappa}(\hat{x}_1^i) \, \dd \hat{x}^i - \int_{\sigma^i (\kappa)} h_{\kappa}^{(i)}(\hat{x}_0^i) \bar{w}_{\kappa}(\hat{x}_0^i) \, \dd \hat{x}^i \Big) \, .
\]
Since $h |_{\partial \Lambda} = 0$, this sum may be re-arranged as
\[ 
\sum_{\text{inner } \sigma^i} \sum_{j=1}^d \int_{\sigma^i}  \Big\{ h_{\kappa_j^-(\sigma^i)}^{(i)}(\hat{x}_1^i)  \bar{w}_{\kappa_j^-(\sigma^i)}(\hat{x}_1^i) - h_{\kappa_j^+(\sigma^i)}^{(i)}(\hat{x}_0^i)  \bar{w}_{\kappa_j^+(\sigma^i)}(\hat{x}_0^i) \Big\} \, \dd \hat{x}^i \, , 
\]
where $\kappa_j^-(\sigma^i)$ and $\kappa_j^+(\sigma^i)$, $j=1,\ldots,d$ are the $2d$ cubes containing $\sigma^i$ as a common face and $\kappa_j^-$ is opposite to $\kappa_j^+$. But by hypothesis $h^{(i)} \bar{w}$ are continuous on $\sigma^i$, i.e. $h^{(i)}_{\kappa^-_j} (\hat{x}_1^i) \bar{w}_{\kappa^-_j} (\hat{x}_1^i) = h^{(i)}_{\kappa^+_j} (\hat{x}_0^i) \bar{w}_{\kappa^+_j} (\hat{x}_0^i)$ a.e. Hence the sum vanishes and $\sum_{\kappa \in \mathcal{K}(\Gamma \cap \Lambda)} \int_{\sigma (\kappa)} h_{\kappa}^{(i)} \bar{w}_{\kappa} \nu^{(i)} \dd\sigma = 0$. The assertion thus follows by summing in (\ref{eq:gri1}) over $\kappa \in \mathcal{K}(\Gamma \cap \Lambda)$ and $i=1,\ldots,n$.
\end{proof}

\begin{thm}[Geometric Resolvent Equation]      \label{thm:GRE}
Let $\Lambda_1^{(n)} \subseteq \Lambda_2^{(n)} \subset \R^{nd}$ be two cubes, $\psi \in \tilde{C}_c^{\infty}(\Gamma \cap \Lambda_1^{(n)})$ real-valued, and $E \in \rho(H_{\Lambda_1}^{(n)}) \cap \rho(H_{\Lambda_2}^{(n)})$. Then
\[ 
G_{\Lambda_1^{(n)}}(E) \psi = \psi G_{\Lambda_2^{(n)}}(E) +G_{\Lambda_1^{(n)}}(E) \left( (\nabla \psi) \cdot \nabla + \nabla \cdot (\nabla \psi) \right) G_{\Lambda_2^{(n)}}(E) \tag{\textsf{GRE}}
\]
as operators on $L^2(\Gamma \cap \Lambda_2^{(n)})$. 
\end{thm}
\begin{proof} 
Let $g \in L^2(\Gamma \cap \Lambda_2^{(n)})$, $u := (\psi G_{\Lambda_2^{(n)}} +G_{\Lambda_1^{(n)}} \left( (\nabla \psi) \cdot \nabla + \nabla \cdot (\nabla \psi) \right) G_{\Lambda_2^{(n)}} ) g$, where $G_{\Lambda^{(n)}_i} := G_{\Lambda^{(n)}_i}(E)$ and put $\mathfrak{h}_{\Lambda_i^{(n)}} := \mathfrak{h}^{(n)}_{\omega,\Lambda_i}$. It suffices to show that $u \in D(\mathfrak{h}_{\Lambda_1^{(n)}})$ and
\[ 
(\mathfrak{h}_{\Lambda_1^{(n)}} - E)[u,w] = \langle \psi g, w \rangle \qquad \text{for all } w \in D(\mathfrak{h}_{\Lambda_1^{(n)}}) \, .
\]
Since $\psi \in \tilde{C}_c^{\infty}(\Gamma \cap \Lambda_1^{(n)})$ and $G_{\Lambda_2^{(n)}} g \in W^{1,2}(\Gamma \cap \Lambda_2^{(n)})$, we have $\psi G_{\Lambda_2^{(n)}} g \in D(\mathfrak{h}_{\Lambda_1^{(n)}})$ by Lemma~\ref{lem:sobolev}. Similarly $(\nabla \psi) G_{\Lambda_2^{(n)}} g \in (W^{1,2}(\Gamma \cap \Lambda_1^{(n)}))^n$, so $\nabla \cdot (\nabla \psi) G_{\Lambda_2^{(n)}} g \in L^2(\Gamma \cap \Lambda_1^{(n)})$ and $G_{\Lambda_1^{(n)}} [\nabla \cdot (\nabla \psi) G_{\Lambda_2^{(n)}} g] \in D(H_{\Lambda_1}^{(n)})$. Finally $(\nabla \psi) \cdot \nabla G_{\Lambda_2^{(n)}} g \in L^2(\Gamma \cap \Lambda_1^{(n)})$, hence $G_{\Lambda_1^{(n)}}[(\nabla \psi) \cdot \nabla G_{\Lambda_2^{(n)}} g] \in D(H_{\Lambda_1}^{(n)})$. Thus, $u \in D(\mathfrak{h}_{\Lambda_1^{(n)}})$ and
\begin{align*}
(\mathfrak{h}_{\Lambda_1^{(n)}} - E)[u,w] & = (\mathfrak{h}_{\Lambda_1^{(n)}} - E)[\psi G_{\Lambda_2^{(n)}}g,w] + \langle ((\nabla \psi) \cdot \nabla + \nabla \cdot (\nabla \psi)) G_{\Lambda_2^{(n)}} g, w \rangle \\
& = (\mathfrak{h}_{\Lambda_1^{(n)}} - E)[\psi G_{\Lambda_2^{(n)}}g,w] + \langle (\nabla \psi) \cdot \nabla (G_{\Lambda_2^{(n)}} g),w \rangle - \langle (\nabla \psi)G_{\Lambda_2^{(n)}} g, \nabla w \rangle \\
& = \langle \psi \nabla (G_{\Lambda_2^{(n)}} g), \nabla w \rangle + \langle  (V^{\omega} - E) \psi G_{\Lambda_2^{(n)}} g, w \rangle + \langle \nabla(G_{\Lambda_2^{(n)}} g),(\nabla \psi) w \rangle \\
& = (\mathfrak{h}_{\Lambda_2^{(n)}} - E)[G_{\Lambda_2^{(n)}} g, \psi w] = \langle g, \psi w \rangle = \langle \psi g, w \rangle
\end{align*}
where we used Lemma~\ref{lem:sym} in the second equality.
\end{proof}

\begin{lem}      \label{lem:SOL}
Let $\Lambda^{(n)}$ be a cube or $\Lambda^{(n)} = \R^{nd}$, let $\tilde{Q} \subset Q \subset \Lambda^{(n)}$ be cellular sets with $\dist(\partial Q, \partial \tilde{Q}) \ge 1$ and let $E_+ \in \R$. Then there exists $C = C(E_+,n,d,q_-)>0$ such that for any $E \le E_+$, if $f \in D(H^{(n)}_{\Lambda})$, then
\[ 
\| \chi_{\tilde{Q}} \nabla f \| \le C \cdot ( \| \chi_Q (H_{\Lambda}^{(n)}-E) f \| + \| \chi_Q f \| ) \, . \tag{\textsf{SOL}} 
\]
\end{lem}
\begin{proof} 
Since $\dist(\partial Q, \partial \tilde{Q}) \ge 1$, we may choose a real $\psi \in \tilde{C}_c^{\infty}(\Gamma \cap Q)$, $0 \le \psi \le 1$ with $\psi \equiv 1$ on $\Gamma \cap \tilde{Q}$ and $\| \nabla \psi \|_{\infty} \le C_1(nd)$. If $w:= f \psi^2$, then $w \in D(\mathfrak{h}_{\Lambda})$ by Lemma~\ref{lem:sobolev} and
\[ 
\langle \nabla f, \nabla w \rangle = \langle \psi \nabla f , \psi \nabla f \rangle + 2 \langle \psi \nabla f , f \nabla \psi \rangle \, . 
\]
Denoting $g:= (H_{\Lambda}^{(n)}-E)f$ we thus get
\begin{align*}
\| \psi \nabla f \|^2 & = \langle \nabla f , \nabla w \rangle - 2 \langle \psi \nabla f , f \nabla \psi \rangle \\
& = \langle g, w \rangle - \langle (V^{\omega}-E) f, w \rangle - 2 \langle \psi \nabla f, f \nabla \psi \rangle \\
& = \langle g \psi, f \psi \rangle - \langle V^{\omega} f \psi, f \psi \rangle + E \|f \psi \|^2 - 2 \langle \psi \nabla f, f \nabla \psi \rangle \\
& \le \|g \|_Q \|f \|_Q + C_2 \|f\|_Q^2 + 2 C_1 \| \psi \nabla f \| \|f \|_Q \, ,
\end{align*}
where $\| \phi \|_Q := \| \chi_Q \phi \|$ and $C_2:= |E_+ - nq_-|$. Hence
\[
\big( \| \psi \nabla f\| - C_1 \|f\|_Q \big)^2 \le \|g \|_Q \|f \|_Q + (C_1^2+C_2) \|f\|_Q^2 \le \Big(C_3 \|f\|_Q + \frac{1}{2C_3} \|g\|_Q \Big)^2,
\]
where $C_3 := \sqrt{ C_1^2 + C_2}$. The assertion follows by taking square roots.
\end{proof}

\begin{thm}      \label{thm:GRI}
Let $\Lambda_l^{(n)} \subset \Lambda_L^{(n)}$ be cubes with $l \ge 7$, let $A \subseteq \Lambda_{l-6}^{(n)}$, $B \subseteq \Lambda_L^{(n)} \setminus \Lambda_l^{(n)}$ be cellular sets and let $E_+ \in \R$. Then there exists $C=C(E_+,n,d,q_-)>0$ such that for all $E \in \rho(H_{\Lambda_L}^{(n)}) \cap \rho(H_{\Lambda_l}^{(n)}) \cap (-\infty,E_+]:$
\[ 
\|\chi_A G_{\Lambda_L^{(n)}}(E) \chi_B \| \le C \cdot \|\chi_A G_{\Lambda_l^{(n)}}(E) \chi_{\Lambda_l^{\out}} \| \cdot \| \chi_{\Lambda_l^{\out}} G_{\Lambda_L^{(n)}}(E) \chi_B \| \, .  \tag{\textsf{GRI.1}} 
\]
In particular, if $\mathbf{u} \in \Lambda_{l-7}^{(n)}$ and $\Lambda_l^{(n)} \subset \Lambda_{L-7}^{(n)}$, then given $\mathbf{y} \in \mathbf{B}_L^{\out}$, we have
\[ 
\| G_{\Lambda_L^{(n)}}(\mathbf{u},\mathbf{y};E) \| \le C \cdot |\mathbf{B}_l^{\out}|^2 \max_{\mathbf{w} \in \mathbf{B}_l^{\out}} \| G_{\Lambda_l^{(n)}}(\mathbf{u},\mathbf{w};E) \|  \max_{\mathbf{z} \in \mathbf{B}_l^{\out}} \| G_{\Lambda_L^{(n)}}(\mathbf{z},\mathbf{y};E) \| \, .  \tag{\textsf{GRI.2}} 
\]
\end{thm}
\begin{proof} 
Let $G_{\Lambda} := G_{\Lambda}(E)$, $Q = \inter \Lambda_l^{\out}$ and choose a real $\psi \in \tilde{C}_c^{\infty}(\Gamma \cap \Lambda_l^{(n)})$ such that $\psi = 1$ on $\Gamma \cap \Lambda_{l-4}^{(n)}$, $\supp \psi \subset \Lambda_{l-2}^{(n)}$ and $\| \nabla \psi\|_{\infty}$ is bounded independently of $\Lambda_l^{(n)}$. Then
\begin{align*}
\|\chi_A G_{\Lambda_L^{(n)}} \chi_B \| & = \|\chi_A(\psi G_{\Lambda_L^{(n)}} - G_{\Lambda_l^{(n)}} \psi) \chi_B \| \qquad (\psi |_{\Gamma \cap A} = 1, \psi |_{\Gamma \cap B} = 0) \\
& = \| \chi_A ( G_{\Lambda_l^{(n)}} ( ( \nabla \psi) \cdot \nabla + \nabla \cdot ( \nabla \psi)) G_{\Lambda_L^{(n)}}) \chi_B \| \qquad \textsf{(GRE)} \\
& \le \| \chi_A G_{\Lambda_l^{(n)}} ( \nabla \psi) \cdot \nabla G_{\Lambda_L^{(n)}} \chi_B \| + \| \chi_A G_{\Lambda_l^{(n)}} \nabla \cdot ( \nabla \psi) G_{\Lambda_L^{(n)}} \chi_B \| \, .
\end{align*}
Now let $\tilde{Q} = \inter(\Lambda_{l-1}^{(n)} \setminus \Lambda_{l-5}^{(n)})$, so $\supp \nabla \psi \subset \tilde{Q}$ and $\dist(\partial Q , \partial \tilde{Q}) = 1$. Hence given $f_1,f_2 \in L^2(\Gamma \cap \Lambda_L^{(n)})$, $\| f_1 \| = \|f_2\|= 1$, we have
\begin{align*}
|\langle \chi_A G_{\Lambda_l^{(n)}} ( \nabla \psi) \cdot \nabla G_{\Lambda_L^{(n)}} \chi_B f_1, f_2 \rangle| & = |\langle \nabla G_{\Lambda_L^{(n)}} \chi_B f_1, (\nabla \psi) G_{\Lambda_l^{(n)}} \chi_A f_2 \rangle| \\
& \le \| \nabla \psi \|_{\infty} \|\chi_{\tilde{Q}} \nabla G_{\Lambda_L^{(n)}} \chi_B \| \| \chi_{\tilde{Q}} G_{\Lambda_l^{(n)}} \chi_A \| \, .
\end{align*}
Furthermore, using Lemma~\ref{lem:sym} we have
\begin{align*}
|\langle \chi_A G_{\Lambda_l^{(n)}} \nabla \cdot ( \nabla \psi) G_{\Lambda_L^{(n)}} \chi_B f_1 , f_2 \rangle| & =  |\langle f_1, \chi_B G_{\Lambda_L^{(n)}} ( \nabla \psi) \cdot \nabla G_{\Lambda_l^{(n)}} \chi_A f_2 \rangle| \\
& \le \| \nabla \psi\|_{\infty} \| \chi_B G_{\Lambda_L^{(n)}} \chi_{\tilde{Q}} \| \| \chi_{\tilde{Q}} \nabla G_{\Lambda_l^{(n)}} \chi_A \| \, .
\end{align*}
Noting that for a bounded operator $T$ we have $\|T\| = \|T^{\ast}\|$, we thus get
\[ 
\|\chi_A G_{\Lambda_L^{(n)}} \chi_B \| \le \| \nabla \psi \|_{\infty} (\| \chi_A G_{\Lambda_l^{(n)}} \chi_{\tilde{Q}} \| \| \chi_{\tilde{Q}} \nabla G_{\Lambda_L^{(n)}} \chi_B \| + \| \chi_{\tilde{Q}} \nabla G_{\Lambda_l^{(n)}} \chi_A \| \| \chi_{\tilde{Q}} G_{\Lambda_L^{(n)}} \chi_B \|).
\]
Now by Lemma~\ref{lem:SOL}, we can find $C_1$ such that
\[ 
\| \chi_{\tilde{Q}} \nabla G_{\Lambda_L^{(n)}} \chi_B \| \le C_1 \cdot \| \chi_{Q} G_{\Lambda_L^{(n)}} \chi_B \| \, .
\]
Indeed, given $u \in L^2(\Gamma \cap \Lambda_L^{(n)})$, apply \textsf{(SOL)} to $f = G_{\Lambda_L^{(n)}} \chi_B u$. Noting that $(H_{\Lambda_L}^{(n)} - E) f =  \chi_B u = 0$ on $Q$ we get $\| \chi_{\tilde{Q}} \nabla G_{\Lambda_L^{(n)}} \chi_B u \| \le C_1 \cdot \| \chi_{Q} G_{\Lambda_L^{(n)}} \chi_B u \|$. As $u$ is arbitrary, the assertion follows. In the same way we find $C_2$ such that
\[ 
\| \chi_{\tilde{Q}} \nabla G_{\Lambda_l^{(n)}} \chi_A \| \le C_2 \cdot \| \chi_Q G_{\Lambda_l^{(n)}} \chi_A \|  = C_2 \cdot \| \chi_A G_{\Lambda_l^{(n)}} \chi_Q \| \, . 
\]
Noting that $\tilde{Q} \subset Q$, we finally get
\[
\|\chi_A G_{\Lambda_L^{(n)}} \chi_B \|  \le C \cdot \| \chi_A G_{\Lambda_l^{(n)}} \chi_Q \| \cdot \| \chi_Q G_{\Lambda_L^{(n)}} \chi_B \|
\]
for $C = \max(2C_1 \| \nabla \psi\|_{\infty}, 2C_2 \| \nabla \psi \|_{\infty})$. We thus have \textsf{(GRI.1)}. 

For \textsf{(GRI.2)}, note that $\Lambda_l^{\out} \subseteq \bigcup _{\mathbf{w} \in \mathbf{B}_l^{\out}} \mathbf{C}(\mathbf{w}) $, so \textsf{(GRI.1)} gives us
\[ 
\| \chi_{\mathbf{u}} G_{\Lambda_L^{(n)}}(E) \chi_{\mathbf{y}} \| \le C \sum_{\mathbf{w},\mathbf{z} \in \mathbf{B}_l^{\out}} \| \chi_{\mathbf{u}} G_{\Lambda_l^{(n)}}(E) \chi_{\mathbf{w}} \| \| \chi_{\mathbf{z}} G_{\Lambda_L^{(n)}}(E) \chi_{\mathbf{y}} \|.  \qedhere 
\]
\end{proof}

We now give a resolvent inequality which is special to multi-particle systems.

\begin{thm}      \label{thm:GRI.2}
Let $\Lambda_L^{(n)}(\mathbf{u})$ be a $\mathcal{J}$-decomposable cube, let $\mathbf{x},\mathbf{y} \in \mathbf{B}^{(n)}_L(\mathbf{u})$ and suppose that $E \in \rho(H_{\Lambda_L(\mathbf{u})}^{(n)})$. There exists $S^{\ast} = S^{\ast}(n,q_-,E)$ such that for $S>S^{\ast}$, and under the notations of Remark~\ref{rem:dec}, if $\delta_1 := |x_{\mathcal{J}^c} - y_{\mathcal{J}^c}| >2$, then
\[ 
\|G_{\Lambda_L^{(n)}(\mathbf{u})}(\mathbf{x},\mathbf{y};E) \| \le M_1 \cdot \max \limits_{a \le M_1 } \|G_{\Lambda_L^{(n'')}(u_{\mathcal{J}^c})}(x_{\mathcal{J}^c},y_{\mathcal{J}^c};E-\lambda_a)\| + |\Lambda_L^{(n')}| e^{- \delta_1 S} \tag{\textsf{GRI.3}} 
\]
for $M_1 =  \Big(\lfloor \frac{d^{n'}((4S)^2+E-nq_-)^{n'/2}}{(4\pi)^{n'/2} \Gamma(n'/2)} \rfloor +1\Big) \cdot |\Lambda_L^{(n')}| $, and if $\delta_2 := |x_{\mathcal{J}} - y_{\mathcal{J}}| > 2$, then
\[ 
\|G_{\Lambda_L^{(n)}(\mathbf{u})}(\mathbf{x},\mathbf{y};E) \| \le M_2 \cdot \max \limits_{b \le M_2} \|G_{\Lambda_L^{(n')}(u_{\mathcal{J}})}(x_{\mathcal{J}},y_{\mathcal{J}};E-\mu_b)\| + |\Lambda_L^{(n'')}| e^{- \delta_2 S} \tag{\textsf{GRI.3}} 
\]
for $M_2 = \Big(\lfloor \frac{d^{n''}((4S)^2+E-nq_-)^{n''/2}}{(4\pi)^{n''/2} \Gamma(n''/2)} \rfloor +1\Big) \cdot |\Lambda_L^{(n'')}|$.
\end{thm}
\begin{proof} 
We only prove the first bound; the second one is similar. Put $\Lambda^{(n)} := \Lambda_L^{(n)}(\mathbf{u}) $, $\Lambda_1^{(n')} := \Lambda_L^{(n')}(u_{\mathcal{J}})$ and $\Lambda_2^{(n'')} := \Lambda_L^{(n'')}(u_{\mathcal{J}^c})$. Using (\ref{eq:dec}) with $\eta (t) := (t-E)^{-1}$ we get
\[
G_{\Lambda^{(n)}}(E) = \sum_a P_a \otimes \Big( \sum_b \frac{1}{\mu_b - (E- \lambda_a)} P_b \Big) = \sum_a P_a \otimes G_{\Lambda_2^{(n'')}}(E- \lambda_a) \, .
\]
Hence noting that $G_{\Lambda^{(n)}}(\mathbf{x},\mathbf{y};E) := \chi_{\mathbf{x}} G_{\Lambda^{(n)}}(E) \chi_{\mathbf{y}}$, we get
\begin{align*}
\|G_{\Lambda^{(n)}}(\mathbf{x},\mathbf{y};E) \| & \le \sum_a \| \chi_{x_{\mathcal{J}}} P_a \chi_{y_{\mathcal{J}}} \otimes \chi_{x_{\mathcal{J}^c}} G_{\Lambda_2^{(n'')}}(E- \lambda_a) \chi_{y_{\mathcal{J}^c}} \| \\
& \le \sum_a \| \chi_{x_{\mathcal{J}^c}} G_{\Lambda_2^{(n'')}}(E- \lambda_a) \chi_{y_{\mathcal{J}^c}} \| \, .
\end{align*}
Now given $S_j \gg 1$, by \textsf{(WEYL.$n'$)} the constants $C_j = \lfloor \frac{d^{n'}(S_j+E-nq_-)^{n'/2}}{(4\pi)^{n'/2}\Gamma(n'/2)} \rfloor +1$ satisfy 
\[ 
a>C_j|\Lambda_1^{(n')}| \implies \lambda_a > S_j+E-n''q_- \implies \eta_a >S_j \, , 
\]
where $\eta_a := n''q_- - (E - \lambda_a)$. Hence if $\delta_1>2$, taking $\delta := \dist(C(x_{\mathcal{J}^c}),C(y_{\mathcal{J}^c})) = \delta_1 -2$ and $S_j := (4Sj)^2$, we get by Combes-Thomas estimate,
\[
\sum_{a = C_j | \Lambda_1^{(n')}| +1}^{C_{j+1} | \Lambda_1^{(n')}|} \| G_{\Lambda_2^{(n'')}}(x_{\mathcal{J}^c},y_{\mathcal{J}^c};E-\lambda_a)\| \le (C_{j+1} - C_j) |\Lambda_1^{(n')}| e^{-\delta \sqrt{\frac{S_j}{2}}} \le |\Lambda_1^{(n')}| e^{-\delta \frac{\sqrt{S_j}}{2}}
\]
provided $S$ is large enough. Hence
\[ 
\sum_{a>C_1|\Lambda_1^{(n')}|} \| G_{\Lambda_2^{(n'')}}(x_{\mathcal{J}^c},y_{\mathcal{J}^c};E-\lambda_a)\| \le |\Lambda_1^{(n')}| \sum_{j=1}^{\infty} e^{-\delta \frac{\sqrt{S_j}}{2}} \, .
\]
But
\[
\sum_{j=1}^{\infty} e^{-\delta \frac{\sqrt{S_j}}{2}} = \sum_{j=1}^{\infty} e^{-(2 \delta S)j} = \frac{e^{-2 \delta S}}{1 - e^{-2 \delta S}} \le 2 e^{-2 \delta S} \le e^{- \delta_1 S} \, .
\]
We thus obtain the first bound with $M_1 := C_1 |\Lambda_1^{(n')}|$.
\end{proof}

\section{Wegner Estimates}    \label{sec:W}
To establish Wegner estimates we use some ideas of \cite{CS}, but we rely entirely on measure-theoretic arguments. For a probability measure $\mu$ on $\R$ we put
\[ 
s(\mu,\varepsilon) := \sup \, \{ \mu[a,b]:b-a \le \varepsilon \} \, .
\]

Given $J \subset \mathcal{E}(\Gamma^{(1)})$ and $\omega \in \Omega$, we denote $\omega = (\omega^J, \omega^{J^c})$, where $\omega^J := (\omega_e)_{e \in J}$. If $A \subseteq \Omega$ is measurable and $\omega^{J^c}$ is fixed, we define the section $A_{\omega^{J^c}} := \{ \omega^J : (\omega^J,\omega^{J^c}) \in A \}$ and put $\prob_J := \mathop \otimes_{e \in J} \mu$. Then by definition of a product measure, we have $\prob(A) = \expect_{J^c} \{\prob_J (A_{\omega^{J^c}})\}$, where $\expect_{J^c}$ denotes the integration over $\omega^{J^c}$. 

\begin{thm}      \label{thm:Wegner1}
Let $E \in \R$ and $\varepsilon >0$. There exists a non-random $C=C(n,d,E+\varepsilon -nq_-)$ such that for any $\Lambda_{\LL}^{(n)}(\mathbf{u})$ and any $1 \le i \le n$, if $J:= \mathcal{E}\big(\Gamma^{(1)} \cap \Pi_i \Lambda_{\LL}^{(n)}(\mathbf{u})\big)$, then
\[ 
\prob_J \big(\{ \dist(\sigma(H_{\Lambda_{\LL}(\mathbf{u})}^{(n)}(\omega)),E) < \varepsilon \big\}_{\omega^{J^c}} \big) \le C \cdot |\Lambda_{\LL}^{(n)}(\mathbf{u})| \cdot |\Pi_i\Lambda_{\LL}^{(n)}(\mathbf{u})| \cdot s(\mu,2 \varepsilon)
\]
for any $\omega^{J^c}$.
\end{thm}
\begin{proof}
Put $\Lambda := \Lambda_{\LL}^{(n)}(\mathbf{u})$ and fix $\omega^{J^c}$. By Lemma~\ref{lem:WEYL}, we may find $C'=C'(n,d,E+\varepsilon -nq_-)$ such that $E_j(\omega) := E_j(H_{\Lambda}^{(n)}(\omega)) > E+\varepsilon$ if $j>C' \cdot |\Lambda|$. Hence,
\begin{equation}
\prob_J \big( \{ \dist(\sigma(H_{\Lambda}^{(n)}(\omega)),E) < \varepsilon \}_{\omega^{J^c}} \big) \le \sum_{j \le C'|\Lambda|} \prob_J \big( \{ |E_j(\omega) - E| < \varepsilon \}_{\omega^{J^c}} \big) \, .       \label{eq:w1} 
\end{equation}

Given $\kappa = (e_1, \ldots , e_n) \in \mathcal{K}(\Gamma \cap \Lambda)$ we have
\[
W_{\kappa}^{\omega} = \omega_{e_1} + \ldots + \omega_{e_n} = \sum_{e \in \mathcal{E}} c_{\kappa}(e) \omega_e, \quad \text{where } c_{\kappa}(e):=
\begin{cases}
1&\text{if }e= e_j \text{ for some }j,\\
0&\text{otherwise}.\end{cases} 
\]

Hence
\[ 
W_{\kappa}^{\omega} = \sum_{e \in J} c_{\kappa}(e) \omega_e + \sum_{e \in J^c} c_{\kappa}(e) \omega_e = W_{\omega^J}(\kappa) + W_{\omega^{J^c}}(\kappa) \, . 
\]

Now
\[ 
H_{\Lambda}^{(n)}(\omega) = -\Delta + U + W_{\omega^{J^c}} + W_{\omega^J} = K_{\omega^{J^c}} + W_{\omega^J}, 
\]
where the operator $K_{\omega^{J^c}}$ does not depend on $\omega^J$. Let $(f_{\kappa}) \in L^2(\Gamma \cap \Lambda)$ with $\| (f_{\kappa})\|=1$, let $t \ge 0$ and denote $\one := (1,\ldots,1) \in \R^J$. Then
\[
H_{\Lambda}^{(n)}(\omega^J+ t \cdot \one, \omega^{J^c}) (f_{\kappa}) = (K_{\omega^{J^c}} + W_{\omega^J+t\cdot \one})(f_{\kappa}) = (K_{\omega^{J^c}} + W_{\omega^J})(f_{\kappa}) + t(n_{\kappa} f_{\kappa}),
\]
where $n_{\kappa}:= \sum_{e \in J} c_{\kappa}(e)$. Since every $\kappa \in \mathcal{K}(\Gamma \cap \Lambda)$ takes the form $(e_1,\ldots,e_n)$ with $e_i \in \mathcal{E}(\Gamma \cap \Pi_i \Lambda) = J$, we have $1 \le n_{\kappa} \le n$. Hence
\[
\langle H_{\Lambda}^{(n)}(\omega^J+ t \cdot \one, \omega^{J^c}) (f_{\kappa}), (f_{\kappa}) \rangle \ge \langle H_{\Lambda}^{(n)}(\omega^J, \omega^{J^c}) (f_{\kappa}), (f_{\kappa}) \rangle + t \,.
\]
By the min-max principle, it follows that $E_j(\omega^J + t \cdot \one, \omega^{J^c}) \ge E_j(\omega^J, \omega^{J^c}) + t$. Finally, if $v_e \le w_e$ for all $e \in J$, then $H(v^J,\omega^{J^c}) \le H(w^J,\omega^{J^c})$ and thus $E_j(v^J,\omega^{J^c}) \le E_j(w^J,\omega^{J^c})$. Hence the $E_j(\,\cdot\, , \omega^{J^c}) : \R^J \to \R$ satisfy the hypotheses of Stollmann's lemma (see \cite{Sto2} and \cite{Sab2}) for any $\omega^{J^c}$, so we get
\[ 
\prob_J \big( \{ \omega^J: |E_j(\omega^J,\omega^{J^c}) - E| < \varepsilon \} \big) \le |J| \cdot s(\mu,2\varepsilon) \le d \cdot | \Pi_i \Lambda| \cdot s(\mu,2 \varepsilon) 
\]
by \textsf{(NB.1)}. The theorem follows by (\ref{eq:w1}).
\end{proof}

\begin{thm}      \label{thm:Wegner2}
Let $I=[a,b]$ be a bounded interval and let $\varepsilon >0$. There exists $C=C(n,d,b+\varepsilon-nq_-)$ such that for any pre-separable $\Lambda_{\LL}^{(n)}(\mathbf{u})$ and $\Lambda_{\KK}^{(n)}(\mathbf{v})$ we have
\[ 
\prob \{ \dist (\sigma_I(H_{\Lambda_{\LL}(\mathbf{u})}^{(n)}), \sigma_I(H_{\Lambda_{\KK}(\mathbf{v})}^{(n)}) ) < \varepsilon \} \le C \cdot |\Lambda_{\LL}^{(n)}(\mathbf{u})| \cdot |\Lambda_{\KK}^{(n)}(\mathbf{v})| \cdot |\Pi_0 \Lambda| \cdot s(\mu,2\varepsilon), 
\]
where $\sigma_I(H_{\Lambda}^{(n)}) := \sigma(H_{\Lambda}^{(n)}(\omega)) \cap I$ and $|\Pi_0 \Lambda| := \max_{i,j} \big( |\Pi_i \Lambda_{\LL}^{(n)}(\mathbf{u})|, |\Pi_j \Lambda_{\KK}^{(n)}(\mathbf{v})| \big)$.
\end{thm}
\begin{proof} 
Suppose $\Lambda_{\KK}^{(n)}(\mathbf{v})$ is $\mathcal{J}$-pre-separable of $\Lambda_{\LL}^{(n)}(\mathbf{u})$ for some $\emptyset \neq \mathcal{J} \subseteq \{1,\ldots,n\}$, i.e. $\Pi_{\mathcal{J}} \Lambda_{\KK}^{(n)}(\mathbf{v}) \cap \big(\Pi_{\mathcal{J}^c} \Lambda_{\KK}^{(n)}(\mathbf{v}) \cup \Pi \Lambda_{\LL}^{(n)}(\mathbf{u}) \big) = \emptyset$. Fix $i \in \mathcal{J}$ and put $J:= \mathcal{E}(\Gamma^{(1)} \cap \Pi_i \Lambda_{\KK}^{(n)}(\mathbf{v}))$. Since the eigenvalues $E_j^{(\mathbf{u})}(\omega)$ of $H_{\Lambda_{\LL}^{(n)}(\mathbf{u})}(\omega)$ do not depend on $\omega^J$, we may apply Theorem~\ref{thm:Wegner1} with $E=E_j^{\mathbf{u}}=E_j^{\mathbf{u}}(\omega^{J^c})$ to get
\begin{align*}
& \prob \{ \dist (\sigma_I(H_{\Lambda_{\LL}(\mathbf{u})}^{(n)}), \sigma_I(H_{\Lambda_{\KK}(\mathbf{v})}^{(n)}) ) < \varepsilon \} \\
& \quad = \expect_{J^c} \Big\{ \prob_J \{ \dist (\sigma_I(H_{\Lambda_{\LL}(\mathbf{u})}^{(n)}), \sigma_I(H_{\Lambda_{\KK}(\mathbf{v})}^{(n)}) ) < \varepsilon \}_{\omega^{J^c}} \Big\} \\
& \quad = \expect_{J^c} \Big\{ \prob_J \Big\{ \min_{a \le E_j^{(\mathbf{u})} \le b} \dist(E_j^{(\mathbf{u})}, \sigma_I(H_{\Lambda_{\KK}(\mathbf{v})}^{(n)})) < \varepsilon \Big\}_{\omega^{J^c}} \Big\} \\
& \quad \le \expect_{J^c} \sum_{j \le C_1|\Lambda_{\LL}^{(n)}(\mathbf{u})|} \prob_J \big(\{ \dist(E_j^{(\mathbf{u})}, \sigma_I(H_{\Lambda_{\KK}(\mathbf{v})}^{(n)})) < \varepsilon \}_{\omega^{J^c}} \big) \\
& \quad \le C \cdot |\Lambda_{\LL}^{(n)}(\mathbf{u})| \cdot |\Lambda_{\KK}^{(n)}(\mathbf{v})| \cdot |\Pi_i \Lambda_{\KK}^{(n)}(\mathbf{v})| \cdot s(\mu,2\varepsilon)
\end{align*} 
where we used Lemma~\ref{lem:WEYL} to obtain $C_1=C_1(n,d,b-nq_-)$. If however $\Lambda_{\LL}^{(n)}(\mathbf{u})$ was $\mathcal{J}$-pre-separable of $\Lambda_{\KK}^{(n)}(\mathbf{v})$, we would get for $i \in \mathcal{J}$,
\[ 
\prob \{ \dist (\sigma_I(H_{\Lambda_{\LL}(\mathbf{u})}^{(n)}), \sigma_I(H_{\Lambda_{\KK}(\mathbf{v})}^{(n)}) ) < \varepsilon \} \le C \cdot |\Lambda_{\LL}^{(n)}(\mathbf{u})| \cdot |\Lambda_{\KK}^{(n)}(\mathbf{v})| \cdot |\Pi_i \Lambda_{\LL}^{(n)}(\mathbf{u})| \cdot s(\mu,2\varepsilon) \, . \qedhere 
\] 
\end{proof}

\section{Initial Length Scale Estimate}    \label{sec:ILS}
In this section we follow the ideas of \cite{Sto} and use a Cheeger inequality from \cite{Post} to prove Lifshitz-type asymptotics for $1$-particle systems. We then deduce the Initial Length Scale estimate (ILS) for our model. We speak of Lifshitz-type asymptotics because our result is not formulated in terms of the integrated density of states $N(E)$, as it is not needed here. Theorem~\ref{thm:Lif1} easily implies bounds of the form $N(E) \le e^{-\gamma' (E-q_-)^{-1/2}}$ for $E$ near $q_-$ if one knows that $N(E) \le \frac{1}{|\Lambda|} \expect \big\{ \tr[ \chi_{(-\infty, E)}(H_{\Lambda}^{(1)})] \big\}$; see \cite[Theorem 2.1.4]{Sto}. The existence of $N(E)$ was established in \cite{HV}, see also \cite{GLV}.

In the following for $l \in \N^{\ast}$ we put
\[ 
n_l := \#\mathcal{E}(\Gamma^{(1)} \cap \Lambda^{(1)}_l) = d(2l)(2l-1)^{d-1} \, . 
\]
\begin{thm}      \label{thm:Lif1} 
There exist $b>0$ and $\gamma>0$ such that for any $u \in \Z^d$,
\[ 
\prob\{E_1(H_{\Lambda_l(u)}^{(1)}(\omega)) \le q_- + b n_l^{-2}\} \le e^{-\gamma n_l} \, .
\]
\end{thm}
\begin{proof}
Put $\tilde{H}_{\Lambda_l}^{(1)}(\omega) = H_{\Lambda_l(u)}^{(1)}(\omega) -q_-$. Then
\[ 
\prob\{E_1(H_{\Lambda_l(u)}^{(1)}(\omega)) \le q_- + b n_l^{-2}\} = \prob\{E_1(\tilde{H}_{\Lambda_l}^{(1)}(\omega)) \le b n_l^{-2}\} \, .
\]
Now $\tilde{H}_{\Lambda_l}^{(1)}(\omega) = (-\Delta + \tilde{W}^{\omega})_{\Lambda_l^{(1)}}$, where $\tilde{W}^{\omega} : (f_e) \mapsto ((\omega_e-q_-)f_e)$. We may assume $\tilde{W}^{\omega} \le 1$ for all $\omega$, since if $\tilde{W}^{\omega}$ is larger, $E_1(\tilde{H}_{\Lambda_l}^{(1)}(\omega))$ gets larger and the probability gets smaller. Define for $t \in [-1,1]$,
\[ 
H(\omega,t) := (-\Delta + t \cdot \tilde{W}^{\omega})_{\Lambda_l^{(1)}}, \qquad E_j(\omega,t) := E_j(H(\omega,t)) \, .
\]
Since the normalized ground state $\phi_0$ of the Kirchhoff Laplacian $H(\omega,0)=- \Delta_{\Lambda_l^{(1)}}$ is the constant function $(n_l^{-1/2})$, we have by the Feynman--Hellmann theorem
\begin{equation}
E_1(\omega,0)' = \langle \tilde{W}^{\omega} \phi_0, \phi_0 \rangle = \frac{1}{n_l} \sum_{e \in \mathcal{E}(\Gamma^{(1)} \cap \Lambda_l^{(1)})} q_e(\omega)=: f_l(\omega) \, ,            \label{eq:ils1}
\end{equation}
where $q_e(\omega) = \omega_e - q_- \ge 0$. By \cite[Lemma 2.1.1]{Sto} we can find $s_0,\gamma>0$ such that
\[ 
\prob \{ f_l(\omega) \le s_0 \} \le e^{-\gamma n_l} \, .
\]
We now estimate the distance between $0=E_1(\omega,0)$ and the rest of the spectrum of $H(\omega,0)$ using Cheeger inequality. Let $X:= \Gamma^{(1)}  \cap \Lambda_l^{(1)}$ and $\mathcal{O} := \{ Y \subset X : Y \text{ open}, Y \neq X, Y \neq \emptyset \}$. For $Y \in \mathcal{O}$, let $| \partial Y|$ be the number of points on the boundary of $Y$, $\vol_1 Y$ be the total length of $Y$ and put $Y^c :=X \setminus Y$. Then any $Y \in \mathcal{O}$ satisfies $\min(\vol_1Y, \vol_1Y^c) \le \frac{1}{2} (\vol_1 X) = \frac{n_l}{2}$, hence the \emph{Cheeger constant} of $X$ satisfies 
\[
h(X) := \inf_{Y \in \mathcal{O}} \frac{| \partial Y|}{\min(\vol_1Y, \vol_1Y^c)} \ge \frac{2}{n_l} \, .
\]
By \cite[Theorem 6.1]{Post}, it follows that $E_2(\omega,0) \ge \frac{1}{4} h(X)^2 \ge n_l^{-2}$. A similar estimate can also be obtained using the Faber-Krahn inequality; see \cite{Nic}.

We may now apply \cite[Theorem 4.1.31]{Sto}, to find $c_1,c_2>0$ such that
\[ 
|E_1(\omega,t) - t \cdot E_1(\omega,0)'| \le c_2n_l^2t^2 \qquad \forall \, 0 \le t \le c_1n_l^{-2}. 
\]
Since $E_1(\omega,t) \le E_1(\omega,1) = E_1(\tilde{H}_{\Lambda_l}^{(1)}(\omega))$,  this gives by (\ref{eq:ils1})
\[ 
f_l(\omega) = E_1(\omega,0)' \le c_2n_l^2t + E_1(\tilde{H}_{\Lambda_l}^{(1)}(\omega))t^{-1} \qquad \forall \, 0 \le t \le c_1n_l^{-2}. 
\]
Choose $0 < c_3 \le c_1$ such that $c_2c_3 \le \frac{1}{2} s_0$. Then for $t = c_3 n_l^{-2}$ we get
\[
f_l(\omega) \le (s_0/2) + E_1(\tilde{H}_{\Lambda_l}^{(1)}(\omega)) c_3^{-1} n_l^2 \, .
\]
Hence, choosing $b>0$ such that $b c_3^{-1} \le \frac{1}{2} s_0$ we finally obtain
\[ 
\prob \{ E_1(\tilde{H}_{\Lambda_l}^{(1)}(\omega)) \le b n_l^{-2} \} \le \prob \{ f_l(\omega) \le s_0 \} \le e^{-\gamma n_l} \, . \qedhere 
\]
\end{proof}

\begin{thm}      \label{thm:Lif2}
There exist $b>0$ and $\gamma>0$ such that for any $\mathbf{u} \in \Z^{nd}$,
\begin{equation}
\prob\{E_1(H_{\Lambda_l(\mathbf{u})}^{(n)}(\omega)) \le nq_- + nb n_l^{-2}\} \le e^{-\gamma n_l} \, .     \label{eq:firstpart} 
\end{equation}
Consequently, for all $\xi >0$ and $\beta \in (0,1)$, we may find $L_0 = L_0 (N,d,\beta,\xi) $ as large as necessary such that for any $\mathbf{u} \in \Z^{nd}$,
\[ 
\prob \{ \dist(\sigma(H_{\Lambda_{L_0}(\mathbf{u})}^{(n)}(\omega)),nq_-) \le L_0^{\beta -1} \} \le L_0^{-\xi} \, .
\]
\end{thm}
\begin{proof} 
Let $\widehat{H}_{\Lambda_l(\mathbf{u})}^{(n)}(\omega) := H_{\Lambda_l(\mathbf{u})}^{(n)} - U_{\Lambda_l(\mathbf{u})}^{(n)} = (-\Delta +W^{\omega})_{\Lambda_l^{(n)}(\mathbf{u})}$. Since $U^{(n)} \ge 0$, we have $H_{\Lambda_l(\mathbf{u})}^{(n)} \ge \widehat{H}_{\Lambda_l(\mathbf{u})}^{(n)}$, hence $E_1(H_{\Lambda_l(\mathbf{u})}^{(n)}) \ge E_1(\widehat{H}_{\Lambda_l(\mathbf{u})}^{(n)})$ and 
\[ 
\prob\{E_1(H_{\Lambda_l(\mathbf{u})}^{(n)}(\omega)) \le nq_- + nb n_l^{-2}\} \le \prob\{E_1(\widehat{H}_{\Lambda_l(\mathbf{u})}^{(n)}(\omega)) \le nq_- + nb n_l^{-2}\} \, .
\]
But $\widehat{H}_{\Lambda_l(\mathbf{u})}^{(n)} = H_{\Lambda_l(u_1)}^{(1)} \otimes I^{n-1} + \ldots + I^{k-1} \otimes H_{\Lambda_l(u_k)}^{(1)} \otimes I^{n-k} + \ldots + I^{n-1} \otimes H_{\Lambda_l(u_n)}^{(1)}$, where $H_{\Lambda_l}^{(1)} = \widehat{H}^{(1)}_{\Lambda_l} :(f_e) \mapsto (-f_e''+\omega_e f_e)$. Thus $E_1(\widehat{H}_{\Lambda_l(\mathbf{u})}^{(n)}) = \sum_{j=1}^n E_1(H_{\Lambda_l(u_j)}^{(1)}) \ge n E_1(H_{\Lambda_l(u_{j_0})}^{(1)})$, where $E_1(H_{\Lambda_l(u_{j_0})}^{(1)}) := \min_{1 \le j \le n} E_1(H_{\Lambda_l(u_j)}^{(1)})$. Hence
\[ 
\prob\{E_1(H_{\Lambda_l(\mathbf{u})}^{(n)}(\omega)) \le nq_- + nb n_l^{-2}\} \le \prob\{E_1(H_{\Lambda_l(u_{j_0})}^{(1)}(\omega)) \le q_- + b n_l^{-2}\}. 
\]
The existence of $b$ and $\gamma$ now follows from Theorem~\ref{thm:Lif1}. So take these $b$, $\gamma$, and given $\xi >0$, $\beta \in (0,1)$, choose $L^{\ast}(n,d,\beta,\xi)$ such that for $L \ge L^{\ast}$, we have
\begin{equation}
6^{nd} b^{-n/2} d^n L^{nd+ \frac{n(\beta-1)}{2}} e^{-\gamma 2^{-d} (bL^{1-\beta})^{1/2}} \le (2L)^{-\xi} \, .    \label{eq:ils2}
\end{equation}
Let $L_{\ast} := \max_{1 \le n \le N} L^{\ast}(n,d,\beta,\xi)$. Given $L \ge L_{\ast}$, let $l := \lfloor \frac{1}{2} (\frac{bL^{1-\beta}}{d^2})^{1/2d} \rfloor$ and choose $L\le L_0 \le 2L$ such that $L_0 = rl$ for some $r \in \N$. Then $\Lambda_{L_0}^{(n)} := \Lambda_{L_0}^{(n)}(\mathbf{u})$ may be divided into $M=L_0^{nd}l^{-nd}$ disjoint cubes $\Lambda_l^k$. Since $H_{\Lambda_{L_0}}^{(n)} \ge \mathop \oplus_k H_{\Lambda_l^k}^{(n)}$, we get $E_1(H_{\Lambda_{L_0}}^{(n)}) \ge \min_k E_1(H_{\Lambda_l^k}^{(n)})$. But $n_l^2 \le d^2 (2l)^{2d} \le bL^{1-\beta}$ and thus $L^{\beta-1} \le b n_l^{-2}$. So using (\ref{eq:firstpart}) we get
\begin{align*}
\prob \{ \dist(\sigma(H_{\Lambda_{L_0}}^{(n)}),nq_-) \le L_0^{\beta -1} \} & \le \prob \{ E_1(H_{\Lambda_{L_0}}^{(n)}) -nq_- \le L^{\beta -1} \} \\
& \le \prob \{ E_1(H_{\Lambda_l^k}^{(n)}) -nq_- \le b n_l^{-2} \text{ for some } k \} \le M e^{- \gamma n_l}.
\end{align*}
Noting that $M \le (2L)^{nd} l^{-nd} \le (2L)^{nd} (\frac{1}{3} (\frac{bL^{1-\beta}}{d^2})^{1/2d})^{-nd} = 6^{nd} b^{-n/2} d^n L^{nd+ \frac{n(\beta-1)}{2}}$ and $n_l \ge d(2l-1)^d \ge d(l+1)^d \ge 2^{-d} (bL^{1-\beta})^{1/2}$, then using (\ref{eq:ils2}) we may bound the RHS by $(2L)^{-\xi} \le L_0^{-\xi}$, which completes the proof.
\end{proof}

\begin{defa} 
Let $E \in \R$, $m>0$ and $\omega \in \Omega$. A cube $\Lambda_L^{(n)}(\mathbf{u})$ is said to be \emph{$(E,m)$-Non Singular} ($(E,m)$-NS) if $E \in \rho(H_{\Lambda_L(\mathbf{u})}^{(n)}(\omega))$ and
\[ 
\max_{\mathbf{y} \in \mathbf{B}_L^{\out}(\mathbf{u})} \|G_{\Lambda_L^{(n)}(\mathbf{u})}(\mathbf{u},\mathbf{y};E) \| \le e^{- mL} \, ,
\]
otherwise it is said to be \emph{$(E,m)$-Singular} ($(E,m)$-S).
\end{defa}

\begin{coro}[ILS estimate]      \label{cor:ILS}
For any $p>0$ and $\beta \in (0,1)$, we may find $L_0 = L_0(N,d,p,\beta)$ as large as necessary such that for $\varepsilon_0 = \frac{L_0^{\beta-1}}{2}$, $I_n=[nq_- - \frac{1}{2},nq_-+\varepsilon_0]$, $m_{L_0} = \frac{L_0^{(\beta-1)/2}}{3}$ and any cube $\Lambda_{L_0}^{(n)}(\mathbf{u})$, we have
\begin{equation} 
\prob\{ \exists E \in I_n: \Lambda_{L_0}^{(n)}(\mathbf{u}) \text{ is } (E,m_{L_0})\text{-S } \} \le L_0^{-2p}.  \label{eq:ils}
\end{equation}
\end{coro}
\begin{proof} 
Given $2 p \equiv \xi >0$, $\beta \in (0,1)$, we find $L_0$ as large as needed satisfying Theorem~\ref{thm:Lif2}. Now let $\mathbf{y} \in \mathbf{B}_{L_0}^{\out}(\mathbf{u})$, so $L_0-8 \le \dist(\mathbf{C}(\textbf{u}),\mathbf{C}(\textbf{y})) \le L_0$. Let $\varepsilon_0 = \frac{L_0^{\beta-1}}{2}$ and suppose $s_{\omega} - nq_- >L_0^{\beta - 1}$, where $s_{\omega} :=\inf \sigma(H_{\Lambda_{L_0}(\mathbf{u})}^{(n)}(\omega))$. Then every $E \in I_n$ satisfies $E < s_{\omega}$ and $\eta:=s_{\omega} - E \ge \frac{L_0^{\beta-1}}{2}$. So by Theorem~\ref{thm:CT2},
\begin{align*} 
\|G_{\Lambda_{L_0}^{(n)}(\mathbf{u})}(\mathbf{u},\textbf{y};E) \| & \le \sqrt{\frac{\pi}{2}} \Big( \frac{L_0^{1/2}}{(L_0^{\beta-1}/2)^{3/4}} + \frac{3}{8(L_0-8)^{1/2}(L_0^{\beta-1}/2)^{5/4}} \Big) e^{-(L_0-8) \sqrt{\frac{L_0^{\beta-1}}{2}}} \\
& \le e^{-m_{L_0}L_0} 
\end{align*}
for $L_0$ large enough. Hence $\Lambda_{L_0}^{(n)}(\mathbf{u})$ is $(E,m_{L_0})$-NS. We thus showed that
\[
\prob\{ \exists E \in I_n: \Lambda_{L_0}^{(n)}(\mathbf{u}) \text{ is } (E,m_{L_0})\text{-S} \} \le \prob\{ s_{\omega}-nq_- \le L_0^{\beta-1} \} \, .
\]
The claim follows by Theorem~\ref{thm:Lif2}, since $s_{\omega}-nq_- = \dist(\sigma(H_{\Lambda_{L_0}(\mathbf{u})}^{(n)}(\omega)),nq_-)$.
\end{proof}

\section{Multi-Particle Multiscale Analysis}    \label{sec:induction}
We now introduce a multi-particle multiscale analysis following the main ideas of \cite{BCS}, providing modifications as necessary. Throughout this section we fix
\[ 
\alpha = 3/2,\qquad \beta = 1/2 \, ,
\]
and denote $K(n):=n^n$ for $1 \le n \le N$. We also denote by $\lfloor x \rfloor$ the integer part of $x \in \R$ and assume that
\[
\mu \text{ is H\"older continuous.}
\]

\begin{defa}
We say that a cube $\Lambda_L^{(n)}(\mathbf{u})$ is \emph{$E$-Non Resonant} ($E$-NR) if
\[ 
\dist(\sigma(H_{\Lambda_L(\mathbf{u})}^{(n)}),E) \ge e^{-L^{\beta}}. 
\]
We say it is \emph{$E$-Completely Non-Resonant} ($E$-CNR)
if any cube $\Lambda_{\ell}^{(n)} \subseteq \Lambda_L^{(n)}(\mathbf{u})$ with $\ell \in \N^{\ast}$, $L^{1/ \alpha} \le \ell \le L$ is $E$-NR. In this case, $\Lambda_L^{(n)}(\mathbf{u})$ is thus $E$-NR in particular.
\end{defa}

\begin{defa}
Let $l \in \N^{\ast}$, $L= \lfloor l^{\alpha} \rfloor +1$ and $J \in \N$. We say that a cube $\Lambda_L^{(n)}(\mathbf{x})$ is \emph{$(E,m_l,J)$-good} if it contains at most $J$ pairwise separable cubes $\Lambda_l^{(n)}$ which are $(E,m_l)$-S. Otherwise, we say it is \emph{$(E,m_l,J)$-bad}; in this case, there are at least $J+1$ separable cubes $\Lambda_l^{(n)}$ which are $(E,m_l)$-S.
\end{defa}

We start by adapting \cite[Lemma 4.2]{DK} to $n$-graphs. For this, we first prove the following geometric argument: given a collection of cubes, either they are already pairwise disjoint, or we can construct larger cubes around each cluster, such that the larger cubes are disjoint. For technical reasons, we consider $\epsilon$-enlargements of the cubes, with $\epsilon=7$.

\begin{lem}      \label{lem:geo}
Given $k$ cubes $\Lambda_L^{(n)}(\mathbf{u}_{(r)})$, $r=1,\ldots,k$, there exists $\tilde{k} \le k$ disjoint cubes $\Lambda_{l_j}^{(n)}$, $j=1,\ldots,\tilde{k}$ such that $\bigcup_{j=1}^{\tilde{k}} \Lambda_{l_j-7}^{(n)} \supseteq \bigcup_{r=1}^k \Lambda_L^{(n)}(\mathbf{u}_{(r)})$, $l_j = n_j(L+7)$ for some $n_j \in \N^{\ast}$ and $\sum_{j=1}^{\tilde{k}} l_j = k(L+7)$.
\end{lem}
\begin{proof}
If the cubes $\Lambda_{L+7}^{(n)}(\mathbf{u}_{(r)})$ are disjoint, we put $\Lambda_{l_j}^{(n)} = \Lambda_{L+7}^{(n)}(\mathbf{u}_{(j)})$. Otherwise, divide $\bigcup_{r=1}^k \Lambda_{L+7}^{(n)}(\mathbf{u}_{(r)})$ into $k'$ connected components with $1 \le k' \le k$ and order them. If the $i$-th component contains $n_i$ cubes, find a cube $\Lambda'_{l_i}$ containing it with $l_i:=n_i(L+7)$. If these $k'$ cubes are disjoint, then we are done. If not, divide them into $k''$ connected components and again find cubes $\Lambda''_{l_i}$ around each component with $l_i = n_i(L+7)$, where $n_i$ is the number of the original cubes $\Lambda_{L+7}^{(n)}(\mathbf{u}_{(r)})$ which this component contains. Repeating this procedure we finally obtain the assertion.
\end{proof}

\begin{lem}      \label{lem:DKn}
Let $l \in \N^{\ast}$, $J\in \N$, $m_l>\frac{8NJK(N)}{l^{1-\beta}}$, $E_+ \in \R$ and $E \le E_+$. Let $L= \lfloor l^{\alpha} \rfloor +1$ and suppose that $\Lambda_L^{(n)}(\mathbf{x})$ is $E$-CNR and $(E,m_l,J)$-good. Then there exists $l^{\ast}=l^{\ast}(E_+,N,d,J,q_-,r_0)$ such that, if $l > l^{\ast}$, then $\Lambda_L^{(n)}(\mathbf{x})$ is $(E,m_L+L^{\beta -1})$-NS, where
\[ 
m_L := m_l- \biggl( \frac{16NJK(N)}{l^{\alpha-1}}m_l + \frac{3}{l^{\alpha(1-\beta)}} \biggr) > \frac{8NJK(N)}{L^{1-\beta}} \, .
\]
\end{lem}

\begin{proof}
By hypothesis there are at most $J$ pairwise separable cubes $\Lambda_l^{(n)}(\mathbf{u}_{(s)}) \subset \Lambda_L^{(n)}(\mathbf{x})$ which are $(E,m_l)$-S. Applying Lemma~\ref{lem:separable} to each of them, we may find $JK(n)$ cubes $\Lambda_{r_{n,l}}^{(n)}(\mathbf{z}^{(k)})$ such that if $\mathbf{v} \notin \bigcup_{k=1}^{JK(n)} \Lambda_{r_{n,l}}^{(n)}(\mathbf{z}^{(k)})$, then $\Lambda_l^{(n)}(\mathbf{v})$ is separable from all the $\Lambda_l^{(n)}(\mathbf{u}_{(s)})$. Now applying Lemma~\ref{lem:geo} to the $JK(n)$ cubes $\Lambda_{r_{n,l}}^{(n)}(\mathbf{z}^{(k)})$, we may construct disjoint cubes $\Lambda_{l_j}^{(n)}$ such that $\bigcup_j \Lambda_{l_j-7}^{(n)} \supseteq \bigcup_k \Lambda_{r_{n,l}}(\mathbf{z}^{(k)})$, $l_j = n_j(r_{n,l}+7)$ for some $n_j \in \N^{\ast}$ and $\sum l_j \le JK(n)(r_{n,l}+7) \le JK(N)(r_{N,l}+7) =: l_{N,J}$. Thus, $\Lambda_l^{(n)}(\mathbf{v})$ is $(E,m_l)$-NS whenever $\mathbf{v} \in \Lambda_{L-l}^{(n)}(\mathbf{x}) \setminus \bigcup_j \Lambda_{l_j-7}^{(n)}$. 

We first assume all the ``bad cubes'' $\Lambda_{l_j}^{(n)}$ are inside $\Lambda_{L-l-7}^{(n)}(\mathbf{x})$. Note that if $\mathbf{v}\in\mathbf{B}_{L-l}^{(n)}(\mathbf{x})$ satisfies $\mathbf{v} \in\mathbf{B}_{l_j}^{\out}=\mathbf{B}_{l_j}^{(n)}\setminus\mathbf{B}_{l_j-6}^{(n)}$ for some $j$, then $\Lambda_l^{(n)}(\mathbf{v})$ is $(E,m_l)\text{-NS}$ since $\mathbf{v} \notin \Lambda_{l_j-7}^{(n)}$ and $\mathbf{v}\notin\Lambda_{l_r-7}^{(n)}$ for $r \neq j$ (because $\mathbf{v} \in \Lambda_{l_j}^{(n)}$ and $\Lambda_{l_j}^{(n)}$ is disjoint from the other $\Lambda_{l_r}^{(n)}$).

Now fix $\mathbf{y} \in \mathbf{B}_L^{\out}(\mathbf{x})$ and let $\mathbf{u} \in \mathbf{B}_{L-l-7}^{(n)}(\mathbf{x})$. We have 2 cases:
\begin{enumerate}[(a)]
\item $\Lambda_l^{(n)}(\mathbf{u})$ is $(E,m_l)$-NS. Then applying \textsf{(GRI.2)} to $\Lambda_l^{(n)}:= \Lambda_l^{(n)}(\mathbf{u})$,
\begin{align*}
\|G_{\Lambda_L^{(n)}(\mathbf{x})}(\mathbf{u},\mathbf{y};E) \| & \le C \cdot |\mathbf{B}_l^{\out}|^2 \max_{\mathbf{k} \in \mathbf{B}_l^{\out}} \| G_{\Lambda_l^{(n)}}(\mathbf{u},\mathbf{k};E) \| \max_{\mathbf{k}' \in \mathbf{B}_l^{\out}} \| G_{\Lambda_L^{(n)}(\mathbf{x})}(\mathbf{k}',\mathbf{y};E) \| \\
& \le C_1(2l-1)^{2(nd-1)} e^{-m_l l}\| G_{\Lambda_L^{(n)}(\mathbf{x})}(\mathbf{w}_1,\mathbf{y};E)\|
\end{align*}
for some $\mathbf{w}_1 \in \mathbf{B}_l^{\out}(\mathbf{u})$.
\item $\Lambda_l^{(n)}(\mathbf{u})$ is $(E,m_l)$-S. In this case, $\mathbf{u} \in \Lambda_{l_j-7}^{(n)}$ for some $j$, so applying \textsf{(GRI.2)} to $\Lambda_{l_j}^{(n)}$,
\begin{align*}
\| G_{\Lambda_L^{(n)}(\mathbf{x})}(\mathbf{u},\mathbf{y};E) \| & \le C \cdot |\mathbf{B}_{l_j}^{\out}|^2 \max_{\mathbf{k} \in \mathbf{B}_{l_j}^{\out}} \| G_{\Lambda_{l_j}^{(n)}}(\mathbf{u},\mathbf{k};E) \| \max_{\mathbf{k}' \in \mathbf{B}_{l_j}^{\out}} \| G_{\Lambda_L^{(n)}(\mathbf{x})}(\mathbf{k}',\mathbf{y};E) \| \\
& \le C_2 (2l_{N,J}-1)^{2(nd-1)} e^{l_{N,J}^{\beta}} \| G_{\Lambda_L^{(n)}(\mathbf{x})}(\mathbf{w},\mathbf{y};E) \|
\end{align*}
for some $\mathbf{w} \in \mathbf{B}_{l_j}^{\out}$ because $\Lambda_L^{(n)}(\mathbf{x})$ is $E$-CNR. But then $\Lambda_l^{(n)}(\mathbf{w})$ is $(E,m_l)\text{-NS}$, so applying \textsf{(GRI.2)} once more we get
\[ 
\| G_{\Lambda_L^{(n)}(\mathbf{x})}(\mathbf{u},\mathbf{y};E) \| \le C_3 ((2l_{N,J}-1)(2l-1))^{2(nd-1)} e^{l_{N,J}^{\beta} - m_l l} \| G_{\Lambda_L^{(n)}(\mathbf{x})}(\mathbf{w}_1,\mathbf{y};E) \| 
\]
for some $\mathbf{w}_1 \in \mathbf{B}_l^{\out}(\mathbf{w})$. Hence
\[ 
\| G_{\Lambda_L^{(n)}(\mathbf{x})}(\mathbf{u},\mathbf{y};E) \| \le e^{-m_l' l} \| G_{\Lambda_L^{(n)}(\mathbf{x})}(\mathbf{w}_1,\mathbf{y};E) \|, 
\]
where
\[ 
m_l' = m_l - l^{-1} \big\{ l_{N,J}^{\beta} + 2(nd-1) \log((2l_{N,J} - 1)(2l-1)) + \log C_3 \big\}  >0 
\]
because for large $l$,
\[ 
l_{N,J} = JK(N)(2(4N-3)l + 4(N-1)r_0 + 7) \le (8N-5)JK(N)l 
\]
so that $m_l' \ge m_l - \frac{8NJK(N)}{l^{1-\beta}} >0$ for $l$ large enough.
\end{enumerate}
Hence starting at $\mathbf{u} = \mathbf{w}_0 :=\mathbf{x}$, we may iterate the procedure $p$ times as long as $\mathbf{w}_{p-1} \in \mathbf{B}_{L-l-7}^{(n)}(\mathbf{x})$. If (a) occurs $n_+$ times and (b) occurs $n_0 = p - n_+$ times, we obtain
\[ 
\| G_{\Lambda_L^{(n)}(\mathbf{x})}(\mathbf{x},\mathbf{y};E) \| \le \big(C_1(2l-1)^{2(nd-1)} e^{-m_l l} \big)^{n_+} e^{-n_0 m_l' l} \| G_{\Lambda_L^{(n)}(\mathbf{x})}(\mathbf{w}_p,\mathbf{y};E) \|. 
\]
Now $\Lambda_L^{(n)}(\mathbf{x})$ is $E\text{-NR}$ and $e^{-n_0 m_l' l} \le 1$ since $m'_l > 0$. Hence
\[ 
\| G_{\Lambda_L^{(n)}(\mathbf{x})}(\mathbf{x},\mathbf{y};E) \| \le \big(C_1(2l-1)^{2(nd-1)} e^{-m_l l} \big)^{n_+} e^{L^{\beta}} \le e^{-(\tilde{m}_L+L^{\beta -1}) L}, 
\]
where 
\[ 
(\tilde{m}_L+L^{\beta -1}) L = - n_+ \big( \log C_1 + 2(nd-1) \log(2l-1)-m_l l \big) -L^{\beta}.
\]

In case (a), $\mathbf{w}_k \in \mathbf{B}_l^{\out}(\mathbf{w}_{k-1})$, so each step cuts a length between $l-6$ and $l-1$. We thus have $\lfloor \frac{L-2l_{N,J}-l-7}{l-1} \rfloor \le n_+ \le \lfloor \frac{L-l-7}{l-6} \rfloor $. Indeed, the lower bound represents the worst scenario in which the iteration met all the bad cubes in its way, a total length of $2l_{N,J}$. The upper bound occurs when it meets no bad cube. In particular, we have $\frac{L- 2l_{N,J} - l -7}{l} - 1 \le n_+ \le \frac{L}{l-6}$, so we get
\[ 
\tilde{m}_LL \ge m_l (L-2l_{N,J}-2l-7) - \frac{L}{l-6} (\log C_1 + 2(nd-1) \log(2l-1) ) -2L^{\beta}. 
\]
But $2l_{N,J}+2l+7 = ((16N-12)JK(N)+2)l + C(r_0,N,J) \le 16NJK(N)l$. Hence
\begin{align*}
\tilde{m}_L L & \ge m_l L - 16NJK(N)m_ll - \frac{2ndL\log(2l-1)}{l-6} -2L^{\beta} \\
& \ge m_l L - 16NJK(N)m_ll - \frac{L}{l^{\alpha(1-\beta)}} -2L^{\beta}
\end{align*}
for large $l$, because $\alpha(1-\beta)=3/4<1$. Noting that $L \ge l^{\alpha}$, we finally get
\[
\tilde{m}_L \ge m_l - \frac{16NJK(N)}{l^{\alpha - 1}} m_l - \frac{3}{l^{\alpha(1-\beta)}} = m_L \, .
\]
Thus, $\| G_{\Lambda_L^{(n)}(\mathbf{x})}(\mathbf{x},\mathbf{y};E) \| \le e^{-(\tilde{m}_L+L^{\beta -1}) L} \le e^{-(m_L+L^{\beta -1}) L}$ and $\Lambda_L^{(n)}(\mathbf{x})$ is $(E,m_L+L^{\beta -1})$-NS. For the lower bound on $m_L$, note that for large $l$,
\begin{equation}
\Big(1-\frac{16NJK(N)}{l^{\alpha - 1}}\Big) m_l \ge \frac{1}{2} m_l \ge \frac{4NJK(N)}{l^{1-\beta}} > \frac{3}{l^{\alpha(1-\beta)}} + \frac{8NJK(N)}{L^{1-\beta}} \, .    \label{eq:lb}
\end{equation}

Finally, if a bad cube lies completely outside $F:= \Lambda_{L-l-7}^{(n)}(\mathbf{x})$, the situation is obviously better. If a bad cube is not contained in $F$ but intersects $F$, we stop the iteration if we reach this bad cube. Then again the situation is better (because here only part of the length $2l_j$ of this cube is counted as bad).
\end{proof}

We define for $n \ge 2$,
\[ 
p_n := \frac{p_{n-1}}{\alpha^2(1+\theta)} - \frac{(2n-1)d}{2\alpha} - nd - 1,
\]
where $\theta := \frac{1}{2p_1}$. We then choose $p_1$ sufficiently large to make sure that
\[ 
p_N \ge 3Nd+1.
\]
In particular, $0 < \theta < 1$.

Fix 
\[
E_+ := \max_{1 \le n \le N} (nq_- +1), \qquad J=6 \, , 
\]
and let $l^{\ast}$ be as in Lemma~\ref{lem:DKn}. Then by Corollary~\ref{cor:ILS}, we may find $L_0 > l^{\ast}$ as large as necessary such that (\ref{eq:ils}) is satisfied for all $1 \le n \le N$, with $\varepsilon_0 = \frac{L_0^{\beta-1}}{2}$, $I_n = [nq_- - \frac{1}{2}, nq_- + \varepsilon_0 ]$, $m_{L_0} = \frac{1}{3L_0^{(1-\beta)/2}}$ and $p:= p_1$. We then define the sequences
\[ 
L_{k+1} := \lfloor L_k^{\alpha} \rfloor +1 \, ,
\]
\[
m_{L_{k+1}} := m_{L_k} - \biggl( \frac{96NK(N)}{L_k^{\alpha-1}}m_{L_k} + \frac{3}{L_k^{\alpha(1-\beta)}} \biggr) \, .
\]
Note that $m_{L_k}>\frac{48NK(N)}{L_k^{1-\beta}}$. Indeed, $m_{L_0} = \frac{1}{3L_0^{(1-\beta)/2}} > \frac{48NK(N)}{L_0^{1-\beta}}$ since $L_0$ is large, hence $m_{L_k}>\frac{48NK(N)}{L_k^{1-\beta}}$ by induction, using (\ref{eq:lb}). We now introduce the property
\[
\left\{\begin{array}{l}
\text{For all pairs of separable cubes }\Lambda_{L_k}^{(n)}(\mathbf{u})\text{ and } \Lambda_{L_k}^{(n)}(\mathbf{v}):\\ 
\prob \{ \exists E \in I_n : \Lambda_{L_k}^{(n)}(\mathbf{u})\text{ and }\Lambda_{L_k}^{(n)}(\mathbf{v})\text{ are }(E,m_{L_k})\text{-S} \} \le L_k^{-2p_n(1+\theta)^k}. \\ 
\end{array} \right. \tag{\textsf{DS}$\,:n,k,m_{L_k},I_n$}
\]

The term $(1+\theta)^k$ in the exponent was introduced in \cite{BCS} and is new in comparison with the usual multiscale analysis. While it complicates a few estimates, it has a powerful advantage, namely it allows to prove dynamical localization of any order $s$ in $I_N$, with $\varepsilon_0$ independent of $s$. This result (among others) was previously obtained for single-particle systems in the continuum using the bootstrap multiscale analysis of \cite{GK}.

To prove this property, we shall need Lemma~\ref{lem:DKn} and the following Wegner bound:
\[
\left\{\begin{array}{l}
\text{For all pairs of separable cubes }\Lambda_{L_k}^{(n)}(\mathbf{u})\text{ and } \Lambda_{L_k}^{(n)}(\mathbf{v}):\\ 
\prob\{\exists E\in I_n : \Lambda_{L_k}^{(n)}(\mathbf{u})\text{ and }\Lambda_{L_k}^{(n)}(\mathbf{v})\text{ are not }E\text{-CNR}\} \le \frac{1}{4} L_k^{-2p_1(1+\theta)^k} \, . \\
\end{array}\right. \tag{\textsf{W2}$\,:n,k,I_n$}
\]

\begin{lem}      \label{lem:regu}
The property \emph{(\textsf{W2}$\,:n,k,I_n)$} holds for all $k \ge 0$ and $1 \le n \le N$.
\end{lem}
\begin{proof}
Let $\Lambda_{L_k}^{(n)}(\mathbf{u})$ and $\Lambda_{L_k}^{(n)}(\mathbf{v})$ be separable. If $\Lambda_{\ell_1}^{(n)} \subseteq \Lambda_{L_k}^{(n)}(\mathbf{u})$ and $\Lambda_{\ell_2}^{(n)} \subseteq \Lambda_{L_k}^{(n)}(\mathbf{v})$, then $\Lambda_{\ell_1}^{(n)}$ and $\Lambda_{\ell_2}^{(n)}$ are pre-separable. Hence by Theorem~\ref{thm:Wegner2},
\begin{align*}
& \prob\{ \exists E \in I_n : \dist(\sigma(H_{\Lambda_{\ell_1}}^{(n)}),E) < \varepsilon \text{ and } \dist(\sigma(H_{\Lambda_{\ell_2}}^{(n)}),E) < \varepsilon \} \\
& \quad \le \prob\{ \dist(\sigma_{J_n}(H_{\Lambda_{\ell_1}}^{(n)}), \sigma_{J_n}(H_{\Lambda_{\ell_2}}^{(n)})) < 2 \varepsilon \} \le  C (2L_k)^{2nd+d} s(\mu,4\varepsilon) \, ,
\end{align*}
where $J_n=[nq_--\frac{1}{2}-\varepsilon,nq_-+\varepsilon_0+\varepsilon]$. Bounding the number of cubes in $\Lambda_{L_k}^{(n)}$ by $| \mathbf{B}_{L_k}^{(n)}| \le (2L_k)^{nd}$ and the number of $\ell \in \N^{\ast}$ satisfying $L_k^{1/\alpha} \le \ell \le L_k$ by $L_k$, we get for $\varepsilon := \max(e^{-\ell_1^{\beta}},e^{-\ell_2^{\beta}}) \le e^{-L_k^{\beta/\alpha}}$,
\[
\prob\{\exists E\in I_n : \Lambda_{L_k}^{(n)}(\mathbf{u})\text{ and }\Lambda_{L_k}^{(n)}(\mathbf{v})\text{ are not }E\text{-CNR}\} \le C (2L_k)^{4nd+d+2} s(\mu,4e^{-L_k^{\beta/\alpha}}) \, .
\]
Since $\mu$ is H\"older continuous, there exist $c_{\mu}$ and $b>0$ such that
\[
C (2L_k)^{4nd+d+2} s(\mu,4e^{-L_k^{\beta/\alpha}}) \le C c_{\mu} (2L_k)^{4nd+d+2} (4e^{-L_k^{\beta/\alpha}})^b \le e^{-L_k^{\zeta}}
\]
for some $\zeta>0$, since $L_0$ is large. Now for any $k \ge 0$,
\[
-\log(1/4)+2p_1 (1+\theta)^k \log L_k \le \log(4)+ 2p_1 2^k \log L_k \le C_{N,d} 2^k \alpha^k \log L_0 \le L_0^{\alpha^k \zeta}
\]
since $\alpha^k \ge \frac{ \log C}{ \zeta \log L_0} + k \frac{\log 2 \alpha }{ \zeta \log L_0} + \frac{ \log \log L_0}{ \zeta \log L_0}$ for large $L_0$, independently of $k$. But $L_0^{\alpha^k \zeta} \le L_k^{\zeta}$. We thus showed that $e^{-L_k^{\zeta}} \le \exp(\log(1/4)-2p_1 (1+\theta)^k \log L_k) = \frac{1}{4} L_k^{-2p_1(1+\theta)^k}$.
\end{proof}

\subsection{Single-particle case}
For $n=1$, separable cubes are just disjoint cubes; see Definition~\ref{def:sep}.

\begin{thm}       \label{thm:1-part}
\emph{(\textsf{DS}$\,:1,k,m_{L_k},I_1$)} implies \emph{(\textsf{DS}$\,:1,k+1,m_{L_{k+1}},I_1$)}.
\end{thm}
\begin{proof}
Put $L=L_{k+1}$, $l=L_k$ and let $\Lambda_L^{(1)}(u)$ and $\Lambda_L^{(1)}(v)$ be a pair of disjoint cubes. Since $\varepsilon_0<1$, any $E \in I_1$ satisfies $E\le E_+=\max_n(nq_-+1)$, so applying Lemma~\ref{lem:DKn} with $J=6$, noting that $l>l^{\ast}$ because $L_0 > l^{\ast}$, we have
\begin{align*}
& \prob\{ \exists E \in I_1 : \Lambda_L^{(1)}(u) \text{ and } \Lambda_L^{(1)}(v) \text{ are } (E,m_L)\text{-S} \} \\
& \quad \le 3 \max_{x=u,v} \prob \{ \exists E \in I_1 : \Lambda_L^{(1)}(x) \text{ is } (E,m_l,6)\text{-bad} \} \\
& \qquad + \prob\{ \exists E \in I_1 : \Lambda_L^{(1)}(u) \text{ and } \Lambda_L^{(1)}(v) \text{ are not } E\text{-CNR} \} \, ,
\end{align*}
since an ($E,m_L$)-S cube is a fortiori ($E,m_L + L^{\beta-1}$)-S. Now by (\textsf{W2}$\,:1,k+1,I_1$),
\[
\prob\{ \exists E \in I_1 : \Lambda_L^{(1)}(u) \text{ and } \Lambda_L^{(1)}(v) \text{ are not } E\text{-CNR} \} \le \frac{1}{4} L^{-2p_1(1+\theta)^{k+1}} \, .
\]
Moreover, given $J \in 2\N^{\ast}$, if $\Lambda_L^{(1)}$ is $(E,m_l,J-1)$-bad then it contains at least $J$ separable cubes which are $(E,m_l)$-S. Since $n=1$, Hamiltonians on disjoint cubes are independent. So by grouping these $J$ cubes two by two, using (\textsf{DS}$\,:1,k,m_{L_k},I_1$) and bounding the number of pairs of cubes in $\Lambda_L^{(1)}$ by $| \mathbf{B}_L^{(1)} |^2 \le (2L)^{2d}$, we get
\begin{align}
\prob \{ \exists E \in I_1 : \Lambda_L^{(1)} \text{ is } (E,m_l,J)\text{-bad} \} & \le \prob \{ \exists E \in I_1 : \Lambda_L^{(1)} \text{ is } (E,m_l,J-1)\text{-bad} \}  \label{eq:useless} \\
& \le \prob \{ \exists E \in I_1 : \Lambda_L^{(1)} \text{ is } (E,m_l,1)\text{-bad}  \}^{J/2} \notag \\
& \le ((2L)^{2d} l^{-2p_1(1+\theta)^k})^{J/2} \notag \\
& \le cL^{(d-\frac{p_1(1+\theta)^k}{\alpha})J} \le cL^{(\frac{p_1 -1}{3}-\frac{2p_1(1+\theta)^k}{3})J} \notag
\end{align}
because $\alpha = 3/2$ and $d \le \frac{p_N -1}{3} \le \frac{p_1 -1}{3}$. Now
\[
cL^{\frac{Jp_1}{3}(1-2(1+\theta)^k-\frac{1}{p_1})} \le \frac{1}{4} L^{\frac{Jp_1}{3}(1-2(1+\theta)^k-\theta)}
\]
and since $(1-\theta) \le (1-\theta)(1+\theta)^k=(1+\theta)^k-\theta(1+\theta)^k$, we have
\begin{equation}
1 - 2(1+\theta)^k - \theta \le - (1+\theta)^k - \theta(1+\theta)^k = -(1+\theta)^{k+1} \, .  \label{eq:ms1}
\end{equation}
Hence,
\[
\prob \{ \exists E \in I_1 : \Lambda_L^{(1)} \text{ is } (E,m_l,J)\text{-bad} \} \le \frac{1}{4} L^{-\frac{Jp_1}{3}(1+\theta)^{k+1}}.
\]
The claim now follows by taking $J=6$.
\end{proof}

\subsection{Multi-particle case: Strategy}             \label{sec:stratn}
The deduction of (\textsf{DS}$\,:1,k+1,m_{L_{k+1}},I_1$) from (\textsf{DS}$\,:1,k,m_{L_k},I_1$) was fairly simple. Once $n\ge 2$ however, we face a difficulty when trying to estimate the probability that a cube is $(E,m_{L_k},J)$-bad. Indeed, Hamiltonians on separable sub-cubes are not independent, so we can no longer multiply the probabilities as in the previous subsection.

To overcome this, we reason as follows: if a cube $\Lambda_{L_{k+1}}^{(N)}$ is $(E,m_{L_k},J)$-bad, then it contains at least $J+1$ pairwise separable cubes $\Lambda_{L_k}^{(N)}$ which are $(E,m_{L_k})$-S. Hence, either it contains $2$ separable $(E,m_{L_k})$-S PI cubes, or it contains at least $J$ separable $(E,m_{L_k})$-S FI cubes. Now separable FI cubes are completely separated by Lemma~\ref{lem:FI}, so taking $J=6$, we can again multiply the probabilities. The main difficulty is in assessing the probability that a cube contains $2$ separable $(E,m_{L_k})$-S PI cubes. The idea is as follows: on PI cubes, the interaction potential decouples by Lemma~\ref{lem:PI}, so the corresponding Hamiltonians take the form $H_{\Lambda_{L_k}(\mathbf{u})}^{(N)} = H_{\Lambda_{L_k}(u_{\mathcal{J}})}^{(n')} \otimes I + I \otimes H_{\Lambda_{L_k}(u_{\mathcal{J}^c})}^{(n'')}$, where $n', n'' <N$. Now using the new resolvent inequality (\textsf{GRI.3}), we may deduce that $\Lambda_{L_k}^{(N)}(\mathbf{u})$ is non-singular if both projections $\Lambda_{L_k}^{(n')}(u_{\mathcal{J}})$ and $\Lambda_{L_k}^{(n'')}(u_{\mathcal{J}^c})$ are non-singular for an array of energies. To show that both projections are indeed non-singular, we show that they cannot contain a lot of bad sub-cubes $\Lambda_{L_{k-1}}^{(n)}$, $n=n',n''$.

Notice that in the above scheme, we reduced the decay problem on PI $N$-cubes to that on $n$-cubes for $n < N$, which indicates that an induction on $n$ will be performed. Also notice that unlike single-particle systems, here we will need good decay bounds on both orders $k-1$ and $k$ to finally deduce the decay for $k+1$.

\subsection{Pairs of PI cubes}
We assume through this subsection that $2 \le n \le N$.

Recall that if $\Lambda_{L_k}^{(n)}(\mathbf{u})$ is a PI cube, then it is $\mathcal{J}$-decomposable for some $\mathcal{J}$ by Lemma~\ref{lem:PI}. We may thus denote it $\Lambda_{L_k}^{(n)}(\mathbf{u}) = \Lambda_{L_k}^{(n')}(u_{\mathcal{J}}) \times \Lambda_{L_k}^{(n'')}(u_{\mathcal{J}^c})$, where $n' = \# \mathcal{J}$ and $n'' = n-n'$. We also denote by $\Sigma'$ and $\Sigma''$ the spectra of $H_{\Lambda_{L_k}(u_{\mathcal{J}})}^{(n')}$ and $H_{\Lambda_{L_k}(u_{\mathcal{J}^c})}^{(n'')}$, respectively.

\begin{defa}          \label{def:tun}
Let $\Lambda_{L_k}^{(n)}(\mathbf{u}) = \Lambda_{L_k}^{(n')}(u_{\mathcal{J}}) \times \Lambda_{L_k}^{(n'')}(u_{\mathcal{J}^c})$ be a PI cube, $k \ge 1$. We say that $\Lambda_{L_k}^{(n)}(\mathbf{u})$ is \emph{$(E,m_{L_{k-1}})$-Non Tunneling} ($(E,m_{L_{k-1}})$-NT) if
\begin{enumerate}[(i)]
\item $\forall \mu_b \in \Sigma'':\Lambda_{L_k}^{(n')}(u_{\mathcal{J}})$ is $(E-\mu_b,m_{L_{k-1}},1)$-good.
\item $\forall \lambda_a \in \Sigma':\Lambda_{L_k}^{(n'')}(u_{\mathcal{J}^c})$ is $(E-\lambda_a,m_{L_{k-1}},1)$-good.
\end{enumerate}
Otherwise, we say it is \emph{$(E,m_{L_{k-1}})$-Tunneling} ($(E,m_{L_{k-1}})$-T).
\end{defa}

The following definition is taken from \cite{KN}, see Definition 3.16.

\begin{defa}          \label{def:hnr}
Let $\Lambda_{L_k}^{(n)}(\mathbf{u}) = \Lambda_{L_k}^{(n')}(u_{\mathcal{J}}) \times \Lambda_{L_k}^{(n'')}(u_{\mathcal{J}^c})$ be a PI cube. We say that $\Lambda_{L_k}^{(n)}(\mathbf{u})$ is \emph{$E$-Highly Non-Resonant} ($E$-HNR) if
\begin{enumerate}[(i)]
\item $\forall \mu_b \in \Sigma'':\Lambda_{L_k}^{(n')}(u_{\mathcal{J}})$ is $(E-\mu_b)$-CNR.
\item $\forall \lambda_a \in \Sigma':\Lambda_{L_k}^{(n'')}(u_{\mathcal{J}^c})$ is $(E-\lambda_a)$-CNR.
\end{enumerate}
\end{defa}

\begin{lem}      \label{lem:NT}
Let $\Lambda_{L_k}^{(n)}(\mathbf{u})$ be a PI cube, $k \ge 1$, and let $E \in I_n$. If $\Lambda_{L_k}^{(n)}(\mathbf{u})$ is $E$-HNR and $(E,m_{L_{k-1}})$-NT, then $\Lambda_{L_k}^{(n)}(\mathbf{u})$ is $(E,m_{L_k})$-NS.
\end{lem}
\begin{proof}
Since $\mu_b \ge n''q_-$ for all $\mu_b \in \Sigma''$, given $E \in I_n$ and $\mu_b \in \Sigma''$ we have
\begin{equation}
E-\mu_b \le E- n''q_- \le (nq_-+\varepsilon_0)-n''q_- = n'q_- + \varepsilon_0.      \label{eq:in}
\end{equation}
As $\varepsilon_0<1$, $E- \mu_b \le E_+=\max_n(nq_-+1)$. By hypothesis, $\Lambda_{L_k}^{(n')}(u_{\mathcal{J}})$ is $(E-\mu_b)$-CNR and $(E-\mu_b,m_{L_{k-1}},1)$-good for all $\mu_b \in \Sigma''$, hence $\Lambda_{L_k}^{(n')}(u_{\mathcal{J}})$ is $(E-\mu_b,m_{L_k}+L_k^{\beta-1})$-NS by Lemma~\ref{lem:DKn}. Similarly, $\Lambda_{L_k}^{(n'')}(u_{\mathcal{J}^c})$ is $(E-\lambda_a,m_{L_k}+L_k^{\beta-1})\text{-NS}$ for any $\lambda_a \in \Sigma'$.

Now let $\mathbf{v} \in \mathbf{B}_{L_k}^{\out}(\mathbf{u})$. Then $|u_{\mathcal{J}^c} - v_{\mathcal{J}^c} | \ge L_k-6$ or $|u_{\mathcal{J}} - v_{\mathcal{J}} | \ge L_k-6$. In the first case, we take a large $S>2m_{L_0} \ge 2m_{L_k}$ and apply the first bound of \textsf{(GRI.3)} to obtain
\[
\| G_{\Lambda_{L_k}^{(n)}(\mathbf{u})}(\mathbf{u},\mathbf{v};E) \| \le cL_k^{n'd}e^{-(m_{L_k}+L_k^{\beta-1})L_k} + c'L_k^{n'd} e^{-(L_k-6) S} \le e^{-m_{L_k}L_k} 
\]
since $L_0$ is large. The second case is similar, using the second bound of \textsf{(GRI.3)}.
\end{proof}

\begin{lem}[cf. \cite{KN}, Lemma 3.18]       \label{lem:hnr}
Let $\Lambda_{L_k}^{(n)}(\mathbf{u})=\Lambda_{L_k}^{(n')}(u_{\mathcal{J}}) \times \Lambda_{L_k}^{(n'')}(u_{\mathcal{J}^c})$ be a PI cube. If $\Lambda_{L_k}^{(n)}(\mathbf{u})$ is not $E$-HNR, then
\begin{enumerate}[\rm a.]
\item either there exists a cube $\Lambda_{\ell}^{(n')} \subseteq \Lambda_{L_k}^{(n')}(u_{\mathcal{J}})$ with $\ell \in \N^{\ast}$, $L_k^{1/\alpha} \le \ell \le L_k$ such that for $\Lambda_{\LL}^{(n)} := \Lambda^{(n')}_{\ell} \times \Lambda^{(n'')}_{L_k}(u_{\mathcal{J}^c})$ we have $\dist(\sigma(H_{\Lambda_{\LL}}^{(n)}),E) < e^{-\ell^{\beta}}$,
\item or there exists a cube $\Lambda_{\ell}^{(n'')} \subseteq \Lambda_{L_k}^{(n'')}(u_{\mathcal{J}^c})$ with $\ell \in \N^{\ast}$, $L_k^{1/\alpha} \le \ell \le L_k$ such that for $\Lambda_{\LL}^{(n)} := \Lambda^{(n')}_{L_k}(u_{\mathcal{J}}) \times \Lambda^{(n'')}_{\ell}$ we have $\dist(\sigma(H_{\Lambda_{\LL}}^{(n)}),E) < e^{-\ell^{\beta}}$.
\end{enumerate}
\end{lem}
\begin{proof}
Suppose condition (i) of Definition~\ref{def:hnr} is not satisfied. Then there exist $\mu \in \Sigma''$ and $\Lambda_{\ell}^{(n')} \subseteq \Lambda_{L_k}^{(n')}(u_{\mathcal{J}})$, $L_k^{1/\alpha} \le \ell \le L_k$ such that $\dist(\sigma(H^{(n')}_{\Lambda_{\ell}}),E-\mu) <e^{-\ell^{\beta}}$. Thus, there exists $\eta \in \sigma(H^{(n')}_{\Lambda_{\ell}})$ such that $|E-\mu - \eta| < e^{-\ell^{\beta}}$.

But $\Lambda_{L_k}^{(n')}(u_{\mathcal{J}}) \times \Lambda_{L_k}^{(n'')}(u_{\mathcal{J}^c})$ is PI and $\Lambda_{\ell}^{(n')} \subseteq \Lambda_{L_k}^{(n')}(u_{\mathcal{J}})$, so the interaction $U^{(n)}$ also decouples on $\Lambda_{\LL}^{(n)} := \Lambda^{(n')}_{\ell} \times \Lambda^{(n'')}_{L_k}(u_{\mathcal{J}^c})$ and we get $H_{\Lambda_{\LL}}^{(n)} = H_{\Lambda_{\ell}}^{(n')} \otimes I + I \otimes H_{\Lambda_{L_k}^{(n'')}(u_{\mathcal{J}^c})}$. In particular, the eigenvalues of $H_{\Lambda_{\LL}}^{(n)}$ take the form $E_{a,b} = \eta_a + \mu_b$ for $\eta_a \in \sigma(H^{(n')}_{\Lambda_{\ell}})$ and $\mu_b \in \Sigma''$. We thus showed that $\dist(\sigma(H_{\Lambda_{\LL}}^{(n)}),E) \le |(\eta + \mu) - E| < e^{-\ell^{\beta}}$.

If instead (ii) of Definition~\ref{def:hnr} is not satisfied, we reason similarly and obtain b.
\end{proof}

\begin{lem}      \label{lem:PIn}
Let $\Lambda_{L_k}^{(n)}(\mathbf{u})$, $k \ge 1$ be a PI cube and suppose \emph{(\textsf{DS}}$\,:n',k-1,m_{L_{k-1}},I_{n'})$ holds for all $n'<n$. Then there exists $C_1 = C_1(n,d,q_-)$ such that
\[ 
\prob \{ \exists E \in I_n : \Lambda_{L_k}^{(n)}(\mathbf{u}) \text{ is } (E,m_{L_{k-1}})\text{-}T \} \le C_1 L_k^{(2n-1)d-\frac{2p_{n-1}(1+\theta)^{k-1}}{\alpha}}. 
\]
\end{lem}
\begin{proof}
Let $\Lambda_{L_k}^{(n)}(\mathbf{u}) = \Lambda_{L_k}^{(n')}(u_{\mathcal{J}}) \times \Lambda_{L_k}^{(n'')}(u_{\mathcal{J}^c})$ be PI and $\Sigma'' := \sigma\big(H_{\Lambda_{L_k}(u_{\mathcal{J}^c})}^{(n'')}\big)$. By (\ref{eq:in}), given $E \in I_n$ and $\mu_b \in \Sigma''$, either $E- \mu_b \in [n'q_- - \frac{1}{2},n'q_-+\varepsilon_0] = I_{n'}$, or $E - \mu_b < n'q_- - \frac{1}{2}$. Suppose $E - \mu_b < n'q_- - \frac{1}{2}$, let $\Lambda_{L_{k-1}}^{(n')}(v_1),\Lambda_{L_{k-1}}^{(n')}(v_2)  \subset \Lambda_{L_k}^{(n')}(u_{\mathcal{J}})$ be two separable cubes and let $\eta_b := n'q_- - (E-\mu_b) > \frac{1}{2}$. Then by Theorem~\ref{thm:CT2} given $y_i \in B_{L_{k-1}}^{\out}(v_i)$,
\begin{align*}
\| G_{\Lambda_{L_{k-1}}^{(n')}(v_i)}(v_i,y_i;E-\mu_b) \| & \le \sqrt{\frac{\pi}{2}} \Big(\frac{\sqrt{L_{k-1}}}{\eta_b^{3/4}} + \frac{3}{8 \sqrt{L_{k-1}-8} \eta_b^{5/4}} \Big) e^{-(L_{k-1}-8) \sqrt{\eta_b}} \\
& \le e^{-m_{L_{k-1}} L_{k-1}} 
\end{align*}
because $\sqrt{\eta_b} > \frac{1}{\sqrt{2}} \ge 2 m_{L_{k-1}}$ (in fact $\frac{1}{\sqrt{2}} \gg \frac{c}{L_0^{(1-\beta)/2}} = 2m_{L_0} \ge 2m_{L_{k-1}}$ for $L_0$ large enough). Thus both cubes are $(E,m_{L_{k-1}})\text{-NS}$ in this case. On the other hand,
\[
\prob \{ \exists E-\mu_b \in I_{n'} : \Lambda_{L_{k-1}}^{(n')}(v_1) \text{ and } \Lambda_{L_{k-1}}^{(n')}(v_2) \text{ are } (E-\mu_b,m_{L_{k-1}})\text{-S} \} \le L_{k-1}^{-2p_{n'}(1+\theta)^{k-1}}
\]
by (\textsf{DS}$\,:n',k-1,m_{L_{k-1}},I_{n'}$). But by Lemma~\ref{lem:WEYL} there exists $C >0$ such that 
\[ 
b > C \cdot |\Lambda^{(n'')}_{L_k}(u_{\mathcal{J}^c}) | \implies \mu_b > E-n'q_-+ \textstyle{\frac{1}{2}} \implies E - \mu_b < n'q_- - \textstyle{ \frac{1}{2} }. 
\]
As the number of pairs of cubes in $\Lambda_{L_k}^{(n')}(u_{\mathcal{J}})$ is bounded by $| \mathbf{B}_{L_k}^{(n')}(u_{\mathcal{J}})|^2$, we finally obtain
\begin{align*}
& \prob \{ \exists E \in I_n, \exists \mu_b \in \Sigma'' \text{ such that } \Lambda_{L_k}^{(n')}(u_{\mathcal{J}}) \text{ is }(E-\mu_b,m_{L_{k-1}},1)\text{-bad} \} \\
& \quad \le | \mathbf{B}_{L_k}^{(n')}(u_{\mathcal{J}})|^2 \sum_{b \le C| \Lambda^{(n'')}|} L_{k-1}^{-2p_{n'}(1+\theta)^{k-1}}  \\
& \quad \le \tilde{C} L_k^{2n'd+n''d} L_k^{\frac{-2p_{n'}(1+\theta)^{k-1}}{\alpha}} = \tilde{C} L_k^{(n+n')d-\frac{2p_{n'}(1+\theta)^{k-1}}{\alpha}} \le \frac{C_1}{2} L_k^{(2n-1)d-\frac{2p_{n-1}(1+\theta)^{k-1}}{\alpha}}
\end{align*}
because $p_{n'} \ge p_{n-1}$ for $n'=1,\ldots,n-1$. The same reasoning with $\Lambda_{L_k}^{(n'')}(u_{\mathcal{J}^c})$ and the spectrum $\Sigma'$ of $H_{\Lambda_{L_k}(u_{\mathcal{J}})}^{(n')}$ yields the theorem.
\end{proof}

From now on we declare that
\[
(\textsf{DS}:n',-1,m_{L_{-1}},I_{n'}) = \text{no assumption.}
\]

\begin{thm}      \label{thm:PIn}
Let $k \ge 0$. Suppose that \emph{(\textsf{DS}}$\,:n',k-1,m_{L_{k-1}},I_{n'})$ holds for all $n'<n$ and let $\Lambda_{L_k}^{(n)}(\mathbf{u})$ and $\Lambda_{L_k}^{(n)}(\mathbf{v})$ be separable PI cubes. Then there exists $C_2=C_2(n,d,q_-)$ such that
\[ 
\prob \{ \exists E \in I_n : \Lambda_{L_k}^{(n)}(\mathbf{u}) \text{ and } \Lambda_{L_k}^{(n)}(\mathbf{v}) \text{ are } (E,m_{L_k})\text{-}S \} \le C_2 L_k^{(2n-1)d-\frac{2p_{n-1}(1+\theta)^{k-1}}{\alpha}} \, .
\]
\end{thm}
\begin{proof}
If $k=0$, recall that $L_0$ is chosen so that for $m_{L_0} = \frac{L_0^{(\beta-1)/2}}{3}$,
\[ 
\prob \{ \exists E \in I_n : \Lambda_{L_0}^{(n)}(\mathbf{u}) \text{ and } \Lambda_{L_0}^{(n)}(\mathbf{v}) \text{ are } (E,m_{L_0})\text{-}S \} \le L_0^{-2p_1} \le C_2 L_0^{(2n-1)d-\frac{2p_{n-1}}{\alpha(1+\theta)}} \, .
\]
So suppose $k \ge 1$. By Lemma~\ref{lem:NT},
\begin{align*} 
& \prob \{ \exists E \in I_n : \Lambda_{L_k}^{(n)}(\mathbf{u}) \text{ and } \Lambda_{L_k}^{(n)}(\mathbf{v}) \text{ are } (E,m_{L_k})\text{-}S \} \\
& \quad \le 3 \max_{\mathbf{x}=\mathbf{u},\mathbf{v}} \prob \{\exists E\in I_n:\Lambda_{L_k}^{(n)}(\mathbf{x}) \text{ is } (E,m_{L_{k-1}})\text{-T}\} \\
& \qquad + \prob \{ \exists E \in I_n : \Lambda_{L_k}^{(n)}(\mathbf{u}) \text{ and } \Lambda_{L_k}^{(n)}(\mathbf{v}) \text{ are not } E\text{-HNR} \} \, .
\end{align*}
For $\mathbf{x}=\mathbf{u},\mathbf{v}$, taking $C_2:=4C_1$, we have by Lemma~\ref{lem:PIn}
\[ 
\prob \{\exists E\in I_n:\Lambda_{L_k}^{(n)}(\mathbf{x}) \text{ is } (E,m_{L_{k-1}})\text{-T}\} \le \frac{C_2}{4}L_k^{(2n-1)d-\frac{2p_{n-1}(1+\theta)^{k-1}}{\alpha}}. 
\]
Since both cubes are PI, they are decomposable, say $\Lambda_{L_k}^{(n)}(\mathbf{u}) = \Lambda_{L_k}^{(n')}(u_{\mathcal{J}}) \times \Lambda_{L_k}^{(n'')}(u_{\mathcal{J}^c})$ and $\Lambda_{L_k}^{(n)}(\mathbf{v}) = \Lambda_{L_k}^{(r')}(v_{\mathcal{I}}) \times \Lambda_{L_k}^{(r'')}(v_{\mathcal{I}^c})$, where $n'+n''=r'+r''=n$. If $\Lambda_{\ell_1}^{(n')} \subseteq \Lambda_{L_k}^{(n')}(u_{\mathcal{J}})$ and $\Lambda_{\ell_2}^{(r')} \subseteq \Lambda_{L_k}^{(r')}(v_{\mathcal{I}})$, where $L_k^{1/\alpha} \le \ell_1,\ell_2 \le L_k$, then the rectangles $\Lambda_{\LL}^{(n)} := \Lambda_{\ell_1}^{(n')} \times \Lambda_{L_k}^{(n'')}(u_{\mathcal{J}^c})$ and $\Lambda_{\KK}^{(n)} := \Lambda_{\ell_2}^{(r')} \times \Lambda_{L_k}^{(r'')}(v_{\mathcal{I}^c})$ are pre-separable. Let $J_n = [nq_--\frac{1}{2}-\varepsilon,nq_-+\varepsilon_0+\varepsilon]$ be an $\varepsilon$-enlargement of $I_n$. Then by Theorem~\ref{thm:Wegner2}, we may find $C=C(n,d,q_-)$ such that
\begin{align*}
& \prob \{ \exists E \in I_n : \dist(\sigma(H_{\Lambda_{\LL}}^{(n)}),E) < \varepsilon \text{ and } \dist(\sigma(H_{\Lambda_{\KK}}^{(n)}),E) < \varepsilon \} \\
& \quad \le \prob \{ \dist(\sigma_{J_n}(H_{\Lambda_{\LL}}^{(n)}), \sigma_{J_n}(H_{\Lambda_{\KK}}^{(n)})) < 2\varepsilon \} \le C (2L_k)^{2nd+nd} s(\mu,4 \varepsilon) \, .
\end{align*}
Reasoning similarly for $\Lambda_{\ell_3}^{(n'')} \subseteq \Lambda_{L_k}^{(n'')}(u_{\mathcal{J}^c})$ and $\Lambda_{\ell_4}^{(r'')} \subseteq \Lambda_{L_k}^{(r'')}(v_{\mathcal{I}^c})$, using Lemma~\ref{lem:hnr}, bounding the number of cubes in $\Lambda_{L_k}^{(s)}$ by $|\mathbf{B}_{L_k}^{(s)}| \le (2L_k)^{nd}$ for $s=n',n'',r',r''$, and the number of $\ell \in \N^{\ast}$ satisfying $L_k^{1/\alpha} \le \ell \le L_k$ by $L_k$, we get for $\varepsilon := \max_j e^{-\ell_j^{\beta}} \le e^{-L_k^{\beta/\alpha}}$,
\[ 
\prob \{ \exists E \in I_n : \Lambda_{L_k}^{(n)}(\mathbf{u}) \text{ and } \Lambda_{L_k}^{(n)}(\mathbf{v}) \text{ are not } E\text{-HNR} \} \le 4 C (2L_k)^{4nd+d+2} s(\mu, 4e^{-L_k^{\beta/\alpha}}) \, ,
\]
where $4C$ appear because we apply the above argument 4 times, since Lemma~\ref{lem:hnr} provides 2 cases for $\Lambda_{L_k}^{(n)}(\mathbf{u})$ and 2 cases for $\Lambda_{L_k}^{(n)}(\mathbf{v})$. As estimated in Lemma~\ref{lem:regu},
\[
4 C (2L_k)^{4nd+d+2} s(\mu, 4e^{-L_k^{\beta/\alpha}}) \le L_k^{-2p_1(1+\theta)^k} \le \frac{C_2}{4} L_k^{(2n-1)d-\frac{2p_{n-1}(1+\theta)^{k-1}}{\alpha}}. 
\] 
We thus obtain the theorem for $k \ge 1$. 
\end{proof}

\subsection{General pairs of cubes}
We assume through this subsection that $2 \le n \le N$.
 
\begin{lem}      \label{lem:M}
Let $k \ge 0$. Suppose \emph{(\textsf{DS}}$\,:n,k,m_{L_k},I_n)$ and \emph{(\textsf{DS}}$\,:n',k-1,m_{L_{k-1}},I_{n'})$ hold for all $n'<n$. Then for any cube $\Lambda_{L_{k+1}}^{(n)}(\mathbf{z})$ and $J \in 2\N^{\ast}$,
\[ 
\prob \{ \exists E \in I_n : \Lambda_{L_{k+1}}^{(n)}(\mathbf{z}) \text{ is } (E,m_{L_k},J)\text{-bad } \} \le \frac{1}{8}(L_{k+1}^{-2p_n(1+\theta)^{k+1}}+L_{k+1}^{-Jp_n(1+\theta)^{k+1}/3}) \, .
\]
\end{lem}
\begin{proof}
Put $L=L_{k+1}$, $l=L_k$. If $\Lambda_L^{(n)}(\mathbf{z})$ is $(E,m_l,J)$-bad, then it contains at least $J+1$ pairwise separable cubes which are $(E,m_l)$-S. Hence, either it contains $2$ separable $(E,m_l)$-S PI cubes, or it contains at least $J$ separable $(E,m_l)$-S FI cubes. By Theorem~\ref{thm:PIn},
\begin{align*}
& \prob \{ \exists E \in I_n : \Lambda_L^{(n)}(\mathbf{z}) \text{ contains } 2 \text{ separable } (E,m_l)\text{-S PI cubes}  \} \\
& \quad \le C_2(2L)^{2nd} l^{(2n-1)d-\frac{2p_{n-1}(1+\theta)^{k-1}}{\alpha}} \le c L^{2nd+\frac{(2n-1)d}{\alpha}-\frac{2p_{n-1}(1+\theta)^{k-1}}{\alpha^2}} \, ,
\end{align*}
where we bounded the number of pairs of cubes in $\Lambda_L^{(n)}$ by $| \mathbf{B}_L^{(n)} |^2 \le (2L)^{2nd}$. Now
\begin{align*}
2p_n(1+\theta)^{k+1} & = (2p_n+2 \theta p_n)(1+\theta)^k \\
& < (2p_n + 2)(1+\theta)^k \\
& = \Big(\frac{2p_{n-1}}{\alpha^2(1+\theta)} - \frac{(2n-1)d}{\alpha} - 2nd \Big) (1+\theta)^k \\
& \le \frac{2p_{n-1}(1+\theta)^{k-1}}{\alpha^2} - \frac{(2n-1)d}{\alpha} - 2nd \, .
\end{align*}
Hence,
\[
\prob\{ \exists E \in I_n : \Lambda_L^{(n)}(\mathbf{z}) \text{ contains } 2 \text{ separable } (E,m_l)\text{-S PI cubes}\} \le \frac{1}{8} L^{-2p_n(1+\theta)^{k+1}} \, .
\]

Next, by Lemma~\ref{lem:FI}, pairs of separable FI cubes are completely separated, so the corresponding Hamiltonians $H_{\Lambda_l}^{(n)}$ are independent. Thus, bounding again the number of pairs of cubes in $\Lambda_L^{(n)}$ by $(2L)^{2nd}$, we get by (\textsf{DS}$\,:n,k,m_l,I_n)$,
\begin{align*}
& \prob \{ \exists E \in I_n : \Lambda_L^{(n)}(\mathbf{z}) \text{ contains at least } J \text{ separable } (E,m_l)\text{-S FI cubes} \} \\
& \quad \le \prob \{ \exists E \in I_n : \Lambda_L^{(n)}(\mathbf{z}) \text{ contains at least } 2 \text{ separable } (E,m_l)\text{-S FI cubes} \}^{J/2} \\
& \quad \le \big( (2L)^{2nd} l^{-2p_n(1+\theta)^k} \big)^{J/2} \le cL^{(nd - \frac{p_n(1+\theta)^k}{\alpha})J} \le cL^{(\frac{p_n - 1}{3} - \frac{2p_n(1+\theta)^k}{3})J}
\end{align*}
because $\alpha = 3/2$ and $nd \le \frac{p_N-1}{3} \le \frac{p_n-1}{3}$. Moreover,
\[
cL^{\frac{Jp_n}{3}(1-2(1+\theta)^k-\frac{1}{p_n})} \le \frac{1}{8} L^{\frac{Jp_n}{3}(1-2(1+\theta)^k-\theta)} \, .
\]
We thus showed that
\[
\prob \{ \exists E \in I_n : \Lambda_L^{(n)}(\mathbf{z}) \text{ is } (E,m_l,J)\text{-bad} \} \le \frac{1}{8} L^{-2p_n(1+\theta)^{k+1}} + \frac{1}{8} L^{\frac{Jp_n}{3}(1-2(1+\theta)^k-\theta)} \, ,
\]
which completes the proof by (\ref{eq:ms1}).
\end{proof}

\begin{thm}      \label{thm:FIn}
Let $k \ge 0$. Then the properties \emph{(\textsf{DS}}$\,:n',k-1,m_{L_{k-1}},I_{n'})$ for $n'<n$ and \emph{(\textsf{DS}}$\,:n,k,m_{L_k},I_n)$ imply \emph{(\textsf{DS}}$\,:n,k+1,m_{L_{k+1}},I_n)$.
\end{thm}
\begin{proof}
Put $L=L_{k+1}$, $l=L_k$ and let $\Lambda_L^{(n)}(\mathbf{u})$ and $\Lambda_L^{(n)}(\mathbf{v})$ be a pair of separable cubes. Since $\varepsilon_0<1$, any $E \in I_n$ satisfies $E\le E_+=\max_n(nq_-+1)$, so applying Lemma~\ref{lem:DKn} with $J=6$, noting that $l>l^{\ast}$ because $L_0 > l^{\ast}$, we have
\begin{align*}
& \prob\{ \exists E \in I_n : \Lambda_L^{(n)}(\mathbf{u}) \text{ and } \Lambda_L^{(n)}(\mathbf{v}) \text{ are } (E,m_L)\text{-S} \} \\
& \quad \le 3 \max_{\mathbf{z}=\mathbf{u},\mathbf{v}} \prob \{ \exists E \in I_n : \Lambda_L^{(n)}(\mathbf{z}) \text{ is } (E,m_l,6)\text{-bad} \} \\
& \qquad + \prob\{ \exists E \in I_n : \Lambda_L^{(n)}(\mathbf{u}) \text{ and } \Lambda_L^{(n)}(\mathbf{v}) \text{ are not } E\text{-CNR} \} \, ,
\end{align*}
since an ($E,m_L$)-S cube is a fortiori ($E,m_L + L^{\beta-1}$)-S. Now by Lemma~\ref{lem:M},
\[
\prob \{ \exists E \in I_n : \Lambda_L^{(n)}(\mathbf{z}) \text{ is } (E,m_l,6)\text{-bad} \} \le \frac{1}{4}L^{-2p_n(1+\theta)^{k+1}}
\]
for $\mathbf{z}=\mathbf{u},\mathbf{v}$. The assertion follows, using (\textsf{W2}$\,:n,k+1,I_n)$.
\end{proof}

\subsection{Conclusion}
\begin{thm}      \label{thm:MPMSA}
There exists $m>0$ such that \emph{(\textsf{DS}}$\,:N,k,m,I_N)$ holds for all $k \ge 0$.
\end{thm}
\begin{proof}
By construction $L_0$ is a large integer such that (\textsf{DS}$\,:n,0,m_{L_0},I_n)$ holds for all $1 \le n \le N$, with $m_{L_0} = \frac{1}{3L_0^{(1-\beta)/2}} > \frac{48 NK(N)}{L_0^{1-\beta}}$. We prove the theorem by induction on $n$.

For $n=1$, we know that (\textsf{DS}$\,:1,k,m_{L_k},I_1)$ holds for all $k \ge 0$ by Theorem~\ref{thm:1-part} and induction on $k$.

Now fix $n \ge 2$ and suppose that (\textsf{DS}$\,:n',k,m_{L_k},I_{n'})$ holds for all $k\ge 0$ and all $n'<n$. We may then apply Theorem~\ref{thm:FIn} to obtain (\textsf{DS}$\,:n,k,m_{L_k},I_n)$ for all $k \ge 0$, by induction on $k$. (Recall that (\textsf{DS}$\,:n',-1,m_{L_{-1}},I_{n'})$ means no assumption).

This completes the induction and we obtain (\textsf{DS}$\,:N,k,m_{L_k},I_N)$ for all $k\ge 0$. Now
\[
S:= \sum_{j=0}^{\infty} (m_{L_j} - m_{L_{j+1}}) = 96NK(N) \sum_{j=0}^{\infty} \frac{m_{L_j}}{L_j^{\alpha-1}} + 3 \sum_{j=0}^{\infty} \frac{1}{L_j^{\alpha(1-\beta)}} \, .
\]
Since $m_{L_j} \le m_{L_0}$, we have
\[ 
S \le 96NK(N)m_{L_0} \sum_{j=0}^{\infty} \frac{1}{L_0^{(\alpha-1)\alpha^j}} + 3 \sum_{j=0}^{\infty} \frac{1}{L_0^{\alpha(1-\beta)\alpha^j}} \le \frac{m_{L_0}}{2} \le m_{L_0} - m
\]
for any $0<m \le \frac{m_{L_0}}{2}$, assuming $L_0$ is large enough. Now given $k \ge 1$, put $S_k := \sum_{j=k}^{\infty} (m_{L_j} - m_{L_{j+1}})$. Again the $m_{L_j}$ are decreasing, so $S_k \ge 0$ for all $k$. Since 
\[
m_{L_0} -m \ge S = m_{L_0} - m_{L_k} + S_k \, ,
\]
we get
\[
m \le m_{L_k} - S_k \le m_{L_k} \, ,
\]
so in particular, (\textsf{DS}$\,:N,k,m,I_N)$ holds for all $k \ge 0$.
\end{proof}

\section{Generalized Eigenfunctions}    \label{sec:EGN}
In this section we prove a generalized eigenfunction expansion for $H^{(n)}(\omega)$ which plays an important role in the proof of localization. For this we show that our model satisfies the hypotheses of \cite[Theorem 3.1]{KKS}. 

Given a bounded potential $v=(v_{\kappa})\ge0$, we define $H_v$ to be the operator associated with the form 
\[
\mathfrak{h}_v[f,g] = \sum_{\kappa \in \mathcal{K}} \mathfrak{a}_{v_{\kappa}}[f_{\kappa}, g_{\kappa}], \qquad D(\mathfrak{h}_v)=W^{1,2}(\Gamma^{(n)}),
\]
where
\[
\mathfrak{a}_{v_{\kappa}}[\phi,\psi] := \langle \nabla \phi, \nabla \psi \rangle + \langle v_{\kappa} \phi, \psi \rangle, \qquad D(\mathfrak{a}_{v_{\kappa}}) = W^{1,2}((0,1)^n).
\]
We first show that $\mathfrak{h}_v$ is a Dirichlet form and that $(e^{-tH^{(n)}(\omega)})_{t \ge0}$ is ultracontractive. For this we follow \cite{O}, as it covers the case where the Hilbert space is over $\C$.
 
\begin{lem}     \label{lem:Dir}
$\mathfrak{h}_v$ is a Dirichlet form.
\end{lem}
\begin{proof}
Combine \cite[Corollary 4.3]{O}, \cite[Corollary 4.10]{O} and \cite[Theorem 2.25]{O} to see that $p(D(\mathfrak{a}_{v_{\kappa}})) \subseteq D(\mathfrak{a}_{v_{\kappa}})$ and $\mathfrak{a}_{v_{\kappa}}[p \circ f] \le \mathfrak{a}_{v_{\kappa}}[f]$ for every $f \in D(\mathfrak{a}_{v_{\kappa}})$ and every normal contraction $p$. Now let $u=(u_{\kappa}) \in D(\mathfrak{h}_v)$ such that $u_{\kappa} \in C([0,1]^n)$ for all $\kappa$ and let $p$ be a normal contraction. If $\sigma^i \equiv (0,1)^{n-1}$ is a common face to $\kappa_1$ and $\kappa_2$ and if $\gamma:W^{1,2}((0,1)^n) \to L^2((0,1)^{n-1})$ is the trace operator, then
\begin{align*}
\| \gamma(p(u_{\kappa_1})) - \gamma(p(u_{\kappa_2})) \|_{L^2(0,1)^{n-1}}^2 & = \| \gamma(p(u_{\kappa_1}) - p(u_{\kappa_2})) \|_{L^2(0,1)^{n-1}}^2 \\
& = \int_{(0,1)^{n-1}} |p(u_{\kappa_1}(x)) - p(u_{\kappa_2}(x))|^2 dx \\
& \le \int_{(0,1)^{n-1}} |u_{\kappa_1}(x) - u_{\kappa_2}(x)|^2 dx  \\
& = \| \gamma(u_{\kappa_1}) - \gamma(u_{\kappa_2}) \|^2_{L^2(0,1)^{n-1}} = 0 \, ,
\end{align*}
where the last equality holds since $u$ is continuous on $\sigma^i$. By the density of $C^{\infty}([0,1]^n)$ in $W^{1,2}((0,1)^n)$ and the continuity of $\gamma$ and $p$, the same is true for all $u \in D(\mathfrak{h}_v)$. Hence $p \circ u$ is continuous on $\sigma^i$ for all $u \in D(\mathfrak{h}_v)$. Thus $p(D(\mathfrak{h}_v)) \subseteq D(\mathfrak{h}_v)$ . Furthermore,
\[
\mathfrak{h}_v[p \circ u] = \sum_{\kappa \in \mathcal{K}} \mathfrak{a}_{v_{\kappa}}[p \circ u_{\kappa}] \le \sum_{\kappa \in \mathcal{K}} \mathfrak{a}_{v_{\kappa}}[u_{\kappa}] = \mathfrak{h}_v[u] \, .
\]
Hence by \cite[Theorem 2.25]{O}, $(e^{-t H_v})_{t \ge 0}$ is sub-Markovian. Thus $\mathfrak{h}_v$ is a Dirichlet form. 
\end{proof}

\begin{lem}      \label{lem:SG}
There exists $c=c(n)>0$ such that for all $\omega \in \Omega$,
\[ 
\forall t>0 : \|e^{-tH^{(n)}(\omega)}\|_{L^2(\Gamma) \to L^{\infty}(\Gamma)} \le ct^{- n /4}e^{-(nq_- -1)t} \, .
\]
\end{lem}
\begin{proof}
Let $Q:= (0,1)^n$. By the Gagliardo-Nirenberg interpolation inequality (see \cite{Nir66}), we have for any $u \in W^{1,2}(Q)$,
\[
\| u\|_{L^2(Q)} \le C( \| \nabla u \|^a_{L^2(Q)} + \| u\|^a_{L^1(Q)} ) \|u\|^{1-a}_{L^1(Q)} \, ,
\]
where $a = \frac{n}{n+2}$. By H\"older inequality, we have $\|u \|_{L^1(Q)} \le \|u\|_{L^2(Q)}$. Using H\"older inequality again, with $p=\frac{2}{a}$ and $q=\frac{2}{2-a}$, we get $(x^a+y^a) \le 2^{1/q}(x^2+y^2)^{a/2}$. Thus,
\[
\|u\|_{L^2(Q)} \le \tilde{C} ( \| \nabla u \|^2_{L^2(Q)} + \| u\|^2_{L^2(Q)} )^{a/2} \|u\|^{1-a}_{L^1(Q)} \le \tilde{C}(\mathfrak{a}_{v_{\kappa}}[u])^{a/2} \|u\|_{L^1(Q)}^{1-a}
\]
for any bounded potential $v_{\kappa} \ge 1$. Hence, for any $f \in D(\mathfrak{h}_v) \cap L^1(\Gamma)$ we have
\[ 
\|f\|_{L^2(\Gamma)}^2 = \sum_{\kappa \in \mathcal{K}} \|f_{\kappa}\|_{L^2(Q)}^2 \le \tilde{C}^2 \sum_{\kappa \in \mathcal{K}} (\mathfrak{a}_{v_{\kappa}}[f_{\kappa}])^{a} \|f_{\kappa}\|_{L^1(Q)}^{2(1-a)} \, .
\]
Using H\"older inequality with $p = \frac{1}{a}$ and $q=\frac{1}{1-a}$ we get
\begin{align*}
\|f\|_{L^2(\Gamma)}^2 & \le \tilde{C}^2 \Big( \sum \mathfrak{a}_{v_{\kappa}}[f_{\kappa}] \Big)^a \Big( \sum \|f_{\kappa}\|^2_{L^1(Q)} \Big)^{(1-a)} \\
& \le \tilde{C}^2 \Big( \sum \mathfrak{a}_{v_{\kappa}}[f_{\kappa}] \Big)^a \Big( \sum \|f_{\kappa}\|_{L^1(Q)} \Big)^{2(1-a)} = \tilde{C}^2 (\mathfrak{h}_v[f])^a \|f\|_{L^1(\Gamma)}^{2(1-a)} \, .
\end{align*}
Using Lemma~\ref{lem:Dir} and applying \cite[Theorem 6.3]{O}, it follows that
\[ 
\forall t>0 : \|e^{-tH_v}\|_{L^1(\Gamma) \to L^2(\Gamma)} \le ct^{- n /4} \, .
\]
But $\|e^{-tH_v}\|_{L^1 \to L^2} = \|e^{-tH_v}\|_{L^2 \to L^{\infty}}$ by duality. So the assertion follows by taking $v:= V^{\omega} -(nq_- - 1) \ge 1$ and noting that 
\[
e^{-tH_v} = \exp(-t(H^{(n)}(\omega) - (nq_- - 1))) = e^{(nq_- - 1)t} e^{-tH^{(n)}(\omega)} \, . \qedhere
\]
\end{proof}

Let $T$ be the self-adjoint operator on $L^2(\Gamma^{(n)})$ given by 
\[
Tf(\mathbf{x}) := w(\mathbf{x}) f(\mathbf{x}), \qquad \text{where }  w(\mathbf{x}) = (1+ \| \mathbf{x} \|_2^2)^{\gamma/4}
\]
for some fixed $\gamma>nd+1$. We now establish

\begin{lem}     \label{lem:tr}
There exists $C=C(n)$ such that for all $\omega \in \Omega$ and $t>0:$
\[
\tr (T^{-1} e^{-2tH^{(n)}(\omega)} T^{-1}) \le Ct^{-n/2} e^{-2(nq_--1)t} \| w^{-1} \|_{L^2}^2 < \infty \, .
\]
Furthermore, if $E_{\omega}$ is the spectral projection of $H^{(n)}(\omega)$, then the set function $\nu_{\omega}$ on $\R$ given by
\[
\nu_{\omega}(J) := \tr (T^{-1} E_{\omega}(J) T^{-1}) = \| E_{\omega}(J) T^{-1} \|_2^2
\]
is a spectral measure for $H^{(n)}(\omega)$ which is finite on bounded Borel sets $J$.
\end{lem}
\begin{proof}
Divide $\Gamma^{(n)}$ into annuli $\Gamma^{(n)} \cap (\Lambda^{(n)}_{k+1}(\mathbf{0}) \setminus \Lambda^{(n)}_k(\mathbf{0}))$. Then by (\textsf{NB.$n$}),
\[ 
\int_{\Gamma^{(n)}} |w^{-1}|^2 \, \dd m \le \sum_k \frac{m(\Gamma^{(n)} \cap \Lambda^{(n)}_{k+1}(\mathbf{0}))}{(1+ k^2)^{\gamma/2}} \le C \sum_k \frac{(k+1)^{nd}}{(1+ k^2)^{\gamma/2}} < \infty \, .
\]
Thus $w^{-1} \in L^2$ and $T^{-1}:L^{\infty} \to L^2$ is $2$-summing (see \cite[Examples 2.9, p. 40]{DJH}). Using Lemma~\ref{lem:SG}, \cite[Theorem 2.4, p. 37]{DJH} and \cite[Theorem 4.10, p. 84]{DJH}, it follows that $T^{-1} e^{-t H}$ is Hilbert-Schmidt and
\[
\| T^{-1} e^{-t H} \|_2 \le \|w^{-1} \|_{L^2} \|e^{-tH}\|_{L^2 \to L^{\infty}} \le c t^{-n/4} e^{-(nq_- - 1)t} \|w^{-1} \|_{L^2} \, .
\]
Noting that $\| T^{-1} e^{-t H} \|_2 = \| (T^{-1} e^{-t H})^{\ast} \|_2 = \| e^{-t H} T^{-1} \|_2$, this yields
\[
\tr (T^{-1} e^{-2tH} T^{-1}) = \|e^{-t H}T^{-1} \|_2^2 \le C t^{-n/2} e^{-2(nq_- - 1)t} \|w^{-1} \|_{L^2}^2 \, .
\]

Now let $J$ be a bounded Borel set and put $b:=\sup \{ \lambda \in J\}$. Then
\[
0 \le e^{-2b} E_{\omega}(J) \le \int_J e^{-2 \lambda} \, \dd E_{\omega}(\lambda) \le \int_{\sigma(H)} e^{-2 \lambda} \, \dd E_{\omega}(\lambda) = e^{-2H} \, .
\]
Hence $\nu_{\omega}(J) \le e^{2b} \tr (T^{-1} e^{-2H} T^{-1}) \le C_J \|w^{-1} \|^2_{L^2}$ and $\nu_{\omega}$ is finite on bounded Borel sets. It is easy to see that $\nu_{\omega}$ is a Borel measure. Finally, $\nu_{\omega}(J) = 0 \iff E_{\omega}(J) = 0$, so $\nu_{\omega}$ is a spectral measure for $H$.
\end{proof}

We note in passing that given a bounded interval $I$, the previous proof yields a constant $C=C(I,n,q_-)>0$ independent of $\omega$ such that
\begin{equation}
\sup_{\omega} \nu_{\omega}(I) \le C \|w^{-1}\|^2_{L^2} =:C_{\tr} \, .  \label{eq:ge1}
\end{equation}

Let $\mathcal{H}_+$ be the space $D(T)$ equipped with the norm $\| \phi\|_+ = \|T \phi\|$ and $\mathcal{H}_-$ the completion of $\mathcal{H}$ in the norm $\| \psi \|_- = \| T^{-1} \psi \|$. By construction $\mathcal{H}_+ \subset \mathcal{H} \subset \mathcal{H}_-$ is then a triple of Hilbert spaces with natural injections $\iota_+ : \mathcal{H}_+ \to \mathcal{H}$ and $\iota_- : \mathcal{H} \to \mathcal{H}_-$ continuous with dense range. The inner product $\langle \, , \, \rangle_{\mathcal{H}}$ extends to a sesquilinear form on $\mathcal{H}_+ \times \mathcal{H}_-$ which turns $\mathcal{H}_+$ and $\mathcal{H}_-$ into conjugate duals (see \cite[Lemma 1]{PSW} and \cite{BSU}). The adjoint of an operator $O$ with respect to this duality is denoted by $O^{\dag}$.

\begin{lem}      \label{lem:dense}
For all $\omega \in \Omega$, the space
\[ 
\mathcal{D}_+ = \{ f \in D(H^{(n)}(\omega)) \cap \mathcal{H}_+ : H^{(n)}(\omega) f \in \mathcal{H}_+ \} 
\]
is dense in $\mathcal{H}_+$ and is an operator core for $H^{(n)}(\omega)$.
\end{lem}
\begin{proof}
Set $H:= H^{(n)}(\omega)$ and let $C_c^{\infty}(\Gamma):= (\mathop \oplus_{\kappa} C_c^{\infty}(0,1)^n) \cap C_c(\Gamma)$. Clearly $\mathcal{D}_+ \supseteq C_c^{\infty}(\Gamma)$ (see the definition of $D(H)$ in the Appendix, Section~\ref{sec:app}). Moreover, $C_c^{\infty}(\Gamma)$ is dense in $L^2(\Gamma)$. Now let $f \in \mathcal{H}_+$, then $Tf \in L^2(\Gamma)$ may be approximated by $g_j \in C_c^{\infty}(\Gamma)$, hence $\|f - T^{-1}g_j\|_+ \to 0$ and clearly $T^{-1} g_j \in C_c^{\infty}(\Gamma)$. Hence $\mathcal{D}_+$ is dense in $\mathcal{H}_+$.

To show $\mathcal{D}_+$ is a core we follow \cite[Proposition 2.4]{BS}: let $E<nq_-$ and consider $D_0=(H-E)^{-1} C_c(\Gamma)$. Since $C_c(\Gamma)$ is dense in $L^2(\Gamma)$, $D_0$ is a core for $H$. By Combes-Thomas estimate, each $f \in D_0$ is exponentially decreasing. Hence $f \in \mathcal{H}_+$ and
\[ 
Hf = (H-E)f + E f = \varphi + E f \in \mathcal{H}_+ 
\]
since $f = (H-E)^{-1} \varphi$ for some $\varphi \in C_c(\Gamma)$. This proves the claim. 
\end{proof}
By \cite[Lemma 3.1]{KKS}, $H^{(n)}(\omega)$ regarded as an operator on $\mathcal{H}_-$ is thus closable and densely defined. We denote its closure by $H^{(n)}_-(\omega)$. We say that $\psi \in \mathcal{H}_-$ is a \emph{generalized eigenfunction} of $H^{(n)}(\omega)$ with corresponding \emph{generalized eigenvalue} $\lambda \in \C$ if $\psi$ is an eigenfunction of $H^{(n)}_-(\omega)$ with eigenvalue $\lambda$, i.e. if $\psi \in D(H^{(n)}_-(\omega))$ and $H^{(n)}_-(\omega) \psi = \lambda \psi$. 

By \cite[Lemma 3.2]{KKS}, we have $H_-^{(n)}(\omega) \psi = H^{(n)}(\omega) \psi$ for any $\psi \in D(H_-^{(n)}) \cap \mathcal{H}$. In particular, if a generalized eigenfunction lies in $\mathcal{H}$, then it is an eigenfunction.

We may now state the main result of this section. Here $\mathcal{T}_1(\mathcal{H}_+,\mathcal{H}_-)$ and $\mathcal{T}_{1,+}(\mathcal{H}_+,\mathcal{H}_-)$ are the spaces of trace class and positive trace class operators from $\mathcal{H}_+$ to $\mathcal{H}_-$, respectively (see \cite{KKS} for details). 

\begin{thm}      \label{thm:GEGN} 
Let $\nu_{\omega}$ be the spectral measure of $H^{(n)}(\omega)$ introduced in Lemma~\ref{lem:tr}. There exists a $\nu_{\omega}$-locally integrable function $P_{\omega}:\R \to \mathcal{T}_{1,+}(\mathcal{H}_+,\mathcal{H}_-)$ such that
\[
\iota_- f(H^{(n)}(\omega)) E_{\omega}(J) \iota_+ = \int_J f(\lambda) P_{\omega}(\lambda) \dd \nu_{\omega}(\lambda)
\]
for all bounded Borel sets $J$ and all bounded Borel functions $f$, where the integral is the Bochner integral of $\mathcal{T}_1(\mathcal{H}_+,\mathcal{H}_-)$-valued functions. Furthermore, for $\nu_{\omega}$-a.e. $\lambda \in \R$,
\[
P_{\omega}(\lambda) = P_{\omega}(\lambda)^{\dag}, \qquad \tr P_{\omega}(\lambda) = 1
\]
and $P_{\omega}(\lambda) \phi \in \mathcal{H}_-$ is a generalized eigenfunction of $H^{(n)}(\omega)$ with generalized eigenvalue $\lambda$ for any $\phi \in \mathcal{H}_+$.
\end{thm}
\begin{proof} 
Applying \cite[Theorem 3.1]{KKS} and \cite[Corollary 3.1]{KKS}, it only remains to show that $P_{\omega}(\lambda) = P_{\omega}(\lambda)^{\dag}$ $\nu_{\omega}$-a.e. This follows from \cite[Eq.(46)]{KKS} and the fact that $\iota_+^{\dag} = \iota_-$.
\end{proof}

\section{Exponential Localization}    \label{sec:exp}
The fundamental link between mutiscale analysis and localization is provided by the following eigenfunction decay inequality. Since we will not rely on the regularity of generalized eigenfunctions, the proof is a bit longer than in \cite{Sto}.

\begin{lem}      \label{lem:EDI}
Let $E_+ \in \R$. There exists $C=C(E_+,n,d,q_-)$ such that, if $\mathbf{x}_0 \in \Z^{nd}$ and $\mathbf{C}(\mathbf{x}) \subset \Lambda_{L-6}^{(n)}(\mathbf{x}_0)$, then every generalized eigenfunction $\psi$ of $H^{(n)}(\omega)$ corresponding to $\lambda \in \rho(H_{\Lambda_L(\mathbf{x}_0)}^{(n)}) \cap (-\infty,E_+]$ satisfies
\[
\| \chi_{\mathbf{x}} \psi \| \le C \cdot |\mathbf{B}_L^{\out}(\mathbf{x}_0)| \max_{\mathbf{y} \in \mathbf{B}_L^{\out}(\mathbf{x}_0)} \| G_{\Lambda_L^{(n)}(\mathbf{x}_0)}(\mathbf{x},\mathbf{y};\lambda) \| \cdot \| \chi_{\Lambda_L^{\out}(\mathbf{x}_0)} \psi \|. 
\]
\end{lem}
\begin{proof}
Let $\Lambda := \Lambda_L^{(n)}(\mathbf{x}_0)$ and $\varphi \in \tilde{C}_c^1(\Gamma \cap \Lambda)$ such that $\varphi = 1$ on a neighborhood of $\Gamma^{(n)} \cap \mathbf{C}(\mathbf{x})$, $\supp \nabla \varphi \subset \tilde{Q} := \inter (\Lambda_{L-2}^{(n)}(\mathbf{x}_0) \setminus \Lambda_{L-4}^{(n)}(\mathbf{x}_0))$, and $\| \nabla \varphi \|_{\infty} \le C_1(nd)$. Then
\[
\| \chi_{\mathbf{x}} \psi \|^2 = \langle \varphi \psi, \chi_{\mathbf{x}} \psi \rangle = \langle \varphi \psi , (H_{\Lambda} - \lambda) G_{\Lambda}(\lambda) \chi_{\mathbf{x}} \psi \rangle.
\]
Put $H:= H^{(n)}(\omega)$. Since $\psi \in D(H_-)$ and $H_-$ is the closure of $H$, there exists $(f_j)$ in $D(H)$ such that $\| f_j - \psi \|_- \to 0$ and $\| Hf_j - H_- \psi\|_- \to 0$ as $j \to \infty$. Now for any $\chi$ of compact support we have
\begin{equation}
\| \chi f_j - \chi \psi \| \le \| \chi w \| \cdot \| f_j - \psi \|_- \to 0        \label{eq:exp1}
\end{equation}
(recall that $Tg :=wg$). Hence taking $v := G_{\Lambda}(\lambda) \chi_{\mathbf{x}} \psi$ we have
\[
\| \chi_{\mathbf{x}} \psi \|^2 = \lim_{j \to \infty} \langle \varphi f_j , (H_{\Lambda} - \lambda) v \rangle = \lim_{j \to \infty} (\mathfrak{h}_{\Lambda} - \lambda) [ \varphi f_j, v]
\]
since $\varphi f_j \in D(\mathfrak{h}_{\Lambda}) = W^{1,2}(\Gamma^{(n)} \cap \Lambda)$ by Lemma~\ref{lem:sobolev}. Now
\begin{align*}
(\mathfrak{h}_{\Lambda} - \lambda) [ \varphi f_j, v] & = \langle \nabla (\varphi f_j), \nabla v \rangle + \langle (V^{\omega}-\lambda) \varphi f_j, v \rangle \\
& = \big[ \langle \nabla f_j, \nabla(\varphi v) \rangle + \langle (V^{\omega} - \lambda) f_j, \varphi v \rangle \big] + \langle f_j \nabla \varphi, \nabla v \rangle - \langle \nabla f_j, v \nabla \varphi \rangle.
\end{align*}
Since $\varphi v \in W^{1,2}(\Gamma^{(n)} \cap \Lambda)$ has compact support in $\Lambda$, we may extend it by zero to a function $g$ in $D(\mathfrak{h}) \cap C_c(\Gamma)$. Hence
\begin{equation}
(\mathfrak{h}_{\Lambda} - \lambda) [ \varphi f_j, v] = \langle (H-\lambda) f_j, g \rangle + \langle f_j \nabla \varphi, \nabla v \rangle - \langle \nabla f_j, v \nabla \varphi \rangle. \label{eq:exp2}
\end{equation}
Now $H_- \psi = \lambda \psi$, so by the choice of $f_j$
\begin{equation}
\| (H-\lambda)f_j\|_- \le \| Hf_j -\lambda \psi \|_- +  |\lambda| \cdot \| f_j - \psi\|_- \to 0. \label{eq:exp3}
\end{equation}
Thus
\[
| \langle (H-\lambda) f_j, g \rangle| \le \| (H - \lambda) f_j \|_- \|g\|_+ \to 0.
\]
The second term in (\ref{eq:exp2}) tends to $\langle \psi \nabla \varphi, \nabla v \rangle$ by (\ref{eq:exp1}). For the third term, note that by Lemma~\ref{lem:SOL}, taking $Q := \inter \Lambda_L^{\out}(\mathbf{x}_0)$, we can find $c_1$ such that
\begin{align*}
\| \chi_{\tilde{Q}} \nabla f_j \| & \le c_1 ( \| \chi_Q (H- \lambda) f_j \| + \| \chi_Q f_j \|) \\
& \le c_1 ( \| \chi_Q w\| \cdot \|(H- \lambda) f_j \|_- + \| \chi_Q f_j \|) \to c_1 \| \chi_Q \psi \|
\end{align*}
using (\ref{eq:exp1}) and (\ref{eq:exp3}). 

Recalling that $\supp \nabla \varphi \subset \tilde{Q}$, the above derivation finally yields
\[
\| \chi_{\mathbf{x}} \psi \|^2 \le \| \nabla \varphi \|_{\infty} \| \chi_{\tilde{Q}} \psi \| \| \chi_{\tilde{Q}} \nabla v \| + c_1 \| \nabla \varphi \|_{\infty} \| \chi_{Q} \psi \| \| \chi_{\tilde{Q}} v \| \, .
\]
By Lemma~\ref{lem:SOL}, we can find $c_2$ such that
\[ 
\| \chi_{\tilde{Q}} \nabla v \| \le c_2 \| \chi_{Q} v \|
\]
(note that $(H_{\Lambda} - \lambda) v = \chi_{\mathbf{x}} \psi=0$ on $Q$). Taking $C = \max(2c_1 \| \nabla \varphi \|_{\infty} , 2c_2 \| \nabla \varphi \|_{\infty})$ and noting that $\tilde{Q} \subset Q \subset \Lambda_L^{\out}(\mathbf{x}_0)$ we thus get
\[
\| \chi_{\mathbf{x}} \psi \|^2 \le C \cdot \| \chi_{\Lambda_L^{\out}(\mathbf{x}_0)} \psi \| \cdot \| \chi_{\Lambda_L^{\out}(\mathbf{x}_0)} v \|.
\]
Since $\| \chi_{\Lambda_L^{\out}(\mathbf{x}_0)} v \| \le \| \chi_{\Lambda_L^{\out}(\mathbf{x}_0)} G_{\Lambda}(\lambda) \chi_{\mathbf{x}} \| \|\chi_{\mathbf{x}} \psi\|$, we get
\[
\| \chi_{\mathbf{x}} \psi \| \le C \cdot \| \chi_{\Lambda_L^{\out}(\mathbf{x}_0)} G_{\Lambda}(\lambda) \chi_{\mathbf{x}} \| \cdot \| \chi_{\Lambda_L^{\out}(\mathbf{x}_0)} \psi \|.
\]
The assertion now follows by the triangle inequality.
\end{proof}

We now prove exponential localization by adapting \cite[Theorem 2.3]{DK}.

\begin{proof}[Proof of Theorem~\ref{thm:exp}]
Choose $m$ and $\varepsilon_0$ such that (\textsf{DS}$\,:N,k,m,I_N)$ holds for all $k \ge 0$ in $I_N = [Nq_- - \frac{1}{2}, Nq_- + \varepsilon_0]$, as guaranteed by Theorem~\ref{thm:MPMSA}. Let $\sigma^{\omega}_{\text{gen}}$ be the set of generalized eigenvalues of $H^{(N)}(\omega)$. By Theorem~\ref{thm:GEGN} there exists $A^{\omega}_0 \subseteq \R$ of full $\nu_{\omega}$-measure such that $A^{\omega}_0 \subseteq \sigma^{\omega}_{\text{gen}}$. If we show that every $\lambda \in \sigma^{\omega}_{\text{gen}} \cap I$ is an eigenvalue, $A^{\omega}_0 \cap I$ will be countable (as $L^2(\Gamma^{(N)})$ is separable), so $\left. {\nu_{\omega}}\right|_{I}$ will be concentrated on a countable set and $\sigma(H^{(N)}(\omega)) \cap I$ will be pure point. It thus suffices to show that with probability one the generalized eigenfunctions of $H^{(N)}(\omega)$ corresponding to $\lambda \in \sigma^{\omega}_{\text{gen}} \cap I$ decay exponentially with mass $m$.

Let $b \in \N^{\ast}$ to be chosen later and define
\[ 
A_{k+1} = \mathbf{B}^{(N)}_{2br_{k+1}}(\mathbf{0}) \, \big\backslash \, \mathbf{B}^{(N)}_{2r_k}(\mathbf{0}),
\]
where $r_k:= r_{N,L_k}$. Then by Lemma~\ref{lem:separable}, any $\mathbf{x} \in A_{k+1}$ satisfies that $\Lambda_{L_k}^{(N)}(\mathbf{x})$ is separable from $\Lambda_{L_k}^{(N)}(\mathbf{y})$ for any $\mathbf{y} \in \mathbf{B}_{r_k}^{(N)}(\mathbf{0})$. Now define the event
\[ 
E_k = \{ \exists \lambda \in I , \mathbf{x} \in A_{k+1}, \mathbf{y} \in \mathbf{B}^{(N)}_{r_k}(\mathbf{0}) : \Lambda_{L_k}^{(N)}(\mathbf{x}) \text{ and } \Lambda_{L_k}^{(N)}(\mathbf{y}) \text{ are } (\lambda,m)\text{-S} \}. 
\]
Then by Theorem~\ref{thm:MPMSA}, we have
\[ 
\prob(E_k) \le (4br_{k+1}-1)^{Nd}(2r_k-1)^{Nd}L_k^{-2p_N(1+\theta)^k} \le c L_k^{2N \alpha d - 2p_N(1+\theta)^k} 
\]
Hence $\sum_{k=0}^{\infty} \prob(E_k) < \infty$. So by the Borel-Cantelli Lemma, if we define the event
\[ 
\Omega_1 = \{ E_k \text{ occurs finitely often} \}, 
\]
we have $\prob(\Omega_1) = 1$. Now let $\omega \in \Omega_1$ and $\lambda \in \sigma^{\omega}_{\text{gen}} \cap I$ correspond to a generalized eigenfunction $\psi$. If $\| \chi_{\mathbf{x}} \psi \| = 0$ for all $\mathbf{x} \in \Z^{Nd}$, then $\psi=0$ and the theorem holds. So suppose $\| \chi_{\mathbf{y}} \psi \| \neq 0$ for some $\mathbf{y} \in \Z^{Nd}$. Then by Lemma~\ref{lem:EDI} we may find $C_1=C_1(N,d,q_-,\gamma,\|\psi\|_-)$ such that
\[ 
\| \chi_{\mathbf{y}} \psi \| \le C_1 \cdot |\mathbf{B}_{L_k}^{\out}(\mathbf{y})| \max_{\mathbf{z} \in \mathbf{B}_{L_k}^{\out}(\mathbf{y})} \| G_{\Lambda_{L_k}^{(N)}(\mathbf{y})}(\mathbf{y},\mathbf{z};\lambda) \| \cdot (1+ (| \mathbf{y} | + L_k)^2)^{\gamma/4}.
\]
Now if $\Lambda_{L_k}^{(N)}(\mathbf{y})$ is $(\lambda,m)\text{-NS}$, we get
\[ 
\| \chi_{\mathbf{y}} \psi \| \le C_1' L_k^{Nd-1} e^{-mL_k} ( 1+ (| \mathbf{y} | + L_k)^2 )^{\gamma/4}. 
\]
Since $\| \chi_{\mathbf{y}} \psi \| \neq 0$, there exists $k_0$ such that $\Lambda_{L_k}^{(N)}(\mathbf{y})$ is $(\lambda,m)\text{-S}$ for all $k \ge k_0$. But there exists $k_1$ such that $\mathbf{y} \in \mathbf{B}_{r_k}^{(N)}(\mathbf{0})$ for all $k \ge k_1$. Finally, since $\omega \in \Omega_1$, we may find $k_2$ such that $E_k$ does not occur if $k \ge k_2$. Let $k_3 = \max(k_0,k_1,k_2)$. Then for $k \ge k_3$, we conclude that $\Lambda_{L_k}^{(N)}(\mathbf{x})$ is $(\lambda,m)\text{-NS}$ for all $\mathbf{x} \in A_{k+1}$. 

Now given $0<\rho<1$, we choose $b > \frac{1+\rho}{1-\rho}$ and define
\[ 
\tilde{A}_{k+1} = \mathbf{B}^{(N)}_{\frac{2b}{1+\rho}r_{k+1}}(\mathbf{0}) \, \big\backslash \, \mathbf{B}^{(N)}_{\frac{2}{1-\rho}r_k}(\mathbf{0}). 
\]
Then $\tilde{A}_{k+1} \subset A_{k+1}$ and for any $\mathbf{x} \in \tilde{A}_{k+1}$, we have
\[ 
\dist(\mathbf{x},\partial A_{k+1}) \ge \rho \cdot | \mathbf{x} |. 
\] 
Indeed, if $\mathbf{x} \in \tilde{A}_{k+1}$, then
\[
d(\mathbf{x},\partial \mathbf{B}^{(N)}_{2br_{k+1}}(\mathbf{0})) \ge 2br_{k+1} - \frac{2b}{1+\rho}r_{k+1} = \rho \frac{2b}{1+\rho}r_{k+1} \ge \rho \cdot |\mathbf{x}|,
\]
\[
d(\mathbf{x},\partial \mathbf{B}^{(N)}_{2r_k}(\mathbf{0})) = |\mathbf{x}| - 2r_k \ge |\mathbf{x}| - (1-\rho)|\mathbf{x}| = \rho \cdot |\mathbf{x}|.
\]
Now for $\mathbf{x} \in \tilde{A}_{k+1}$ with $k \ge k_3$, $\Lambda_{L_k}^{(N)}(\mathbf{x})$ is $(\lambda,m)\text{-NS}$, so by Lemma~\ref{lem:EDI},
\[ 
\| \chi_{\mathbf{x}} \psi \| \le C_2 L_k^{2(Nd-1)} e^{-mL_k} \| \chi_{\mathbf{w}_1} \psi \| 
\]
for some $\mathbf{w}_1 \in \mathbf{B}_{L_k}^{\out}(\mathbf{x})$. We may iterate at least $\lfloor \frac{ \rho \cdot | \mathbf{x}|}{L_k-1} \rfloor$ times and obtain
\begin{align*}
\| \chi_{\mathbf{x}} \psi \| & \le (C_2 L_k^{2(Nd-1)} e^{-mL_k} )^{\lfloor \frac{\rho \cdot | \mathbf{x} | }{L_k-1} \rfloor} C_3\| \psi \|_- ( 1+ ( 2br_{k+1})^2 )^{\gamma/4} \\
& \le e^{- m \rho' \rho \cdot | \mathbf{x} |}
\end{align*}
for any $0<\rho'<1$, provided $k \ge k_4$ for some $k_4 \ge k_3$. But if $\mathbf{x} \notin \mathbf{B}^{(N)}_{\frac{2}{1-\rho}r_{k_4}}(\mathbf{0})$, then $\mathbf{x} \in \tilde{A}_{k+1}$ for some $k\ge k_4$ (since $\frac{2b}{1+\rho} r_{k+1} > \frac{2}{1-\rho} r_{k+1}$) and the bound is satisfied. Thus,
\[ 
\log \| \chi_{\mathbf{x}} \psi \| \le -m \rho' \rho \cdot | \mathbf{x} | 
\]
whenever $\mathbf{x} \notin \mathbf{B}^{(N)}_{ \frac{2}{1-\rho}r_{k_4}}(\mathbf{0})$. Hence
\[ 
\limsup_{| \mathbf{x} | \to \infty } \frac{ \log \| \chi_{\mathbf{x}} \psi \|}{|\mathbf{x}|} \le -m\rho'\rho
\]
for all $\rho,\rho' \in (0,1)$, which completes the proof of the theorem.
\end{proof}

\section{Dynamical Localization}    \label{sec:dyn}
We finally establish dynamical localization for $H^{(N)}(\omega)$ using the approach of \cite{GK}. In the following we consider the event
\[
R(m,L,I,\mathbf{x},\mathbf{y}) := \{ \, \forall \lambda \in I: \Lambda_L^{(N)}(\mathbf{x}) \text{ or } \Lambda_L^{(N)}(\mathbf{y}) \text{ is } (\lambda,m)\text{-NS} \, \}
\]
for $\mathbf{x}$, $\mathbf{y}$ such that the corresponding cubes are separable. We start with the following key lemma.

\begin{lem}      \label{lem:proj}
Let $m>0$, $I \subset \R$ and assume $\omega \in R(m,L,I,\mathbf{x},\mathbf{y})$. Then 
\[
\| \chi_{\mathbf{x}} P_{\omega}(\lambda) \chi_{\mathbf{y}} \|_2 \le C e^{-mL/2} (1+|\mathbf{x}|)^{\gamma /2}(1+|\mathbf{y}|)^{\gamma /2}
\]
for $\nu_{\omega}$-a.e. $\lambda \in I$ and large $L$, with $C=C(I,m,N,d,\gamma,q_-)<\infty$.
\end{lem}
\begin{proof}
Let $A^{\omega}_0$ be the set of full $\nu_{\omega}$-measure such that Theorem~\ref{thm:GEGN} holds for all $\lambda \in A^{\omega}_0$. Given $\lambda \in I \cap A^{\omega}_0$, either $\Lambda_L^{(N)}(\mathbf{x})$ or $\Lambda_L^{(N)}(\mathbf{y})$ is $(\lambda,m)$-NS. Since $P_{\omega}(\lambda) = P_{\omega}(\lambda)^{\dag}$, we have $\| \chi_{\mathbf{x}} P_{\omega}(\lambda) \chi_{\mathbf{y}} \|_2 = \| \chi_{\mathbf{y}} P_{\omega}(\lambda) \chi_{\mathbf{x}} \|_2$, so we may assume that $\Lambda_L^{(N)}(\mathbf{x})$ is $(\lambda,m)$-NS. Now if $\phi \in \mathcal{H}$, then by Theorem~\ref{thm:GEGN}, the vector $P_{\omega}(\lambda) \chi_{\mathbf{y}} \phi$ is a generalized eigenfunction of $H^{(N)}(\omega)$, hence by Lemma~\ref{lem:EDI}, 
\[
\| \chi_{\mathbf{x}} P_{\omega}(\lambda) \chi_{\mathbf{y}} \phi \| \le C_1 (2L-1)^{Nd-1} e^{-mL} \| \chi_{\Lambda_L^{\out}(\mathbf{x})} P_{\omega}(\lambda) \chi_{\mathbf{y}} \phi \| \, .
\] 
Hence by definition of the HS norm, 
\[
\| \chi_{\mathbf{x}} P_{\omega}(\lambda) \chi_{\mathbf{y}} \|_2 \le C_1 (2L-1)^{Nd-1} e^{-mL} \| \chi_{\Lambda_L^{\out}(\mathbf{x})} P_{\omega}(\lambda) \chi_{\mathbf{y}} \|_2 \, .
\]
But
\begin{align*}
\| \chi_{\Lambda_L^{\out}(\mathbf{x})} P_{\omega}(\lambda) \chi_{\mathbf{y}} \|_1 & \le \| \chi_{\Lambda_L^{\out}(\mathbf{x})} \|_{\mathcal{H}_- \to \mathcal{H}} \| P_{\omega}(\lambda) \|_{\mathcal{T}_1(\mathcal{H}_+,\mathcal{H}_-)} \| \chi_{\mathbf{y}} \|_{\mathcal{H} \to \mathcal{H}_+} \\
& \le c (1+(|\mathbf{x}| +L)^2)^{\gamma/4}(1+(|\mathbf{y}|+1)^2)^{\gamma/4}
\end{align*}
since $\tr P_{\omega}(\lambda) =1$ and $P_{\omega}(\lambda)\ge 0$. The claim follows since $\| \cdot \|_2 \le \| \cdot \|_1$.
\end{proof}

We now establish the decay of the operator kernel. Given a bounded $K$ as in the statement of Theorem~\ref{thm:main}, we find $k_0>0$ such that $K \subset \Gamma \cap \Lambda_{r_{N,L_{k_0}}}^{(N)}(\mathbf{0})$. For $j \ge k_0$ put
\begin{align*}
F_j = \Lambda_{2r_{N,L_j}}^{(N)}(\mathbf{0}), \qquad & \qquad \tilde{F}_j = \mathbf{B}_{2r_{N,L_j}}^{(N)}(\mathbf{0}), \\
M_j=F_{j+1} \setminus F_j, \qquad & \qquad \tilde{M}_j = \tilde{F}_{j+1} \setminus \tilde{F}_j.
\end{align*}
In the following, we choose $m$ and $\varepsilon_0$ such that (\textsf{DS}$\,:N,k,m,I_N)$ holds for all $k \ge 0$ in $I_N = [Nq_- - \frac{1}{2}, Nq_- + \varepsilon_0]$, as guaranteed by Theorem~\ref{thm:MPMSA}.
 
\begin{lem}     \label{lem:kernel}
There exists $c=c(N,d,q_-,r_0,\gamma)$ such that for $\mathbf{x} \in \tilde{M}_j$ and $\mathbf{y} \in \mathbf{B}_{r_{N,L_j}}^{(N)}(\mathbf{0})$ with $j$ large enough, we have for $I=[Nq_-,Nq_-+\varepsilon_0]:$
\[
\expect \Big( \sup_{\| f\| \le 1} \| \chi_{\mathbf{x}} f(H^{(N)}(\omega)) E_{\omega}(I) \chi_{\mathbf{y}} \|_2^2 \Big) \le c (e^{-mL_j/2} + L_j^{-2p_N(1+\theta)^j+ \gamma}).
\]
\end{lem}
\begin{proof}
Given a bounded Borel function $f$ put $f_I:= f \chi_I$ and $H_{\omega} := H^{(N)}(\omega)$. By Theorem~\ref{thm:GEGN} and standard properties of the Bochner integral in the space of HS operators we have
\[
\| \chi_{\mathbf{x}} f_I(H_{\omega}) \chi_{\mathbf{y}} \|_2 \le \int_I |f(\lambda)| \| \chi_{\mathbf{x}} P_{\omega}(\lambda) \chi_{\mathbf{y}} \|_2 \, \dd \nu_{\omega}(\lambda).
\]
Since $\mathbf{x} \in \tilde{M}_j$ and $\mathbf{y} \in \mathbf{B}_{r_{N,L_j}}^{(N)}(\mathbf{0})$, we know by Lemma~\ref{lem:separable} that $\Lambda_{L_j}^{(N)}(\mathbf{x})$ and $\Lambda_{L_j}^{(N)}(\mathbf{y})$ are separable. Hence if $\omega \in B_j:= R(m,L_j,I,\mathbf{x},\mathbf{y})$, we have by Lemma~\ref{lem:proj}
\[
\| \chi_{\mathbf{x}} P_{\omega}(\lambda) \chi_{\mathbf{y}} \|_2 \le C_1 L_{j+1}^{\gamma/2} L_j^{\gamma/2} e^{-mL_j/2} \le e^{-mL_j/4}
\]
for $\nu_{\omega}$-a.e. $\lambda \in I$ and $j$ large enough. Hence
\[
\| \chi_{\mathbf{x}} f_I(H_{\omega}) \chi_{\mathbf{y}} \|_2 \le \| f\|_{\infty} e^{-mL_j/4} \nu_{\omega}(I) \le C_{\tr} \| f\|_{\infty} e^{-mL_j/4}
\] 
where $C_{\tr}=C_{\tr}(N,d,q_-,\gamma)$ is given by (\ref{eq:ge1}). For $\omega \in B_j^c$ we have the bound
\begin{align*}
\| \chi_{\mathbf{x}} f_I(H_{\omega}) \chi_{\mathbf{y}} \|_2^2 & \le \|f\|_{\infty}^2 \|E_{\omega}(I) \chi_{\mathbf{y}} \|_2^2 \\
& \le \|f\|_{\infty}^2 \| \chi_{\mathbf{y}} T\|^2 \|E_{\omega}(I) T^{-1}\|_2^2 \le C_2 \|f\|_{\infty}^2 L_j^{\gamma} \nu_{\omega}(I)
\end{align*}
for $C_2=C_2(\gamma,N,d,r_0)$. Again $\nu_{\omega}(I) \le C_{\tr}$, so we finally get
\[
\expect \Big( \sup_{\| f\| \le 1} \| \chi_{\mathbf{x}} f_I(H_{\omega}) \chi_{\mathbf{y}} \|_2^2 \Big) \le C^2_{\tr} e^{-mL_j/2} \prob(B_j) + C_2 C_{\tr} L_j^{\gamma} \prob(B_j^c).
\]
Using Theorem~\ref{thm:MPMSA} to estimate $\prob(B_j^c)$, we obtain the assertion.
\end{proof}

We are finally ready to prove our main result. Note that if $R$ is a Hilbert-Schmidt operator on $L^2(\Gamma)$ and if $A, B \subset \R^{Nd}$ are disjoint, then
\begin{equation}
\| \chi_{A \cup B} R \|_2^2 = \tr[R^{\ast} \chi_{A \cup B} R] = \tr[R^{\ast} \chi_A R] + \tr[R^{\ast} \chi_B R] = \| \chi_A R \|_2^2 + \| \chi_B R \|_2^2, \label{eq:dyn1}
\end{equation}
\begin{equation}
\| R \chi_{A \cup B} \|_2^2 = \| \chi_{A \cup B} R^{\ast} \|_2^2 = \| \chi_A R^{\ast} \|_2^2 + \| \chi_B R^{\ast} \|_2^2 = \| R \chi_A\|_2^2 + \|R \chi_B \|_2^2. \label{eq:dyn2}
\end{equation}

\begin{proof}[Proof of Theorem~\ref{thm:main}]
Let $k \ge k_0$ be sufficiently large so that Lemma~\ref{lem:kernel} holds for $j \ge k$. Given $s>0$ and a bounded Borel function $f$ put $f_I := f \chi_I$ and $H_{\omega}:=H^{(N)}(\omega)$. Since $F_k \cup \big(\cup_{j \ge k} M_j \big) = \R^{Nd}$, we have by (\ref{eq:dyn1})
\begin{align*}
\expect \Big\{ \sup_{\|f\| \le 1} \| X^{s/2} f_I(H_{\omega}) \chi_K \|_2^2 \Big\} & \le \expect \Big\{ \sup_{\| f\| \le 1} \| \chi_{F_k} X^{s/2} f_I(H_{\omega}) \chi_K \|_2^2 \Big\} \\
& \quad + \expect \Big\{ \sum_{j \ge k} \sup_{\| f\| \le 1} \| \chi_{M_j} X^{s/2} f_I(H_{\omega}) \chi_K \|_2^2 \Big\} \, . 
\end{align*}
Let us estimate the first term. We have
\begin{align*}
\| \chi_{F_k} X^{s/2} f_I(H_{\omega}) \chi_K \|_2^2 & \le c_1 \|f\|_{\infty}^2 L_k^s \| E_{\omega}(I) \chi_K \|_2^2 \\
& \le c_1 \|f\|_{\infty}^2 L_k^s \| \chi_K T \|^2 \| E_{\omega}(I) T^{-1} \|_2^2.
\end{align*}
Since $\| E_{\omega}(I) T^{-1} \|_2^2 = \nu_{\omega}(I) \le C_{\tr}$ by (\ref{eq:ge1}), we get 
\[
\expect \Big\{ \sup_{\| f\| \le 1} \| \chi_{F_k} X^{s/2} f_I(H_{\omega}) \chi_K \|_2^2 \Big\} \le c_2 L_k^{s+\gamma} < \infty.
\]
For the second term, note that $\chi_{M_j} X^{s/2} g= X^{s/2} \chi_{M_j} g$ for $g \in D(X^{s/2})$, so using (\ref{eq:dyn1}) and (\ref{eq:dyn2}),
\begin{align*}
& \expect \Big\{ \sum_{j \ge k} \sup_{\| f\| \le 1} \| \chi_{M_j} X^{s/2} f_I(H_{\omega}) \chi_K \|_2^2 \Big\} \\
& \qquad \le \sum_{j \ge k} c_3 L_{j+1}^s \sum_{\mathbf{x} \in \tilde{M}_j, \mathbf{y} \in \mathbf{B}_{r_{N,L_{k_0}}}^{(N)}(\mathbf{0})} \expect \Big\{ \sup_{\| f\| \le 1} \| \chi_{\mathbf{x}} f_I(H_{\omega}) \chi_{\mathbf{y}} \|_2^2 \Big\} \, .
\end{align*}
Estimating $|\tilde{M}_j| \le c L_{j+1}^{Nd}$, $|\mathbf{B}_{r_{N,L_{k_0}}}^{(N)}| \le c'L_{k_0}^{Nd}$ and using Lemma~\ref{lem:kernel}, the series converges. This completes the proof of the theorem.
\end{proof}

\section{Appendix}    \label{sec:app}
In this section we prove various results used in the text. We shall repeat the statements of the theorems for the reader's convenience.

\begin{thm}
Given $\omega \in \Omega$, $\mathfrak{h}^{(n)}_{\omega}$ is closed, densely defined and bounded from below. The unique self-adjoint operator $H^{(n)}(\omega)$ associated with $\mathfrak{h}^{(n)}_{\omega}$ is given by
\[ 
H^{(n)}(\omega) :(f_{\kappa}) \mapsto (- \Delta f_{\kappa} + V_{\kappa}^{\omega} f_{\kappa}), \qquad \text{for } (f_{\kappa}) \in D(H^{(n)}(\omega)).
\]
\end{thm}
\begin{proof}
As a direct sum of Hilbert spaces, the space 
\[
\big(\mathop \oplus _{\kappa \in \mathcal{K}} W^{1,2}((0,1)^n), \| \cdot \|_{W^{1,2}(\Gamma)} \big), \qquad \|f\|_{W^{1,2}(\Gamma)}^2: = \sum_{\kappa \in \mathcal{K}} \|f_{\kappa}\|_{W^{1,2}((0,1)^n)}^2
\]
is a Hilbert space. By the trace theorem for $W^{1,2}((0,1)^n)$ (see e.g. \cite[Theorem 1.1.2]{Necas}), $\big(W^{1,2}(\Gamma), \| \cdot \|_{W^{1,2}(\Gamma)}\big)$ is a closed subspace of $\big(\mathop \oplus_{\kappa} W^{1,2}((0,1)^n), \| \cdot \|_{W^{1,2}(\Gamma)}\big)$, hence a Hilbert space. Finally, $\mathfrak{h}^{(n)}_{\omega} \ge nq_-$. If for $f \in D(\mathfrak{h}^{(n)}_{\omega})$, we define $\|f\|^2_{\mathfrak{h}^{(n)}_{\omega}} := \mathfrak{h}^{(n)}_{\omega}[f]+(-nq_-+1)\|f\|^2_{L^2(\Gamma)}$, then $\| \, \|_{\mathfrak{h}^{(n)}_{\omega}}$ is equivalent to $\| \, \|_{W^{1,2}(\Gamma)}$. Hence, $\mathfrak{h}^{(n)}_{\omega}$ is closed.

Let $C_c^{\infty}(\Gamma):= (\mathop \oplus_{\kappa} C_c^{\infty}(0,1)^n) \cap C_c(\Gamma)$. Since $D(\mathfrak{h}^{(n)}_{\omega}) \supset C_c^{\infty}(\Gamma)$, $\mathfrak{h}^{(n)}_{\omega}$ is densely defined. By \cite[Theorem 4.1.5]{Sto}, the associated operator $H^{(n)}(\omega)$ is given by
\[ 
D(H^{(n)}(\omega)) = \{ f  \in D(\mathfrak{h}^{(n)}_{\omega}) \, | \, \exists g \in \mathcal{H} : \forall v \in D(\mathfrak{h}^{(n)}_{\omega}), \, \mathfrak{h}^{(n)}_{\omega}[f,v] = \langle g, v \rangle \}, 
\]
\[ 
H^{(n)}(\omega)f := g. 
\]
So let $f \in D(H^{(n)}(\omega))$. Then in particular, given $v \in C_c^{\infty}(\Gamma)$, we have
\[ 
\langle \nabla f, \nabla v \rangle = \langle g - V^{\omega} f , v \rangle. 
\]
Hence $- \Delta f = g - V^{\omega} f$ in the sense of distributions. As $g,V^{\omega}f \in L^2(\Gamma)$ and as $C_c^{\infty}(\Gamma)$ is dense in $L^2(\Gamma)$, the equality holds in the $L^2$ sense\footnote{It is this part that distinguishes the difficulty of the domain for multi-particles: for $n=1$, if $-f'' \in L^2$ then $f \in W^{2,2}$, but for $n>1$, the fact that $-\Delta f \in L^2$ does not imply that $f \in W^{2,2}$.}. Hence $H^{(n)}(\omega)f = g = - \Delta f + V^{\omega} f$.
\end{proof}

\begin{thm}
There exists $\Omega_0 \subseteq \Omega$ with $\prob(\Omega_0) = 1$ such that for all $\omega \in \Omega_0:$
\[ 
[nq_-,nq_+] \subset \sigma(H^{(n)}(\omega)) \subseteq [nq_-,+\infty). 
\]
In particular, $\inf \sigma(H^{(n)}(\omega)) = nq_-$ almost surely.
\end{thm}
\begin{proof}
Since $U_{\kappa}^{(n)} \ge 0$ and $W_{\kappa}^{\omega} \ge nq_-$, then $H^{(n)}(\omega) \ge nq_-$ and $\sigma(H^{(n)}(\omega)) \subseteq [nq_-,+\infty)$ for all $\omega \in \Omega$. To prove that $\sigma(H^{(n)}(\omega)) \supset [nq_-,nq_+]$ almost surely, let $E \in [nq_-,nq_+]$, put $I^E_m = [\frac{E}{n} - \frac{1}{nm}, \frac{E}{n} + \frac{1}{nm}]$ for $m=1,2,\ldots$ and let
\[
B_m := \big\{ (x_1,\ldots,x_n) \in (\Z^d)^n : \min_{i\neq j} |x_i-x_j| \ge 2m+r_0 \big\}, 
\]
where $r_0$ is the interaction range. Given $k \in \N^{\ast}$, consider the event
\[ 
\Omega^E_m(k) := \big\{ \omega \in \Omega : \omega_e \in I^E_m \quad \forall e \in \mathcal{E}\big(\Gamma^{(1)} \cap \Pi \Lambda_m^{(n)}(\mathbf{x}_{k,m})\big) \big\} \, ,
\]
where $\mathbf{x}_{k,m} := 2^{kn}(2m+r_0)(1,2,\ldots,n)$. Then $\mathbf{x}_{k,m} \in B_m$ for each $k$, $\prob(\Omega^E_m(\mathbf{x}_{k,m})) = \mu (I^E_m)^{\# \{ \mathcal{E}(\Gamma^{(1)} \cap \Pi \Lambda_{m+1}^{(n)}) \}}$ is the same for all $k$ and it is strictly positive since $\frac{E}{n} \in [q_-,q_+] = \supp \mu$. Hence, $\sum_{k \ge 1} \prob(\Omega_m(\mathbf{x}_{k,m})) = \infty$. Moreover, $\Pi \Lambda_m^{(n)}(\mathbf{x}_{k,m}) \cap \Pi \Lambda_m^{(n)}(\mathbf{x}_{k',m}) = \emptyset$ for $k \neq k'$, so the events $\{\Omega^E_m(k)\}_{k \in \N^{\ast}}$ are independent. Thus, by Borel-Cantelli lemma II, if $\Omega^E_m := \cap_{k' \ge 1} \cup_{k \ge k'} \Omega^E_m(k)$, then $\prob(\Omega^E_m) = 1$. Let $\Omega^E := \cap_{m \in \N^{\ast}} \Omega^E_m$, then $\prob(\Omega^E) = 1$.

Fix $\omega \in \Omega^E$ and let $m \in \N^{\ast}$. Then $\omega \in \Omega^E_m$, so we may find $k \in \N^{\ast}$ such that $\omega \in \Omega^E_m(k)$. We finally construct a Weyl sequence: choose $g_m \in D(H^{(n)})$ such that $0 \le g_m \le 1$, $g_m = 1$ on $\Gamma \cap \Lambda_{m-1}^{(n)}(\mathbf{x}_{k,m})$, $g_m = 0$ on $\Gamma \cap \Lambda_m^{(n)}(\mathbf{x}_{k,m})^c$ and $\| \Delta g_m\|_{\infty} \le C$, for some $C=C(nd)$. Let $f_m := c_m g_m$, where $c_m := \|g_m\|^{-1}$. Then $\|f_m\|=1$, $\| \Delta f_m\|_{\infty} \le C c_m$ and
\[
\|(H^{(n)}(\omega)-E)f_m\| = \|\chi_{\Lambda_m^{(n)}(\mathbf{x}_{k,m})} (-\Delta + U^{(n)} + W^{\omega} - E)f_m\| \,. 
\]
But $\mathbf{x}_{k,m} \in B_m$, so $U^{(n)}=0$ on $\Lambda_m^{(n)}(\mathbf{x}_{k,m})$. Also $\omega \in \Omega^E_m(k)$, so $|W_{\kappa}^{\omega} - E| \le \frac{1}{m}$ for all $\kappa \in \mathcal{K}(\Gamma \cap \Lambda_m^{(n)}(\mathbf{x}_m))$. Thus
\[
\|(H^{(n)}(\omega)-E)f_m\| = \|\chi_{\Lambda_m^{(n)}} (-\Delta + W^{\omega} - E)f_m\| \le \|\Delta f_m\| + \frac{1}{m}\|f_m\| \to 0 \, .
\]
Indeed, note that
\[ 
1 = \|f_m\|^2 \ge \|\chi_{\Lambda_{m-1}^{(n)}}f_m\|^2 = c_m^2 (\#\{\mathcal{K}(\Gamma \cap \Lambda_{m-1}^{(n)})\}),
\]
hence $c_m^2 \le (\#\{\mathcal{K}(\Gamma \cap \Lambda_{m-1}^{(n)})\})^{-1}$ and using \textsf{(NB.$n$)},
\begin{align*}
\|\Delta f_m\|^2 & = \|\chi_{\Lambda_m^{(n)} \setminus \Lambda_{m-1}^{(n)}}\Delta f_m\|^2 \le \|\Delta f_m\|_{\infty}^2 \left(\#\{\mathcal{K}(\Gamma \cap \Lambda_m^{(n)})\} - \#\{\mathcal{K}(\Gamma \cap \Lambda_{m-1}^{(n)})\} \right) \\
& \le C^2 \frac{(2m)^n(2m - 1)^{nd-n} - (2m-2)^n(2m-3)^{nd-n}}{(2m-2)^n(2m-3)^{nd-n}} \xrightarrow [m \rightarrow \infty] {} 0.
\end{align*}

Thus, for any $\omega \in \Omega^E$ we have $E \in \sigma(H^{(n)}(\omega))$. Let $\Omega_0 := \bigcap_{E \in [nq_-,nq_+] \cap \Q} \Omega^E$. Then $\prob(\Omega_0)$ = 1 and for any $\omega \in \Omega_0$ we have $\sigma(H^{(n)}(\omega)) \supset [nq_-,nq_+] \cap \Q$. Since the spectrum is closed, the proof is complete.
\end{proof}

\begin{lem}
The following estimates hold:
\begin{align*}
\#\{\mathcal{E}(\Gamma^{(1)} \cap \Lambda_L^{(1)}) \} & = d(2L)(2L-1)^{d-1} \le d \cdot |\Lambda_L^{(1)}|, \tag{\textsf{NB.$1$}} \\
\#\{\mathcal{K}(\Gamma^{(n)} \cap \Lambda_{\LL}^{(n)}) \} & = \prod_{j=1}^n\Big(d(2L_j)(2L_j-1)^{d-1}\Big) \le d^n \cdot |\Lambda_{\LL}^{(n)}|. \tag{\textsf{NB.$n$}}
\end{align*}
\end{lem}
\begin{proof}
For $d=1$, it is obvious that $\#\{\mathcal{E}(\Gamma^{(1)} \cap \Lambda_L^{(1)}) \} = 2L$ since in this case $\Lambda_L^{(1)}$ is just an open segment of length $2L$ and each edge has length $1$.

So let us suppose the estimate is true for $d=m$ and calculate the number of edges in a $1$-cube in $\R^{m+1}$, with coordinate axes $x_1,\ldots,x_{m+1}$. Since this number is invariant by translations, we may suppose the cube is $\Lambda_L^{(1)}(0)$. By hypothesis, the hyperplane $\{x_{m+1} = L-1\} \cap \Lambda_L^{(1)}(0)$ contains $m(2L)(2L-1)^{m-1}$ edges. The same holds for the hyperplane $\{x_{m+1} = L-2\} \cap \Lambda_L^{(1)}(0)$ and so on, by calculating the number of edges in the hyperplanes $x_{m+1} = L-1,L-2,\ldots,-L+1$, we obtain $(2L-1)(m(2L)(2L-1)^{m-1}) = m(2L)(2L-1)^m$ edges. It remains to calculate the number of ``vertical'' edges, i.e. edges that lie in the translates of the axis $x_{m+1}$ in $\Lambda_L^{(1)}(0)$. There are $(2L-1)^m$ such translates (since each $x_j$, $j=1,\ldots,m$ varies from $L-1$ to $-L+1$), and each axis contains $2L$ edges by the case $d=1$. Hence we get $(2L-1)^m(2L)$ vertical edges. The total number of edges is thus $m(2L)(2L-1)^m + (2L-1)^m(2L) = (m+1)(2L)(2L-1)^m$. Thus \textsf{(NB.$1$)} holds $\forall d \ge 1$.

Since $\Gamma^{(n)} = \Gamma^{(1)} \times \ldots \times \Gamma^{(1)}$, \textsf{(NB.$n$)} follows directly from \textsf{(NB.$1$)}.
\end{proof}

\begin{lem}
$H_{\Lambda_{\LL}}^{(n)}(\omega)$ has a compact resolvent. Its discrete set of eigenvalues denoted by $E_j(H_{\Lambda_{\LL}}^{(n)}(\omega))$ counting multiplicity satisfies the following Weyl law:
\[ 
\forall S \in \R \ \exists C = C(n,d,S-nq_-):\ j > C |\Lambda_{\LL}^{(n)}| \implies E_j(H_{\Lambda_{\LL}}^{(n)}(\omega)) > S.  \tag{\textsf{WEYL.$n$}} \]
Moreover, $C$ is independent of $\omega$, and if $S>S^{\ast}(n,q_-)$, then $C \le \lfloor \frac{d^n(S-nq_-)^{n/2}}{(4\pi)^{n/2}\Gamma(n/2)} \rfloor +1$.
\end{lem}
\begin{proof}
Put $\Lambda = \Lambda_{\LL}^{(n)}$ and define the Neumann-decoupled Laplacian $-\Delta_{\Lambda}^{\mathrm{N},\dec}$ via the form $\mathfrak{h}_{\Lambda}^{\dec}[f,g] = \sum_{\kappa \in \mathcal{K}(\Gamma \cap \Lambda)} \langle \nabla f_{\kappa} , \nabla g_{\kappa} \rangle$, with $D(\mathfrak{h}_{\Lambda}^{\dec}) = \mathop \oplus_{\kappa \in \mathcal{K}(\Gamma \cap \Lambda)} W^{1,2}((0,1)^n)$. Then $D(\mathfrak{h}_{\omega, \Lambda}^{(n)}) \subset D(\mathfrak{h}_{\Lambda}^{\dec})$ and $\mathfrak{h}_{\omega, \Lambda}^{(n)}[f] \ge \mathfrak{h}_{\Lambda}^{\dec}[f] + nq_- \|f\|^2$, hence
\[ 
H_{\Lambda}^{(n)}(\omega) \ge -\Delta_{\Lambda}^{\mathrm{N},\dec} + nq_- \, . \qquad (\star) 
\]
Since $\Delta_{\Lambda}^{\mathrm{N},\dec} = \mathop \oplus_{\kappa \in \mathcal{K}(\Gamma \cap \Lambda^{(n)})} \Delta_{(0,1)^n}^{\mathrm{N}}$, the eigenvalues $E_j(- \Delta_{\Lambda}^{\mathrm{N},\dec})$ are just the eigenvalues $E_k(- \Delta_{(0,1)^n}^{\mathrm{N}})$ with multiplicity $\# \{\mathcal{K}(\Gamma \cap \Lambda^{(n)}) \} \le d^n |\Lambda^{(n)}|$. In particular, $E_j(- \Delta_{\Lambda}^{\mathrm{N},\dec}) \to \infty$ as $j \to \infty$, so by $(\star)$ and \cite[Theorem XIII.64]{RS}, $H_{\Lambda}^{(n)}$ has a compact resolvent and thus a discrete spectrum. Now by Weyl law for $E_k(- \Delta_{(0,1)^n}^{\mathrm{N}})$ (\cite[Section XIII.15]{RS}), there exists $C_1$ such that $k > C_1 \implies E_k(- \Delta_{(0,1)^n}^{\mathrm{N}}) > S - nq_-$, and if $S$ is large, $C_1 \approx \frac{(S-nq_-)^{n/2}}{(4\pi)^{n/2}\Gamma(n/2)}$. Thus $j > C_1 d^n |\Lambda^{(n)}| \implies E_j(-\Delta_{\Lambda}^{\mathrm{N},\dec}) > S - nq_-$. But by $(\star)$, $E_j(H_{\Lambda}^{(n)}(\omega)) \ge E_j(-\Delta_{\Lambda}^{\mathrm{N},\dec}) + nq_-$. Thus $j > d^n C_1 |\Lambda^{(n)}| \implies E_j(H_{\Lambda}^{(n)}(\omega)) > S$. We get \textsf{(WEYL.$n$)} with $C=d^n C_1$ and $C \le \lfloor \frac{d^n(S-nq_-)^{n/2}}{(4\pi)^{n/2}\Gamma(n/2)} \rfloor +1$ if $S>S^{\ast}(n,q_-)$.
\end{proof}

Before proceeding further, we need the following notion.

\begin{defa}         \label{def:R-con}
Given $\mathbf{y} \in \Z^{nd}$ and $\emptyset \neq \mathcal{J} \subseteq \{1,\ldots,n\}$, we say that $\mathcal{P} = \{ y_j : j \in \mathcal{J} \}$ is \emph{$R$-connected} if $\mathcal{Z}=\bigcup_{j \in \mathcal{J}} \Lambda_R^{(1)}(y_j) \subset \R^{d}$ is connected. In this case, it is easily shown by induction on $\# \mathcal{J}$ that if $\# \mathcal{J} \ge 2$, we have
\[ 
\forall i, j \in \mathcal{J}: |y_i - y_j | < ( \# \mathcal{J} -1 ) (2R) \le 2(n-1)R \, . 
\]
\end{defa}

\begin{lem}
A partially interactive cube is decomposable.
\end{lem}
\begin{proof}
Suppose $\Lambda_L^{(n)}(\mathbf{u})$ is not decomposable. Then $\exists \,i \neq 1$ such that $|u_1 - u_{i}| < 2L + r_0$ (otherwise $\mathcal{J} = \{1\}$ would give a possible partition). Let $\mathcal{J}_2 = \{1, i\}$. Since $\mathcal{J}_2$ is not a possible partition, $\exists \,i_2 \notin \mathcal{J}_2$ such that $|u_1 - u_{i_2}| < 2L+r_0$ or $|u_{i} - u_{i_2} | < 2L+r_0$. Taking $\mathcal{J}_3 = \{1,i,i_2\}$, the set $\{u_k:k \in \mathcal{J}_3\}$ is thus $(L+r_0/2)$-connected. As $\mathcal{J}_3$ is not a possible partition, we may repeat the procedure and finally obtain $\mathcal{J}_n = \{1,\ldots,n\}$ and $\{u_k:k \in \mathcal{J}_n\}$ is $(L+r_0/2)$-connected. Consequently,
\[ 
\forall 1 \le j \le n: |u_j - u_1 | < (n-1)(2L+r_0) 
\]
Hence 
\[ 
\dist(\mathbf{u},\D) \le |\mathbf{u} - (u_1,\ldots,u_1) | = \max_{1\le j\le n} |u_j - u_1| < (n-1)(2L+r_0)
\]
The lemma now results by contraposition.
\end{proof}

\begin{lem}
Let $\mathbf{x}, \mathbf{y} \in \Z^{nd}$, $L \ge 1$ and take $r_{n,L}$ as in Definition~\ref{def:sep}. Then
\begin{enumerate}[\rm 1)]
\item If $\mathbf{y} \notin \bigcup_{j=1}^{K(n)} \Lambda_{2nL}^{(n)}(\mathbf{x}^{(j)})$, then $\Lambda_L^{(n)}(\mathbf{y})$ and $\Lambda_L^{(n)}(\mathbf{x})$ are pre-separable.
\item If $\mathbf{y} \notin \bigcup_{j=1}^{K(n)} \Lambda_{r_{n,L}}^{(n)}(\mathbf{x}^{(j)})$, then $\Lambda_L^{(n)}(\mathbf{y})$ and $\Lambda_L^{(n)}(\mathbf{x})$ are separable.
\item If $\mathbf{y} \notin \Lambda_{2r_{n,L}}^{(n)}(\mathbf{0})$, then $\Lambda_L^{(n)}(\mathbf{y})$ is separable from any $\Lambda_L^{(n)}(\mathbf{x})$ satisfying $\mathbf{x} \in \Lambda_{r_{n,L}}^{(n)}(\mathbf{0})$.
\end{enumerate}
\end{lem}
\begin{proof}
\begin{enumerate}[1)]
\item Decompose $\{ y_1,\ldots,y_n\}$ into maximal $L$-connected subsets 
\[
\mathcal{P}_k = \{ y_j : j \in \mathcal{J}_k \},\quad k=1,\ldots,m,
\]
and let $\mathcal{Z}_k = \bigcup_{j \in \mathcal{J}_k} \Lambda_L^{(1)}(y_j)$. Then $(\mathcal{Z}_k)_k$ forms a partition of $\Pi \Lambda_L^{(n)} (\mathbf{y})$. Suppose now that $\Lambda_L^{(n)}(\mathbf{x})$ and $\Lambda_L^{(n)}(\mathbf{y})$ are not pre-separable. Then 
\[ 
\forall \, \emptyset \neq \mathcal{J} \subseteq \{1, \ldots , n \}: \Pi_{\mathcal{J}} \Lambda_L^{(n)}(\mathbf{y}) \cap \big( \Pi_{\mathcal{J}^c} \Lambda_L^{(n)}(\mathbf{y}) \cup \Pi \Lambda_L^{(n)}(\mathbf{x}) \big) \neq \emptyset 
\]
Since $(\mathcal{Z}_k)_k$ forms a partition of $\Pi \Lambda_L^{(n)} (\mathbf{y})$, we have in particular
\[
\forall \, 1 \le k \le m : \mathcal{Z}_k \cap \Pi \Lambda_L^{(n)}(\mathbf{x}) \neq \emptyset,
\]
hence
\[
\forall \, 1 \le k \le m, \exists \, y_j \in \mathcal{P}_k, \exists \, x_i : |y_j - x_i | < 2L.
\]
But $\mathcal{P}_k$ are $L$-connected, hence $\forall i, j \in \mathcal{J}_k : |y_i - y_j | < 2(n-1)L$. Thus,
\[
\forall \, y_j \ \exists \, x_i : |y_j - x_i | < 2nL,
\]
so that $\mathbf{y} \in \Lambda_{2nL}^{(n)}(\mathbf{x}^{(k)})$ for some $k$. The claim follows by contraposition. 
\item This follows from 1) by noting that $r_{n,L} \ge 2nL$ and that $|\mathbf{y} - \mathbf{x}^{(j)} | \ge r_{n,L}$ for all $j$ implies $|\mathbf{y}-\mathbf{x}| \ge r_{n,L}$ (since $\mathbf{x}$ is one of the $\mathbf{x}^{(j)}$).
\item Let 
\[ 
F = \bigcup_{\mathbf{x} \in \Lambda_{r_{n,L}}^{(n)}(\mathbf{0})} \bigcup_{j=1}^{K(n)} \Lambda_{r_{n,L}}^{(n)}(\mathbf{x}^{(j)}).
\] 
Then by 2), if $\mathbf{y} \notin F$, then $\Lambda_L^{(n)}(\mathbf{y})$ is separable from any $\Lambda_L^{(n)}(\mathbf{x})$ with $\mathbf{x} \in \Lambda_{r_{n,L}}^{(n)}(\mathbf{0})$. Thus it suffices to show that $\Lambda_{2r_{n,L}}^{(n)}(\mathbf{0}) = F$. For this note that if $\mathbf{x} \in \Lambda_{r_{n,L}}^{(n)}(\mathbf{0})$, then $| x_k | < r_{n,L}$ for all $k$, so by definition of $\mathbf{x}^{(j)}$, $| \mathbf{x}^{(j)} | < r_{n,L}$ for all $j$ and so $\mathbf{x}^{(j)} \in \Lambda_{r_{n,L}}^{(n)}(\mathbf{0})$ for all $j$. Thus $F = \bigcup_{\mathbf{x} \in \Lambda_{r_{n,L}}^{(n)}(\mathbf{0})} \Lambda_{r_{n,L}}^{(n)}(\mathbf{x}) = \Lambda_{2r_{n,L}}^{(n)}(\mathbf{0})$.
\end{enumerate}
\end{proof}

\begin{lem}
Separable FI cubes are completely separated.
\end{lem}
\begin{proof}
Let $\Lambda_L^{(n)}(\mathbf{u})$ and $\Lambda_L^{(n)}(\mathbf{v})$ be FI. Then there exists $\mathbf{x}, \mathbf{y} \in \D$ such that $|\mathbf{u} - \mathbf{x} | < (n-1)(2L+r_0)$ and $|\mathbf{v} - \mathbf{y} | < (n-1)(2L+r_0)$. Hence for all $j,k=1,\ldots,n:$
\[
\Pi_j \Lambda_L^{(n)}(\mathbf{u}) \subseteq \Pi_j \Lambda_{(n-1)(2L+r_0)+L}^{(n)}(\mathbf{x}) \quad \text{and} \quad \Pi_k \Lambda_L^{(n)}(\mathbf{v}) \subseteq \Pi_k \Lambda_{(n-1)(2L+r_0)+L}^{(n)}(\mathbf{y}). \tag{$\star$}
\]
Now
\[
| \mathbf{u}-\mathbf{v} | \le | \mathbf{u} - \mathbf{x} | + | \mathbf{x} - \mathbf{y} | + | \mathbf{y} -\mathbf{v} | < 2(n-1)(2L+r_0) + | \mathbf{x} - \mathbf{y} |.
\]
Moreover, $\Lambda_L^{(n)}(\mathbf{u})$ and $\Lambda_L^{(n)}(\mathbf{v})$ are separable, so by definition $| \mathbf{u} - \mathbf{v} | \ge r_{n,L} = 4(n-1)(2L+r_0)+2L$. We thus get
\[
| \mathbf{x} - \mathbf{y} | > | \mathbf{u} - \mathbf{v} | - 2(n-1)(2L+r_0) \ge 2(n-1)(2L+r_0)+2L.
\]
Since $\mathbf{x},\mathbf{y} \in \D$, this implies
\[
\Pi_j \Lambda_{(n-1)(2L+r_0)+L}^{(n)}(\mathbf{x}) \cap \Pi_k \Lambda_{(n-1)(2L+r_0)+L}^{(n)}(\mathbf{y}) = \emptyset \tag{$\star \star$}
\] 
for all $j,k = 1,\ldots,n$. By ($\star$) and ($\star \star$), we see that $\Pi_j \Lambda_L^{(n)}(\mathbf{u}) \cap \Pi_k \Lambda_L^{(n)}(\mathbf{v}) = \emptyset$ for all $j,k$. Hence $\Pi \Lambda_L^{(n)}(\mathbf{u}) \cap \Pi \Lambda_L^{(n)}(\mathbf{v}) = \emptyset$, as asserted.
\end{proof}

\subsection*{Acknowledgements}
This work is part of my PhD thesis, supervised by Professors Anne Boutet de Monvel and Victor Chulaevsky. I would like to express here my gratitude for their patience and their numerous suggestions, which led to significant improvements to the results presented here.

\providecommand{\bysame}{\leavevmode\hbox to3em{\hrulefill}\thinspace}
\providecommand{\MR}{\relax\ifhmode\unskip\space\fi MR }
\providecommand{\MRhref}[2]{%
  \href{http://www.ams.org/mathscinet-getitem?mr=#1}{#2}
}
\providecommand{\href}[2]{#2}

\end{document}